\documentclass{article}

\usepackage[preprint]{icml2026}

\usepackage{microtype}
\usepackage{booktabs, makecell, multirow}
\usepackage{array, tabularx, graphicx}
\usepackage{float}
\usepackage{caption}
\usepackage{subcaption}
\usepackage{xstring}
\usepackage{expl3}
\usepackage[dvipsnames]{xcolor}
\usepackage{colortbl}
\newcommand{\graycline}[1]{
  \arrayrulecolor{gray} 
  \cline{#1} 
  \arrayrulecolor{black} 
}
\newcommand{\graymidrule}{
  \arrayrulecolor{gray} 
  \specialrule{0.2pt}{0.3pt}{0.3pt} 
  \arrayrulecolor{black} 
}
\ExplSyntaxOn
\NewDocumentCommand{\signedparenbox}{m m}
 {
  \fp_compare:nTF { #2 > 0 }
    { \textcolor{teal}{(\makebox[#1][r]{#2})} }
    {
      \fp_compare:nTF { #2 < 0 }
        { \textcolor{red}{(\makebox[#1][r]{#2})} }
        { (\makebox[#1][r]{#2}) } 
    }
 }
\NewDocumentCommand{\signedparen}{m}
 {
  \fp_compare:nTF { #1 > 0 }
    { \textcolor{teal}{(#1)} }
    {
      \fp_compare:nTF { #1 < 0 }
        { \textcolor{red}{(#1)} }
        { (#1) }
    }
 }

\NewDocumentCommand{\tworightcols}{m m}
 {
  \begin{tabular}[t]{@{}r@{\hspace{0.6em}}r@{}}
    #1 & #2
  \end{tabular}%
 }

\ExplSyntaxOff
\newcommand{\aaccsign}[2]{%
  \tworightcols{#1}{\signedparen{#2}}%
}
\newcommand{\ttimesign}[2]{%
  \tworightcols{#1}{\signedparen{#2}}%
}

\usepackage{wrapfig}
\usepackage{rotating}
\usepackage{pdflscape}
\usepackage{changepage}
\usepackage{enumitem}

\usepackage{amsmath,amssymb,amsthm,mathtools} 

\usepackage{amsfonts,bm}

\def\eqref#1{equation~\ref{#1}}

\def\1{\bm{1}}

\DeclareMathAlphabet{\mathsfit}{\encodingdefault}{\sfdefault}{m}{sl}
\SetMathAlphabet{\mathsfit}{bold}{\encodingdefault}{\sfdefault}{bx}{n}

\def\gE{{\mathcal{E}}}

\def\gG{{\mathcal{G}}}

\def\gP{{\mathcal{P}}}

\def\gV{{\mathcal{V}}}

\theoremstyle{plain}
\newtheorem{theorem}{Theorem}[section]

\newtheorem{lemma}[theorem]{Lemma}
\newtheorem{corollary}[theorem]{Corollary}
\theoremstyle{definition}
\newtheorem{definition}[theorem]{Definition}
\newtheorem{assumption}[theorem]{Assumption}
\theoremstyle{remark}

\usepackage[ruled,vlined,linesnumbered, algo2e]{algorithm2e}
\SetKw{KwTo}{to}
\SetKw{KwContinue}{continue}
\SetKwProg{Fn}{Function}{}{}

\usepackage{hyperref}
\usepackage[capitalize,noabbrev]{cleveref}

\DeclareRobustCommand{\TimelyFreeze}{\text{\text{TimelyFreeze}}}
\newcommand{\TimelyAPF}{\TimelyFreeze{}{\scriptsize+APF}}
\newcommand{\TimelyAuto}{\TimelyFreeze{}{\scriptsize+Auto}}

\icmltitlerunning{TimelyFreeze: Adaptive Parameter Freezing Mechanism for Pipeline Parallelism}

\begin{document}

\twocolumn[
  \icmltitle{\TimelyFreeze{}: \\Adaptive Parameter Freezing Mechanism for Pipeline Parallelism}



  \icmlsetsymbol{equal}{*}

  \begin{icmlauthorlist}
    \icmlauthor{Seonghye Cho}{DM Lab}
    \icmlauthor{Jaemin Han}{DM Lab}
    \icmlauthor{Hyunjin Kim}{DM Lab}
    \icmlauthor{Euisoo Jung}{DM Lab}
    \icmlauthor{Jae-Gil Lee}{DM Lab}
  \end{icmlauthorlist}

  \icmlaffiliation{DM Lab}{KAIST, Daejeon, South Korea}
  
  \icmlcorrespondingauthor{Jae-Gil Lee}{jaegil@kaist.ac.kr}

  \icmlkeywords{Distributed Learning, Pipeline Parallelism, Parameter Freezing}

  \vskip 0.3in
]



\printAffiliationsAndNotice{}  

\begin{abstract}


    Pipeline parallelism enables training models that exceed single-device memory, but practical throughput remains limited by pipeline bubbles.
    Although parameter freezing can improve training throughput by adaptively skipping backward computation, existing methods often \emph{over}-freeze parameters, resulting in unnecessary accuracy degradation. 
    To address this issue, we propose \textbf{\textit{TimelyFreeze}}, which models the pipeline schedule as a directed acyclic graph and solves a linear program to compute \emph{optimal} freeze ratios that minimize batch execution time under accuracy constraints. 
    Experiments show that TimelyFreeze achieves up to 40\% training throughput improvement on LLaMA-8B with comparable accuracy. 
    Overall, it enables faster large-scale model training without compromising convergence and generalizes across diverse pipeline-parallel settings.

\end{abstract}

\vspace*{-0.6cm}
\section{Introduction}
\label{sec:intro}
Driven by scaling laws~\citep{kaplan2020_scalinglaw}, deep learning models and datasets have grown beyond the capacity of a single GPU, making distributed learning indispensable for large-scale training.
Among parallelization strategies, including data parallelism (DP)~\citep{you2018imagenet_dp, sergeev2018horovod_dp}, tensor parallelism (TP)~\citep{shoeybi2019megatron}, and pipeline parallelism (PP)~\citep{huang2019gpipe}, \emph{pipeline parallelism} has emerged as a crucial technique for scaling model training. Unlike DP, which replicates the model across devices, PP splits the model into sequential stages, enabling memory-efficient scaling.
 Compared to TP, PP incurs lower communication overhead and is better suited to environments with limited inter-node bandwidth. 


However, PP suffers from \emph{training inefficiency} due to the inherently sequential nature of forward and backward computations. Each partitioned model stage must wait for preceding stages to complete their computations, creating idle periods of computational resources (e.g., GPUs), commonly known as \emph{pipeline bubbles}. To mitigate this inefficiency, various scheduling algorithms, such as GPipe~\citep{huang2019gpipe} or 1F1B~\citep{fan2021dapple1f1b, narayanan2019pipedream1f1b} have been proposed to improve hardware utilization. Nevertheless, pipeline bubbles persist, often leading to lower training throughput than data parallelism. 


To improve training throughput, \emph{parameter freezing}~\citep{adaptivefreezeapf2024, freezeout2017} has emerged as an effective approach, particularly in parameter-efficient fine-tuning and federated learning.
%
By adaptively skipping backward computations and gradient updates for a subset of parameters, it reduces computational cost and improves training throughput.  
%
Nonetheless, parameter freezing may degrade accuracy due to skipped gradient updates in exchange for higher throughput. Moreover, prior methods~\citep{adaptivefreezeapf2024, autofreeze2021} do not account for pipeline-parallel training dynamics and thus tend to apply excessive freezing, which unnecessarily degrades the model accuracy. 
%
For example, the gray region labeled ``Ineffective Freezing'' in \autoref{fig:overview}(b) illustrates a case where GPU~1--3 freeze parameters during the backward pass even though subsequent execution cannot yet begin due to schedule dependencies. As a result, such excessive freezing under PP neither improves throughput nor preserves training stability. 
Therefore, achieving an optimal balance between accuracy and throughput while accounting for pipeline-parallel training dynamics is crucial.

\begin{figure}[t]
    \centering
    \includegraphics[width=\linewidth]{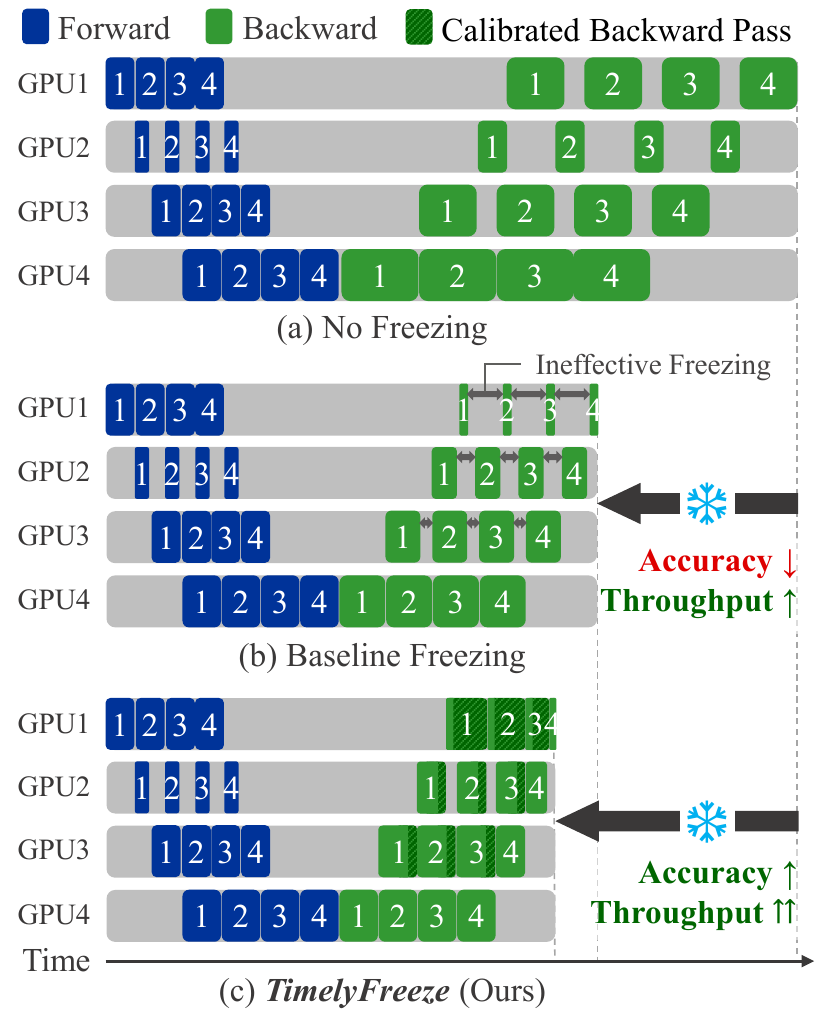}
    \caption{Effect of parameter freezing on pipeline parallelism. Each block denotes a microbatch execution in the forward (blue) or backward (green) pass. TimelyFreeze in (c) maintains the throughput gain of the baseline freezing method in (b) while preserving accuracy by effectively freezing according to the pipeline schedule.}
    \label{fig:overview}
\end{figure}



In this paper, we propose \textbf{\textit{TimelyFreeze}}, a pipeline-aware parameter freezing method that balances accuracy and throughput in pipeline-parallel training. 
As indicated by the ``Calibrated Backward Pass'' region in \autoref{fig:overview}(c), it controls freezing by considering idle time in later execution block, or even increases freezing in the computational bottlenecks to balance computation across GPUs. 
By explicitly modeling pipeline schedules as directed acyclic graphs and optimizing freeze ratios accordingly, TimelyFreeze preserves accuracy more effectively while achieving better throughput.

We validate our method on the LLaMA series (1B/8B/13B) across pipeline schedules.
TimelyFreeze lies on or near the Pareto frontier along two key axes---accuracy and throughput---achieving up to 46\% throughput improvement with comparable or higher accuracy than baselines.
Moreover, TimelyFreeze generalizes to vision architectures (ViT-L/32 and ConvNeXtV2-L), reducing training time by up to 25\% while keeping accuracy drops within 1.5 percentage points relative to the no-freezing baseline. 
%

Collectively, these results demonstrate that TimelyFreeze robustly accelerates training across diverse pipeline-parallel settings, spanning different model architectures, scales, model partitioning heuristics, and hardware configurations.
This work entails substantial, in-depth \emph{system-oriented} research that directly impacts model training performance.


%
%
Our key contributions are summarized as follows:
\begin{itemize}[leftmargin=10pt, nosep]
\item \textbf{Pipeline-aware freezing}: We reduce unnecessary freezing by leveraging runtime GPU idle time that arises from execution dependencies in pipeline schedules.
\item \textbf{DAG-based formulation}: We represent pipeline schedules as directed acyclic graphs and derive expected freeze ratios by solving a linear program.
\item \textbf{Time-to-accuracy analysis}: We provide a theoretical analysis of time-to-accuracy, characterizing the trade-off between throughput and convergence. 
\end{itemize}

\section{Preliminary and Related Work}
\label{sec:related_work}


\subsection{Pipeline Parallelism}
\label{subsec:pipeline_parallelism}

Pipeline parallelism partitions a model across multiple devices, enabling training of large-scale models that exceed the memory capacity of a single GPU. By splitting computation into sequential stages, it alleviates memory constraints and accelerates training, and various scheduling strategies have been proposed to maximize throughput.

\noindent
\textbf{GPipe.} \emph{GPipe}~\citep{huang2019gpipe} divides a batch into smaller microbatches, enabling different pipeline stages to be executed concurrently across multiple GPUs. 
However, GPipe still suffers from notable underutilization due to pipeline bubbles 
and incurs an activation memory overhead that grows linearly with the number of microbatches.

\noindent
\textbf{1F1B.}
\emph{1F1B} (1-Forward 1-Backward)~\citep{narayanan2019pipedream1f1b,fan2021dapple1f1b} interleaves forward and backward computations. By performing a backward pass immediately after the forward pass, 1F1B limits the accumulation of activations and reduces memory overhead. 
Nevertheless, 1F1B schedule still suffers from low GPU utilization similarly to GPipe. Subsequently, to further improve utilization, \emph{Interleaved 1F1B}~\citep{narayanan2021interleaved1f1b} divides each stage into multiple micro-stages.

\noindent
\textbf{Zero-Bubble.}
Recent work such as \emph{Zero-Bubble}~\citep{qi2023zerobubble, qi2024pipeline_controllable_memory} decomposes backward computations into finer-grained units to maximize GPU utilization, approaching near 100\% utilization. However, this approach introduces additional communication and synchronization overhead, which can limit practical scalability, especially in heterogeneous or bandwidth-constrained environments.


\subsection{Parameter Freezing}
\label{subsec:parameter_freezing}

\emph{Parameter freezing} is a training acceleration technique that selectively skips backward computations for a subset of model parameters to reduce training time.
Prior work has focused on designing stability metrics to determine which parameters to freeze and when, aiming to minimize accuracy loss while maximizing computational efficiency.

\textbf{Monotonic Freezing from Front Layers.}
%
Early approaches~\citep{freezeout2017,lee2019would_elsa_freezing} rely on the empirical observation that \emph{front layers tend to converge earlier than deeper layers}, and thus employ monotonic prefix freezing that freezes layers sequentially from the front to the back. \citet{freezeout2017} progressively freeze shallow layers using predefined schedules,
reducing computational cost in CNN models. 
\citet{autofreeze2021} and \citet{egeria2023_knowledge_guided_freezing} calculate layer-level stability metrics to determine the depth of the frozen prefix. 
Both methods also cache intermediate activations for the frozen layers, skipping not only backward passes but also redundant forward passes. 

\textbf{Non-Monotonic Freezing.}
More recent work~\citep{li2024smartfrz_freezing} questions the depth-ordered convergence assumption and proposes \emph{non-monotonic} strategies. \citet{li2024smartfrz_freezing} observe that, due to residual connections, later layers may stabilize earlier than some intermediate layers. 
\citet{adaptivefreezeapf2024} further increase granularity by performing parameter-wise freezing in federated learning, using gradient-based stability to reduce communication between clients and the central server. 
These methods relax the strict prefix constraint and thus offer greater flexibility.


However, prior freezing methods generally do \emph{not} account for the characteristics of pipeline parallelism. 
As illustrated in \autoref{fig:overview}(b), when applied to pipeline-parallel settings, such pipeline-unaware designs can lead to unnecessary freezing that is misaligned with the pipeline execution timeline. 
PipeTransformer~\citep{he2021pipetransformer_freezing} partly addresses this issue by introducing dynamic load balancing. However, it permanently freezes parameters once selected, which introduces a non-negligible risk of accuracy degradation.

\subsection{Baseline Freezing Methods}
\label{sec:baseline}
We select AutoFreeze~\citep{autofreeze2021} and APF~\citep{adaptivefreezeapf2024} as our baselines, respectively representing the monotonic and non-monotonic freezing methods. Both methods are among the most widely cited in this area.


\textbf{AutoFreeze.}
AutoFreeze estimates layer stability via the \emph{gradient-norm change} for each layer and freezes layers in prefix order under the assumption that earlier layers converge faster. The gradient norm change for each layer at the $K$-th stability check is defined as 
\begin{align}
\text{Score}_{\text{Auto},K} = \frac{\left| \left\| \Delta_{K - 1} \right\| - \left\| \Delta_{K} \right\| \right|}{\left\| \Delta_{K - 1} \right\|},
\end{align}
where $\Delta_K$ denotes the cumulative parameter update since the $(K-1)$-th check.
A layer is frozen when (i) all preceding layers are already frozen and (ii) its gradient-norm change falls into the lower $P_{\text{Auto}}$-th percentile among all layers, controlled by a hyperparameter $P_{\text{Auto}}$. 


\textbf{APF.}
APF identifies parameters whose gradient updates oscillate without a clear trend as sufficiently stabilized and freezes them to reduce computation.
To quantify such oscillatory behavior, APF computes the \emph{effective perturbation score} for each parameter at the $K$-th stability check as
\begin{align}
\begin{alignedat}{2}
\text{Score}_{\text{APF},K}
&= \frac{|E_K|}{E_K^{\text{abs}}}, \;\;
&\left\{
\begin{aligned}
E_K
&= \alpha E_{K-1} + (1-\alpha)\Delta_K, \\
E_K^{\text{abs}}
&= \alpha E_{K-1}^{\text{abs}} + (1-\alpha)|\Delta_K|.
\end{aligned}
\right.
\end{alignedat}
\end{align}
Parameters whose 
score falls below the threshold $T_{\text{APF}}$
are considered sufficiently stable
and frozen.
\section{Proposed Framework: TimelyFreeze}
\label{sec:method}

\begin{figure*}[t]
    \centering
    \includegraphics[width=\textwidth]{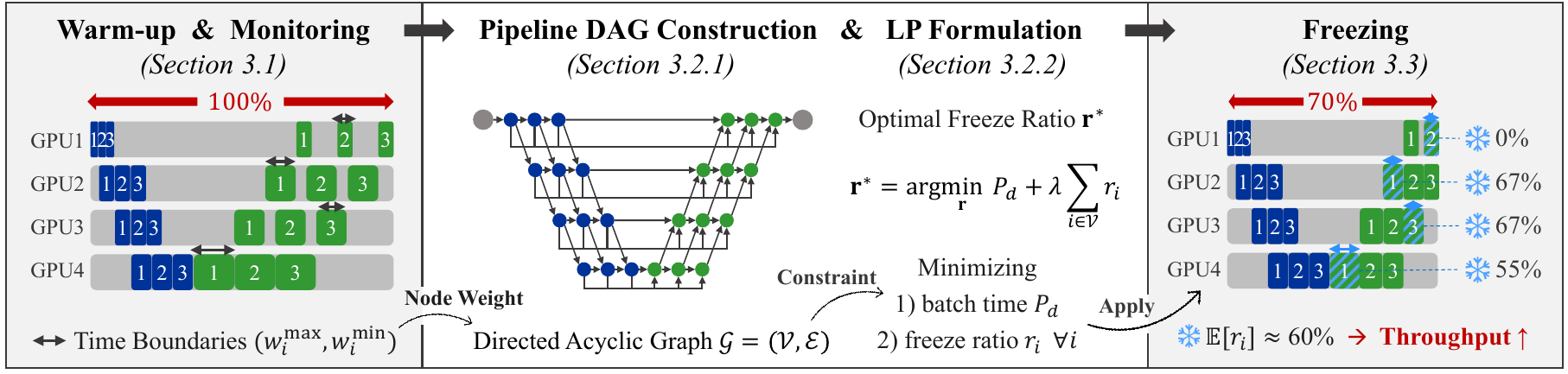}
    \caption{
    Method overview of TimelyFreeze. 
    The system first monitors the upper and lower bounds of execution times and represents the pipeline schedule as a directed acyclic graph (DAG), where nodes and edges respectively denote action blocks and execution dependencies.
    Based on the DAG constraints, it solves a linear program to compute the expected freeze ratios $\mathbf{r}^*$.
    Finally, it freezes each action block at its corresponding freeze ratio to improve training throughput while preserving accuracy.
    }
    \label{fig:method_overview}
\end{figure*}



TimelyFreeze targets time efficiency and model performance preservation in practical pipeline parallel settings. To handle dynamic real world environment, we determine \emph{freeze ratios}---how many parameters to freeze---in a fine-grained manner, per microbatch and per stage in the pipeline schedule. 
Through the overall training process, our method applies step-level parameter freezing through three phases:
\emph{(1) Warm-up and Monitoring} (\autoref{subsec:monitoring}),
\emph{(2) Freeze Ratio Formulation} (\autoref{subsec:freeze_ratio_formulation}), and
\emph{(3) Freezing} (\autoref{subsec:freezing_phase}).

{\bf (1)} During the monitoring phase, the system measures the GPU execution time for each forward and backward action, where each action corresponds to a unit of microbatch execution at a pipeline stage.
{\bf (2)} Based on these runtime measurements, TimelyFreeze constructs a graph representation of the pipeline schedule, which is then used to formulate a linear program that computes the expected freeze ratios $\mathbf{r}^*$. {\bf (3)} In the subsequent freezing phase, each backward action $v_i$ starts freezing its corresponding parameters at $r_i$. 

\textbf{Example.}
As an illustrative example, the right gray box in \autoref{fig:method_overview}
shows a case where the batch time is reduced to about 70\% of the original
(left gray box), by freezing backward actions according to the expected freeze
ratios computed in the previous phase (white box).
%
In this example, each backward action has its own expected freeze ratio, with an
average expected freeze ratio of 60\%.


\textbf{Notation.}
To be clear, let $t$ denote the training step and $\{T_w, T_m, T_f\}$ be the last step of the warm-up phase, the monitoring phase, and the progressive freezing phase, respectively. The notation is summarized in Appendix~\ref{appendix:notations}.

\subsection{Phase I: Warm-up and Monitoring}
\label{subsec:monitoring}

We employ an explicit \emph{warm-up} phase at the beginning of training, as
monitoring is only meaningful after CUDA kernel execution has stabilized.
Accordingly, both monitoring and parameter freezing are disabled for the first
$T_w$ steps.
We align $T_w$ with the learning-rate warm-up steps, since the parameter updates
during learning-rate ramp-up are highly unstable and even mild freezing can
suppress critical early updates, destabilizing training.
Freezing is therefore activated only after the learning-rate warm-up completes.



Once warm-up ends ($t > T_w$), the system starts recording the execution time of 
every forward and backward action on each GPU 
in a two-part monitoring phase.
During the first half, no parameters are frozen so that each action's maximum backward duration can be measured (i.e., \emph{upper-bound monitoring}). In the next half, all parameters are frozen to capture the minimum achievable latency (i.e., \emph{lower-bound monitoring}). 
After samples for both bounds are collected (at step $t = T_m$), the freeze ratios $\mathbf{r}$ are computed by solving the linear program described in \autoref{subsec:freeze_ratio_formulation}.

\subsection{Phase II: Freeze Ratio Formulation}
\label{subsec:freeze_ratio_formulation}
At the end of the monitoring phase ($t=T_m$), TimelyFreeze formulates a linear program (LP) based on the GPU execution time measurements of actions. The objective of this LP is to compute the \emph{expected freeze ratio} for each action. 

To ensure broad applicability across pipeline schedules and training settings, we represent the pipeline execution schedule of a batch as a directed acyclic graph (DAG), as described in \autoref{subsec:pipeline_dag_construction}.
Building on this representation, we formulate a linear program that selects freeze ratios 
for actions to maximize training throughput while minimizing unnecessary parameter freezing, as detailed in \autoref{subsec:lp_formulation}.

\subsubsection{Pipeline DAG Construction}
\label{subsec:pipeline_dag_construction}

\textbf{Pipeline DAG Nodes.}
\label{def:pipeline_dag_nodes}
Let $\gG=(\gV,\gE)$ be a directed acyclic graph (DAG) representing the pipeline schedule of a batch. 
The \emph{node} set $\gV$ primarily consists of action nodes $v_{(a,m,s)}$, each corresponding to an action $a$ of a microbatch $m$ at a stage $s$, where the action is defined as a combination of
\[
a \in \{f,b\}, \quad
m \in \{1,\dots,M\}, \quad
s \in \{1,\dots,S\}.
\]
Here, $a=f$ and $a=b$ denote forward and backward computations, respectively, while $M$ and $S$ denote the total numbers of microbatches and pipeline stages.
In addition, we introduce a source node $v_s$ and a destination node $v_d$, which represent the start and the end of a batch, respectively.

\textbf{Pipeline DAG Edges.}
\label{def:pipeline_dag_edges}
Let $\gE$ denote the set of directed edges in the DAG $\gG=(\gV,\gE)$. 
Each edge $v_i \to v_j \in \gE$ represents an execution dependency, indicating that a node $v_j$ cannot begin until a node $v_i$ finishes. 
$\gE$ encodes all necessary intra- and inter-stage execution constraints of the pipeline. The edge construction rules are provided in \autoref{appendix:edge_construction_rules}.

\textbf{Node Weights.} 
Each action node $v_i \in \mathcal{V}$ is associated with a \emph{weight} $w_i$, denoting the execution duration of the action,
\begin{equation}
w_i \in [\, w_i^{\min}, \, w_i^{\max} \,],
\label{def:action_duration_node_weight}
\end{equation}
where $w_i^{\min}$ and $w_i^{\max}$ denote the execution duration when all parameters are frozen and when no parameters are frozen, respectively, as obtained from the monitoring phase.
Additionally, for the source and destination nodes, we define $w_s = w_d = 0$, since they are abstract nodes that are not associated with any computation.

\begin{figure}[t]
    \centering
    \includegraphics[width=0.95\linewidth]{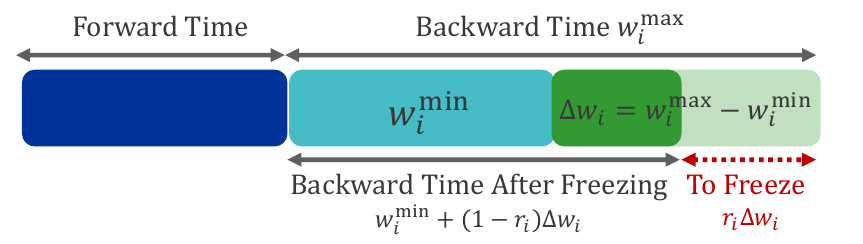}
    \caption{Backward-time reduction via parameter freezing.}
    \label{fig:fwdbwd_ratio}
\end{figure}

\autoref{fig:fwdbwd_ratio} illustrates how the node weight $w_i$ varies with the freeze ratio.
The execution time of forward actions (blue) remains unchanged, as forward computations are not affected by parameter freezing. Accordingly, for any forward nodes $i=(f,m,s)$, $w_i = w_i^{\min} = w_i^{\max}$ holds.
In contrast, backward actions are affected by parameter freezing.
Following the decomposition of backward computation in ZeroBubble analysis~\citep{qi2023zerobubble}, the backward pass consists of two components: 
(1) gradient computation with respect to input activations (in {\color{cyan}cyan} in \autoref{fig:fwdbwd_ratio}) and 
(2) gradient computation with respect to parameters (in {\color{OliveGreen}green}).
Only the latter is influenced by parameter freezing.
As a result, even when parameters are fully frozen, the activation-gradient component remains at $w_i^{\min}$, while the parameter-gradient component decreases according to the freeze ratio.


\textbf{Freeze Ratio.}
\label{def:freeze_ratio}
The \emph{freeze ratio} of a node $v_i$ is defined as
\begin{equation}
    r_i \;:=\; 1 - \dfrac{w_i - w_i^{\min}}{\,w_i^{\max} - w_i^{\min}}, 
    \qquad r_i \in [0,1]
    \label{eq:def_freeze_ratio}
\end{equation}
when $w_i^{\max} \!>\! w_i^{\min}$, and we set $r_i \!=\! 0$ when $w_i^{\max} \!=\! w_i^{\min}$.

Here, $r_i$ quantifies the average fraction of parameters frozen in a node $v_i$. In particular, for forward actions $v_{(f,m,s)}$ whose execution time does not vary with freezing,
\begin{displaymath}
w_{(f,m,s)} = w_{(f,m,s)}^{\min} = w_{(f,m,s)}^{\max} \;\; \implies \;\; r_{(f,m,s)} = 0.
\end{displaymath}

\textbf{Start Time of Action.}
\label{def:start_time}
The \emph{start time} $P_i$ of an action node $v_i$ is determined by the longest-path over all its predecessors in the DAG such that
\begin{equation}
P_i := \max_{(v_j \to v_i) \in \mathcal{E}} \big(P_j + w_j\big),
\end{equation}
where the batch start time according to the source node $v_s$ is defined as $P_s = 0$, and the batch execution time corresponds to the start time of the destination node, $P_d$.

\subsubsection{Linear Programming Formulation}
\label{subsec:lp_formulation}

We formulate the problem of determining the freeze ratios $\mathbf{r}$ as a linear program (LP) on the previously defined pipeline DAG.
Applying the resulting expected freeze ratios $\mathbf{r}^*$ selectively shortens backward computations that are related to the critical path of the end-to-end batch execution time.
Consequently, the overall batch latency is reduced from $P_d^{\max}$ (i.e., $P_d$ where all $w_i = w_i^{\max}$) to $P_d^*$.

\textbf{Decision Variables.}
For each node $v_i \in \mathcal{V}$, including the source and destination nodes $v_s$ and $v_d$, we introduce two non-negative real-valued decision variables:
$P_i \in \mathbb{R}_{\ge 0}$ and $w_i \in \mathbb{R}_{\ge 0}$, denoting the start time and duration of a node $v_i$, respectively. Together, these variables define the temporal placement of each action block within the pipeline schedule.

\textbf{Objective Function.}
The objective is twofold: 
(1) to minimize the batch completion time $P_d$ as the \emph{primary objective},
and (2) to avoid excessive freezing by minimizing the freeze ratios $r_i$ as a \emph{secondary objective}, subject to an upper-bound $r_{\text{max}}$.
We enforce this priority by assigning a sufficiently small $\lambda \ll 1$ to the secondary goal, so that minimizing $P_d$ always dominates the optimization.
Accordingly, we formulate the linear objective as
\begin{equation}
\min \;\Big[\, P_d - \lambda \sum_{i \in \mathcal{V}} \delta_i w_i \,\Big], 
\qquad \lambda \ll 1.
\end{equation}
The second term therefore acts only as a tie-breaker among solutions with similar batch time $P_d$. 
Here, $\delta_i$ denotes the reciprocal of the execution-time range, defined as 
\begin{equation}
    \delta_i = 
\frac{1}{w_i^{\max} - w_i^{\min}} \quad\text{if } w_i^{\max} > w_i^{\min}, \text{  and }0 \text{ otherwise }.
\label{eq:delta_i}
\end{equation}
With this definition, the term $-\delta_i w_i$ is linearly proportional to the freeze ratio $r_i$ for freezable nodes,
whereas $\delta_i = 0$ for nodes whose execution time is unaffected by freezing.

\textbf{Constraints.}
The optimization is subject to the following linear constraints. 

\newpage

\vspace*{-1.3cm}
\begin{align*} 
    &\text{[1]}\quad P_j \ge P_i + w_i, && \forall (i \to j) \in \mathcal{E},
    \\ 
    &\text{[2]}\quad w_i^{\min} \le w_i \le w_i^{\max}, && \forall i \in \gV,
    \\ 
    &\text{[3]}\quad P_s = 0,\;\; w_s = 0, 
    \\ 
    &\text{[4]}\quad \sum_{i \in \gV_s} \delta_i (w_i^{\max}-w_i) \le r_{\max}\,|\gV_s|, && \forall s\in \{1,\dots,S\}.
\end{align*}
Constraint~[1] enforces precedence in the DAG, ensuring that an action $j$ cannot start until all predecessor actions $i$ have finished. 
Constraint~[2] bounds each action duration within $[w_i^{\min}, w_i^{\max}]$, capturing the feasible latency range under parameter freezing. 
Constraint~[3] anchors the beginning of the schedule by setting the source node to $P_s=w_s=0$.

Finally, Constraint~[4] imposes a stage-wise \emph{average freezing budget}. Using the linearized freeze ratio $r_i=\delta_i (w_i^{\max}-w_i)$ (see \autoref{eq:def_freeze_ratio}) for freezable backward nodes $i\in\gV_s$ at a stage $s$, Constraint~[4] is equivalent to 
\begin{equation}
    \frac{1}{|\gV_s|} \sum_{i \in \gV_s} r_i \le r_{\max} ,\quad \forall s,
    \label{eq:max_freeze_ratio}
\end{equation}
which prevents any pipeline stage from being over-frozen on average.
Here, let $\gV_s \subseteq \gV$ be the set of action nodes assigned to a stage $s$, and $r_{\max}$ be a user-specified \emph{maximum freeze ratio}. Although the constraint is imposed at the stage level, $r_{\max}$ effectively bounds the expected parameter-level freeze ratio, since each stage is associated with a distinct set of parameters, and parameter freezing is applied via random selection within each action.
As a result, this constraint ensures that no parameter exceeds the allowable freeze budget in expectation while preserving the full linearity of the LP.


Solving this LP yields execution times $w_i$ for all nodes $v_i \in \mathcal{V}$, from which freeze ratios $r_i$ are derived.
This formulation ensures that each batch completes in minimal time while freezing parameters in a controlled manner across stages.
Moreover, since the formulation imposes no integrality constraints on the variables, it can be solved in polynomial time using standard linear programming solvers, such as 
interior-point methods~\citep{karmarkar1984new}.

\subsection{Phase III: Freezing}
\label{subsec:freezing_phase}
%
After determining the expected freeze ratios $\mathbf{r}^*$, TimelyFreeze applies freezing in a \emph{progressive} manner. 
Specifically, the freeze ratio for each action $v_i$ is gradually increased from zero toward $r^*_i$.
This smooth transition avoids abrupt changes in freezing behavior that could harm training stability.
Formally, at each step $t$, parameters are frozen according to the actual freeze ratio $\textrm{AFR}_{i,t}$. 
Once $t > T_f$, the system enters a \emph{stable freezing phase}, where $\textrm{AFR}_{i,t} = r_i$. The actual freeze ratio for $t > T_m$ is defined as
\begin{align}
    \textrm{AFR}_{i,t}
    = \min\Big\{
        r_i,\;
        r_i \cdot \frac{t - T_m}{T_f - T_m}
      \Big\}.
\end{align}

Regarding which parameters to freeze, TimelyFreeze adopts \emph{uniform random selection} as a strong and unbiased reference.
Under a given $\textrm{AFR}_{i,t}$, this strategy assigns equal freezing probability to all parameters associated with each action.
As a result, it isolates the effect of the freeze ratio itself from additional heuristics or parameter-level statistics.

\autoref{fig:thp_fr} illustrates the evolution of the actual freeze ratio (blue line) over training steps.
The actual freeze ratio starts from zero at $T_m$, the boundary between the monitoring and freezing phases, and gradually increases toward the expected freeze ratio $r^*$.
As the actual freeze ratio increases, the training throughput (gray line) improves accordingly. The full step-level algorithm is provided in Appendix~\ref{appendix:step_level_algorithm}.

\begin{figure}[t]
    \centering
    \includegraphics[width=\linewidth, trim=15pt 0 5pt 0, clip]{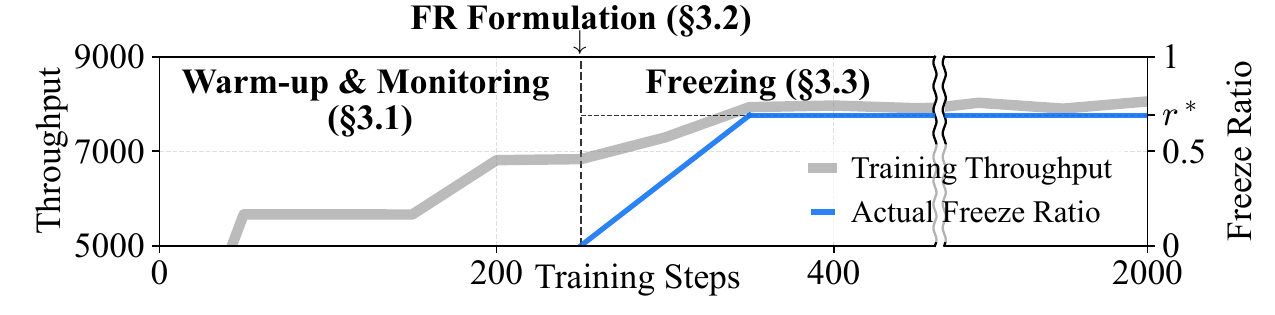}
    \caption{Freeze ratio and training throughput (tokens/sec) across training steps. 
    After warm-up and monitoring, the freeze ratio is gradually increased toward the expected value $r^*$, resulting in a corresponding increase in training throughput.}
    \label{fig:thp_fr}
\end{figure}

\subsection{Time-to-Accuracy Analysis}
\label{subsec:tta_analysis}

We analyze convergence in terms of \emph{time-to-accuracy} (TTA), defined as the
wall-clock time required to reach an $\varepsilon$-stationary point.
Formally, TTA is given by
\begin{equation}
\mathrm{TTA}
\;=\;
\sum_{t=1}^{T_{\varepsilon}} \tau_t,
\end{equation}
where $T_{\varepsilon}$ denotes the number of optimization steps required to
reach the target stationary level $\varepsilon$, and $\tau_t$ denotes the average execution time at step $t$.

\textbf{Iteration Complexity.}
Compared to the no-freezing baseline, TimelyFreeze may require more optimization
steps, since parameter freezing partially suppresses gradient updates.
Our analysis makes this trade-off explicit. In expectation, the effective descent
per step is reduced proportionally to the fraction of parameters being updated,
which yields the iteration-complexity scaling,
\begin{equation}
T_{\varepsilon}^{\mathrm{ours}}
\;\approx\;
\frac{1}{\bar p_{\mathrm{eff}}}\,
T_{\varepsilon}^{\mathrm{base}},
\label{eq:Teps_ratio_main}
\end{equation}
where $\bar p_{\mathrm{eff}}\in[1-r_{\max},1]$ denotes the average effective update
probability (see Appendix~\ref{appendix:convergence_analysis} for the definition).

\textbf{Per-step Execution Time.}
Despite the increase in iteration count, TimelyFreeze reduces the per-step
execution time by freezing throughput bottlenecks along the pipeline critical
path.
We summarize this effect by
\begin{equation}
\tau^{\mathrm{ours}}
\;=\;
\kappa\,\tau^{\mathrm{base}},
\qquad
\kappa \in (0,1),
\end{equation}
where $\kappa$ captures the per-step speedup induced by parameter freezing under TimelyFreeze.
We note that $\kappa$ is directly observable in practice.
Since $\tau \propto 1\,/\,\text{throughput}$, it can be computed from the measured
throughput values reported in the experimental results.

\textbf{Time-to-Accuracy Trade-off.}
Combining the above results, the overall TTA ratio satisfies
\begin{equation}
\frac{\mathrm{TTA}_{\varepsilon}^{\mathrm{ours}}}
{\mathrm{TTA}_{\varepsilon}^{\mathrm{base}}}
\;\approx\;
\frac{\kappa}{\bar p_{\mathrm{eff}}}.
\label{eq:TTA_ratio_main}
\end{equation}
Therefore, TimelyFreeze achieves a smaller time-to-accuracy whenever
the per-step speedup outweighs the increase in iteration complexity, i.e.,
$\kappa < \bar p_{\mathrm{eff}}$, which depends on the effective update probability in practice, even for a large $r_{\max}$.
A detailed derivation is provided in Appendix~\ref{appendix:tta_analysis}.



\section{Experiments}
\label{sec:experiment}

\begin{table*}[t]
\renewcommand{\arraystretch}{1.13}
\centering
\caption{
Comparison of freezing methods across different pipeline schedules (LLaMA-3-8B).
The average accuracy of the pretrained (no fine-tuning) baseline is 50.81.
The best and second-best values among the five freezing methods (APF, AutoFreeze, \TimelyFreeze{}, and two hybrid variants) are highlighted in bold and underline, respectively. $r_{max}$ for \TimelyFreeze{} is set to 0.8 for all experiments.
}
\setlength{\tabcolsep}{3pt}
\vspace*{-0.2cm}

\resizebox{\textwidth}{!}{%
\begin{tabular}{l|cc|cc c l|cc|cc}
\multicolumn{5}{c}{\textbf{GPipe}} & \cellcolor{white}\quad\quad &\multicolumn{5}{c}{\textbf{1F1B}} \\
\cmidrule[\heavyrulewidth]{1-5}\cmidrule[\heavyrulewidth]{7-11}

\multicolumn{1}{c|}{\multirow{2}{*}{\makecell[c]{Freeze \\ Method}}} &
\multicolumn{2}{c|}{Accuracy Preservation} &
\multicolumn{2}{c}{Time Efficiency}
& \cellcolor{white}\quad\quad &
\multicolumn{1}{c|}{\multirow{2}{*}{\makecell[c]{Freeze \\ Method}}} &
\multicolumn{2}{c|}{Accuracy Preservation} &
\multicolumn{2}{c}{Time Efficiency} \\
\cmidrule{2-5}\cmidrule{8-11}
&
\makecell[c]{Avg.\,Acc.\,($\Delta$)$\uparrow$} &
Frz.\,Ratio &
\makecell[c]{Throughput\,($\Delta$)$\uparrow$} &
MFU$\uparrow$
& \cellcolor{white}\quad\quad &
&
\makecell[c]{Avg.\,Acc.\,($\Delta$)$\uparrow$} &
Frz.\,Ratio &
\makecell[c]{Throughput\,($\Delta$)$\uparrow$} &
MFU$\uparrow$ \\
\cmidrule{1-5}\cmidrule{7-11}

No Freezing
& \aaccsign{54.63}{+0.00} & 0.00 & \ttimesign{5737}{0.00} & 25.94
& \cellcolor{white}\quad\quad &
No Freezing
& \aaccsign{54.63}{+0.00} & 0.00 & \ttimesign{5748}{0.00} & 25.99 \\
\graycline{1-5}\graycline{7-11}

APF
& \aaccsign{54.65}{+0.02} & 28.85 & \ttimesign{7293}{27.12} & 33.20
& \cellcolor{white}\quad\quad &
APF
& \aaccsign{54.65}{+0.02} & 28.85 & \ttimesign{7312}{27.20} & 33.29 \\
AutoFreeze
& \aaccsign{53.99}{-0.64} & 41.69 & \ttimesign{7351}{28.13} & 33.45
& \cellcolor{white}\quad\quad &
AutoFreeze
& \aaccsign{53.99}{-0.64} & 41.69 & \ttimesign{7367}{28.15} & 33.52 \\

\rowcolor{gray!15}
\TimelyFreeze{}
& \aaccsign{\underline{54.79}}{+0.17} & 35.64 & \ttimesign{\underline{7821}}{36.33} & \underline{35.69}
& \cellcolor{white}\quad\quad &
\TimelyFreeze{}
& \aaccsign{54.69}{+0.07} & 33.14 & \ttimesign{\underline{7867}}{36.87} & \underline{35.89} \\

\rowcolor{lightgray!10}
\scriptsize\quad +APF
& \aaccsign{\textbf{54.82}}{+0.19} & 35.66 & \ttimesign{\textbf{7839}}{36.63} & \textbf{35.77}
& \cellcolor{white}\quad\quad &
\scriptsize\quad +APF
& \aaccsign{\textbf{54.84}}{+0.21} & 34.23 & \ttimesign{\textbf{8024}}{39.59} & \textbf{36.66} \\

\rowcolor{lightgray!10}
\scriptsize\quad +AutoFreeze
& \aaccsign{54.65}{+0.03} & 35.06 & \ttimesign{7774}{35.50} & 35.44
& \cellcolor{white}\quad\quad &
\scriptsize\quad +AutoFreeze
& \aaccsign{\underline{54.82}}{+0.19} & 32.62 & \ttimesign{7827}{36.16} & 35.70 \\
\cmidrule[\heavyrulewidth]{1-5}\cmidrule[\heavyrulewidth]{7-11}

\multicolumn{5}{c}{\textbf{Interleaved 1F1B}} & \cellcolor{white}\quad\quad & \multicolumn{5}{c}{\textbf{ZBV}} \\
\cmidrule[\heavyrulewidth]{1-5}\cmidrule[\heavyrulewidth]{7-11}

\multicolumn{1}{c|}{\multirow{2}{*}{\makecell[c]{Freeze \\ Method}}} &
\multicolumn{2}{c|}{Accuracy Preservation} &
\multicolumn{2}{c}{Time Efficiency}
& \cellcolor{white}\quad\quad &
\multicolumn{1}{c|}{\multirow{2}{*}{\makecell[c]{Freeze \\ Method}}} &
\multicolumn{2}{c|}{Accuracy Preservation} &
\multicolumn{2}{c}{Time Efficiency} \\
\cmidrule{2-5}\cmidrule{8-11}
&
\makecell[c]{Avg.\,Acc.\,($\Delta$)$\uparrow$} &
Frz.\,Ratio &
\makecell[c]{Throughput\,($\Delta$)$\uparrow$} &
MFU$\uparrow$
& \cellcolor{white}\quad\quad &
&
\makecell[c]{Avg.\,Acc.\,($\Delta$)$\uparrow$} &
Frz.\,Ratio &
\makecell[c]{Throughput\,($\Delta$)$\uparrow$} &
MFU$\uparrow$ \\
\cmidrule{1-5}\cmidrule{7-11}

No Freezing
& \aaccsign{54.78}{+0.00} & 0.00 & \ttimesign{6173}{0.00} & 27.91
& \cellcolor{white}\quad\quad &
No Freezing
& \aaccsign{54.66}{+0.00} & 0.00 & \ttimesign{6889}{0.00} & 31.16 \\
\graycline{1-5}\graycline{7-11}

APF
& \aaccsign{\textbf{54.99}}{+0.21} & 58.05 & \ttimesign{7734}{25.29} & 35.20
& \cellcolor{white}\quad\quad &
APF
& \aaccsign{54.47}{-0.19} & 43.35 & \ttimesign{8522}{23.70} & 38.71 \\
AutoFreeze
& \aaccsign{54.54}{-0.24} & 76.86 & \ttimesign{7495}{21.43} & 34.10
& \cellcolor{white}\quad\quad &
AutoFreeze
& \aaccsign{54.86}{+0.20} & 24.28 & \ttimesign{7139}{3.62} & 32.28 \\

\rowcolor{gray!15}
\TimelyFreeze{}
& \aaccsign{54.64}{-0.14} & 60.07 & \ttimesign{\underline{8081}}{30.91} & \underline{36.78}
& \cellcolor{white}\quad\quad &
\TimelyFreeze{}
& \aaccsign{54.72}{+0.06} & 69.92 & \ttimesign{\textbf{8939}}{29.75} & \textbf{40.70} \\

\rowcolor{lightgray!10}
\scriptsize\quad +APF
& \aaccsign{54.71}{-0.06} & 60.72 & \ttimesign{\textbf{8108}}{31.36} & \textbf{36.92}
& \cellcolor{white}\quad\quad &
\scriptsize\quad +APF
& \aaccsign{\underline{55.03}}{+0.37} & 69.97 & \ttimesign{\underline{8923}}{29.51} & \underline{40.63} \\

\rowcolor{lightgray!10}
\scriptsize\quad +AutoFreeze
& \aaccsign{\underline{54.72}}{-0.06} & 60.76 & \ttimesign{8057}{30.53} & 36.67
& \cellcolor{white}\quad\quad &
\scriptsize\quad +AutoFreeze
& \aaccsign{\textbf{55.05}}{+0.39} & 70.03 & \ttimesign{8909}{29.32} & 40.56 \\

\cmidrule[\heavyrulewidth]{1-5}\cmidrule[\heavyrulewidth]{7-11}
\end{tabular}%
}

\vspace*{-0.3cm}
\label{tab:llama8b}
\end{table*}

\subsection{Hybrid Variants}
\label{sec:hybrid_variants}

As discussed in \autoref{sec:baseline}, both baselines, APF~\citep{adaptivefreezeapf2024} 
and AutoFreeze~\citep{autofreeze2021}, jointly determine 
\emph{how many parameters to freeze} and \emph{which parameters to freeze}, based on gradient-derived metrics.

In contrast, our method focuses solely on the first component by computing
stage-wise freeze ratios via our LP formulation, while remaining agnostic to
parameter selection.
To demonstrate that our method is complementary to existing parameter
selection heuristics, we introduce two hybrid variants, \TimelyAPF{} and \TimelyAuto{}.
In these variants, TimelyFreeze determines the freeze budget, while APF or AutoFreeze provides the parameter selection metric.
%
Detailed hybrid mechanism is provided in Appendix~\ref{appendix:hybrid_variants_algorithm}.

\subsection{Experimental Setup}
\textbf{Models and Datasets.} 
We conduct instruction finetuning task on the \texttt{LLaMA-3.2-1B}, \texttt{LLaMA-3-8B}~\citep{dubey2024llama}, and \texttt{LLaMA-2-13B}~\citep{touvron2023llama} models using the \texttt{Alpaca-GPT4}~\citep{peng2023instruction_alpaca_gpt4} and \texttt{OpenHermes-2.5}~\citep{OpenHermes_2.5} datasets to evaluate various freezing methods. 
To further demonstrate the applicability across domains, 
we conduct image classification by finetuning pretrained \texttt{ViT-L/32}~\citep{huggingface_vit} and \texttt{ConvNeXt-V2-L}~\citep{woo2023convnext} models on \texttt{ImageNet-1K}~\citep{deng2009imagenet} and \texttt{Food-101}~\citep{bossard14food101} datasets, respectively. 
For detailed experiment setup, refer to Appendix~\ref{appendix:exp_setting}.
The source code is available at \url{https://anonymous.4open.science/r/TimelyFreeze/}.

\textbf{Pipeline Schedules.} 
We validate our method under four representative pipeline schedules: \emph{GPipe}~\citep{huang2019gpipe}, \emph{1F1B}~\citep{fan2021dapple1f1b}, \emph{Interleaved 1F1B}~\citep{narayanan2021interleaved1f1b}, and \emph{Zero-Bubble V-shaped (ZBV)}~\citep{qi2023zerobubble}. 
We employ the \texttt{TorchTitan} framework~\citep{liang2024torchtitan} developed by PyTorch team for distributed training.

\textbf{Hardware Settings.}
All experiments were conducted on four Nvidia A6000 48GB GPUs and four H200 140GB GPUs. 
A6000 GPUs communicate via PCIe, whereas H200 GPUs communicate through NVLink. 

\textbf{Metrics.}
%
Model performance is assessed using the \emph{average accuracy} over four widely used benchmarks: \emph{MMLU} (5 shots)~\citep{hendrycks2020measuring_mmlu}, 
\emph{HellaSwag} (zero shot)~\citep{zellers2019hellaswag}, 
\emph{ARC-Challenge (ARC-C)} (10 shots)~\citep{clark2018think_arc_challenge}, 
and \emph{TruthfulQA} (zero shot)~\citep{lin2022truthfulqa}. 
We report the mean proportion of frozen parameters across all GPUs using the
\emph{average freeze ratio} (``Frz. Ratio'' in \autoref{tab:llama8b}), which is defined as
\[
\text{Average Freeze Ratio}
\;:=\;
\mathbb{E}_{t,i,j}\!\left[\mathbb{I}^{(j)}_{t,i}\right],
\]
where $\mathbb{I}^{(j)}_{t,i}$ is a Bernoulli random variable induced by uniform
random freezing, indicating whether a parameter $j$ of node $v_i$ is frozen at
step $t$.

\section{Results and Discussion}
\label{subsec:main_results}

\autoref{tab:llama8b} shows the accuracy preservation and time efficiency results for LLaMA-3-8B under four pipeline schedules. 
Here, the accuracy gain (e.g., $\Delta$ of Avg. Acc.) denotes the difference in average accuracy relative to the
no-freezing baseline, while throughput gain ($\Delta$ of Throughput) represents the relative throughput improvement.

Across all pipeline schedules, TimelyFreeze and its hybrid variants consistently improve throughput over the no-freezing baseline, while largely preserving accuracy.
In particular, under the 1F1B schedule, \TimelyAPF{} achieves the highest throughput improvement of 39.59\% with comparable accuracy, and in some cases slightly surpasses the no-freezing baseline.
Compared to the baseline freezing methods (APF and AutoFreeze), TimelyFreeze and its hybrids generally achieve higher or comparable accuracy across the GPipe, 1F1B, and ZBV schedules.


\subsection{Scaling Toward Larger Models}
\label{subsec:scaling_larger_models_llama}

\autoref{fig:llama_result} presents the accuracy--throughput trade-offs for the
LLaMA family as the model size scales from 1B to 8B and 13B under four pipeline schedules.
Across most settings, TimelyFreeze and its hybrid variants consistently lie on or near the Pareto frontier, achieving higher throughput with comparable accuracy relative to both the no-freezing baseline and baseline freezing methods.
 
For LLaMA-1B\,(top) in \autoref{fig:llama_result}, \TimelyFreeze{} and its hybrid variants achieve moderate throughput improvements of up to 19--21\% with less than 1\%p accuracy degradation compared to the no-freezing baseline.
 
As the model size increases, these gains become more pronounced.
On LLaMA-8B\,(middle) and LLaMA-13B\,(bottom), \TimelyFreeze{} yields substantially larger throughput improvements across GPipe and 1F1B schedules, even up to 46\% than the no-freezing baseline in 13B, while preserving accuracy at the comparable extent. 
For Interleaved 1F1B and ZBV on LLaMA-13B, accuracy shows higher variance, suggesting that additional multi-seed evaluation would be useful to assess stability at this scale.

%


Overall, the scaling results indicate that \TimelyFreeze{} scales favorably with
model size, yielding higher throughput gains while maintaining comparable
accuracy.
%
See Appendices~\ref{appendix:llama1b}--\ref{appendix:llama1b8b_benchmarks} for the detailed numbers and benchmark scores. In addition, the examples of actual pipeline execution schedules are visualized in Appendix~\ref{appendix:schedules_comparison_4gpus}.


\begin{figure}[t]
  \centering

  \makebox[\linewidth]{\fontsize{8.5pt}{10pt}\selectfont LLaMA-1B}
  \vspace{-0.8em}
  \includegraphics[width=\linewidth, trim=0cm 0cm 0.5cm 0.3cm, clip]{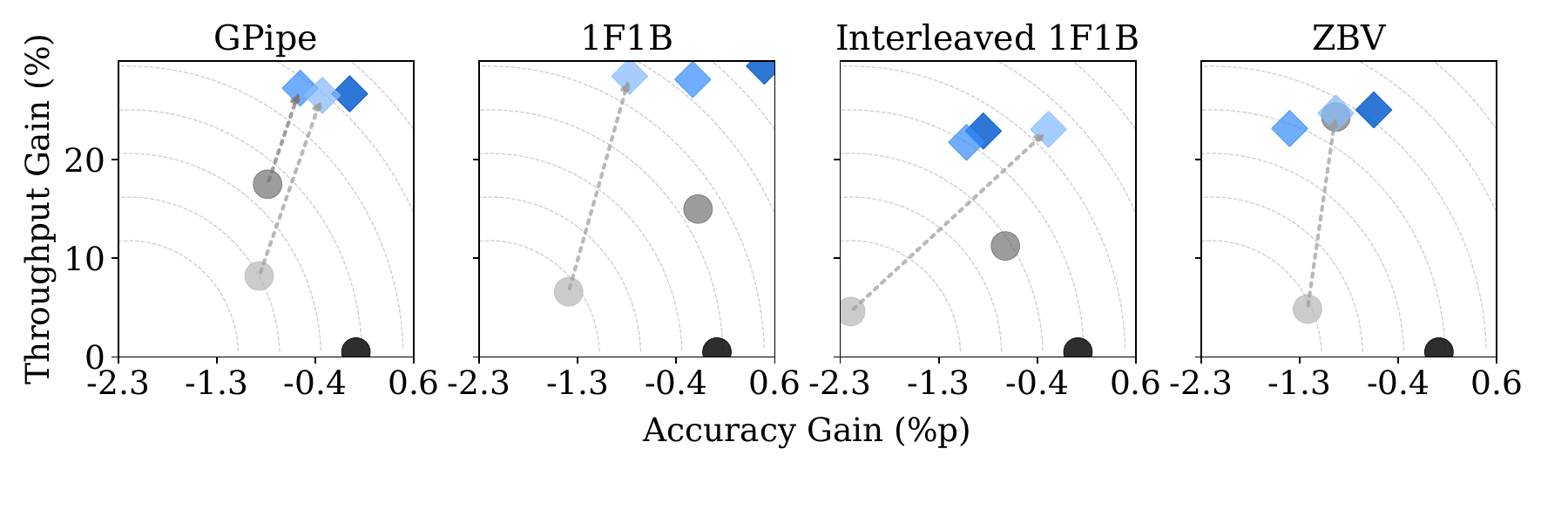}

  \vspace{-0.3em}
  \makebox[\linewidth]{\fontsize{8.5pt}{10pt}\selectfont LLaMA-8B}
  \vspace{-0.8em}
  \includegraphics[width=\linewidth, trim=0cm 0cm 0.5cm 0.3cm, clip]{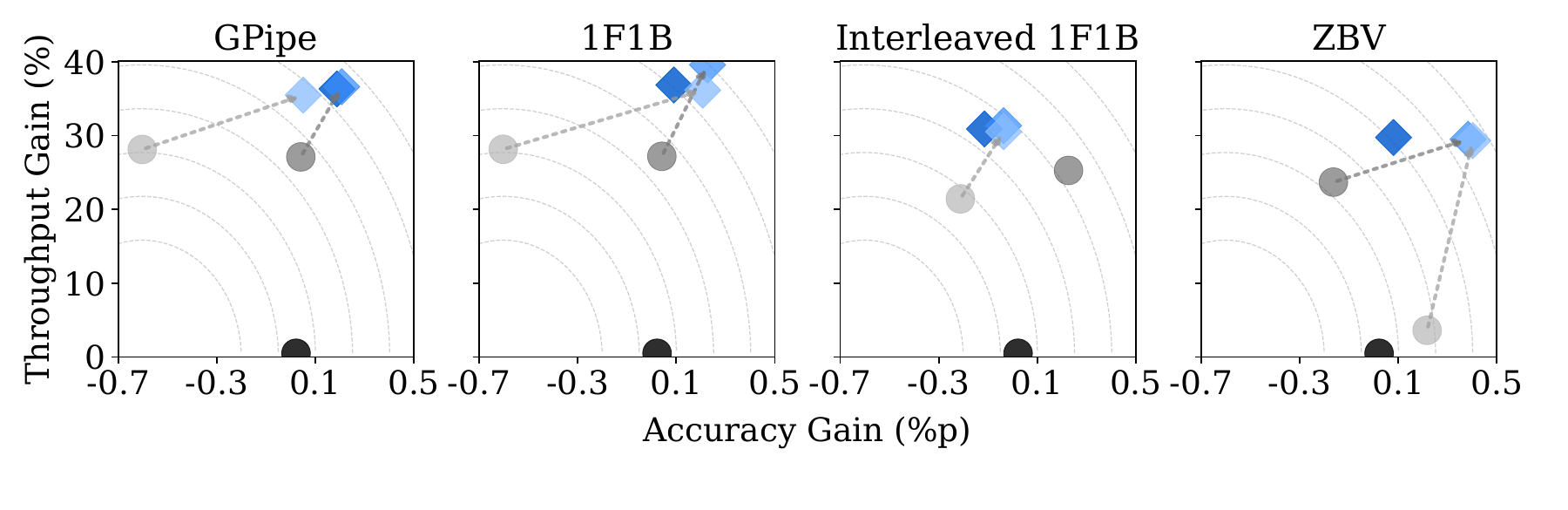}

  \vspace{-0.3em}
  \makebox[\linewidth]{\fontsize{8.5pt}{10pt}\selectfont LLaMA-13B}
  \includegraphics[width=\linewidth, trim=0cm 0cm 0.5cm 0.3cm, clip]{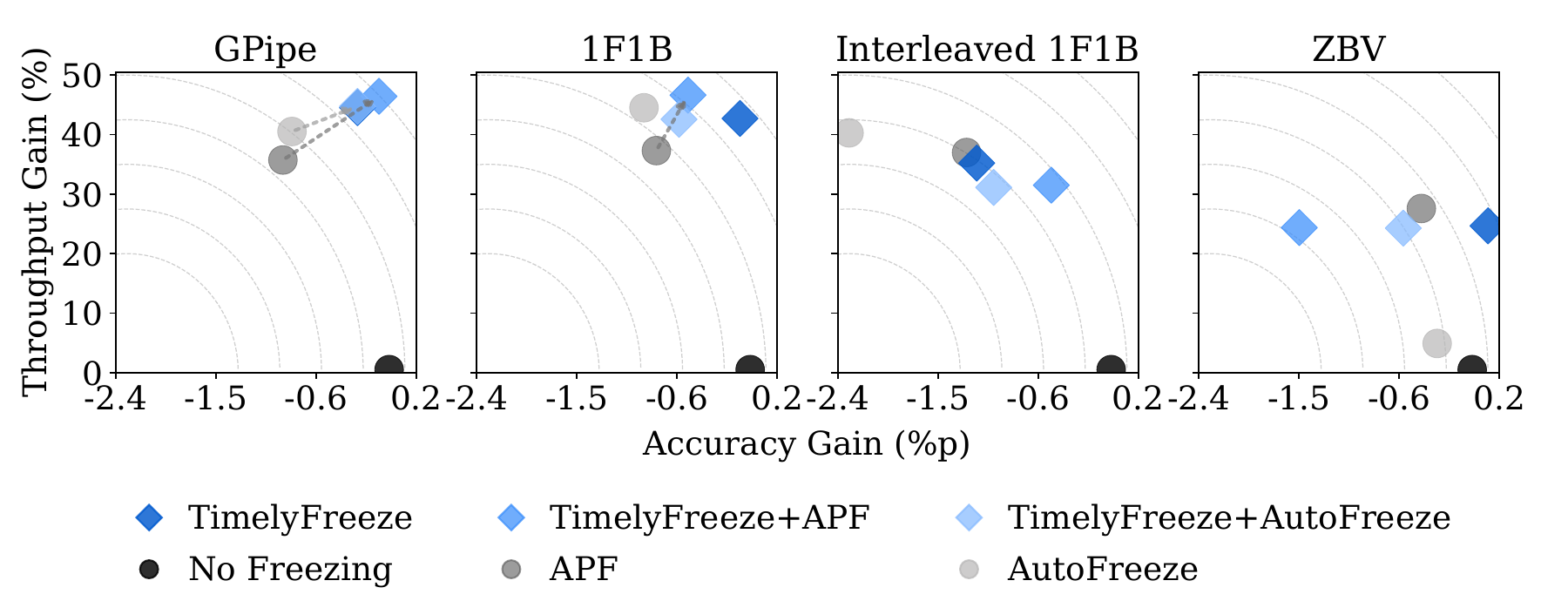}

  \caption{Comparison of accuracy--throughput trade-offs under GPipe, 1F1B, Interleaved 1F1B, and Zero-Bubble (ZBV).}
  \label{fig:llama_result}
\end{figure}

\subsection{Freezing Controller Sensitivity}

\begin{figure}[t]
    \centering
    \includegraphics[width=1.02\linewidth, trim=0.0cm 0.0cm 0.5cm 0.5cm, clip]{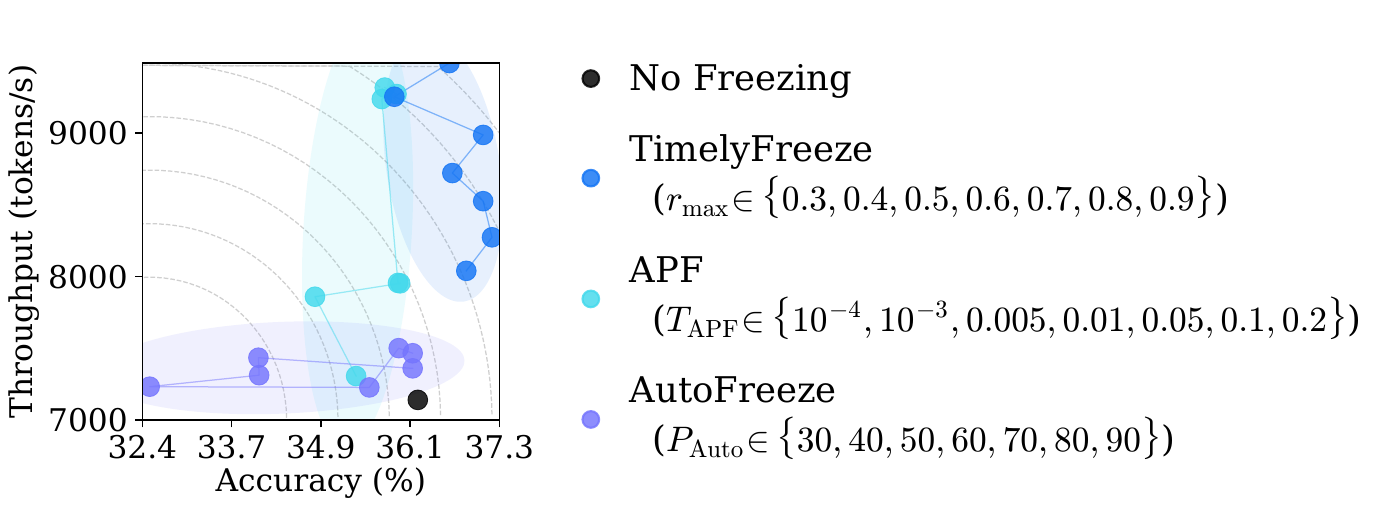}
    \caption{Accuracy--throughput trade-offs on LLaMA-1B under different freezing controller values with the 1F1B schedule.}
    \label{fig:llama1b_hyperparam}
\end{figure}
%

In \autoref{fig:llama1b_hyperparam}, we analyze the accuracy--throughput trade-off
of the LLaMA-1B model across different user-specified freezing controller values
(i.e., $r_{\max}$ for TimelyFreeze, $T_{\mathrm{APF}}$ for APF, and
$P_{\mathrm{Auto}}$ for AutoFreeze). The control values used for each method are indicated in the figure, and the corresponding results are arranged from left to right in the order of increasing freezing strength.


TimelyFreeze exhibits a consistent trend, as $r_{\max}$ decreases (i.e., weaker freezing), throughput decreases monotonically while accuracy remains well preserved across different $r_{\max}$ values.
In contrast, the two baselines consistently achieve lower throughput and accuracy
than TimelyFreeze across all controller settings.
Moreover, their throughput and accuracy vary irregularly as the user-specified hyperparameters change, indicating that $r_{\max}$ provides more predictable and controllable throughput behavior in TimelyFreeze.
Based on our experience, we recommend selecting $r_{max}$ in the range of $[0.5, 0.8]$, where accuracy is rather stable, depending on the desired accuracy--throughput trade-off.


\subsection{Generalization on Vision Models}
%
We further evaluate the generality of TimelyFreeze on vision models, including ViT-L/32 and ConvNeXt-V2-L, with detailed results reported in
Appendix~\ref{appendix:vision_models}.

ConvNeXt-V2-L exhibits highly unbalanced parameter distributions, where deeper layers contain substantially larger parameter groups, leading to execution-time skew under pipeline parallelism.
As explained in Appendix~\ref{appendix:convnext}, across different partitioning heuristics, TimelyFreeze consistently reduces training time while maintaining accuracy comparable to existing freezing methods.
Under the GPipe and 1F1B schedules, TimelyFreeze often attains the best or second-best training time among all freezing methods, despite using moderate freeze ratios.
This result shows that LP-based, pipeline-aware freezing effectively addresses execution-time imbalance arising from architectural unevenness.

Consistent robustness is also observed for ViT-L/32 finetuning in Appendix~\ref{appendix:vit}.
While APF suffers from severe accuracy degradation due to aggressive, metric-driven freezing, TimelyFreeze preserves accuracy close to the no-freezing baseline while achieving the largest training-time reduction under both pipeline schedules.
This result indicates that TimelyFreeze generalizes robustly beyond a specific model structure or task, providing stable performance even when importance metrics are noisy or unreliable.



\section{Conclusion}
\label{sec:conclusion}

\vfill
We proposed \textbf{\emph{TimelyFreeze}}, an adaptive parameter-freezing mechanism for pipeline-parallel training that explicitly incorporates pipeline execution structure.
By modeling the pipeline schedule as a DAG and computing stage-wise freeze ratios via linear programming, TimelyFreeze reduces batch latency while preserving model accuracy.
Extensive experiments on large language models and vision models demonstrate that TimelyFreeze consistently achieves superior accuracy--throughput trade-offs across pipeline schedules.
Future work includes extending TimelyFreeze to multi-node settings and integrating it with hybrid parallelism.



\section*{Acknowledgements}

This work was supported by Institute of Information \& Communications Technology Planning \& Evaluation\,(IITP) grant funded by the Korea government\,(MSIT) (No.\ RS-2020-II200862, DB4DL: High-Usability and Performance In-Memory Distributed DBMS for Deep Learning, 50\% and No.\ RS-2022-II220157, Robust, Fair, Extensible Data-Centric Continual Learning, 50\%).

\section*{Impact Statement}

This work introduces \TimelyFreeze{}, a system-level optimization technique that improves the efficiency of pipeline-parallel training by aligning parameter-freezing decisions with execution-time dynamics. By reducing redundant computation without modifying model architectures or learning objectives, \TimelyFreeze{} can lower training time and computational cost for large-scale models, potentially reducing energy consumption and improving accessibility to large-model training in resource-constrained environments.

The method is intended for use in controlled training settings and does not alter model inference behavior or introduce new capabilities that could raise additional ethical concerns beyond those associated with the underlying models. As with other efficiency-oriented system optimizations, the broader societal impacts of this work depend on how large-scale models are ultimately deployed. Overall, we do not foresee any specific negative impacts unique to \TimelyFreeze{} beyond those already present in large-scale model training and deployment.


\bibliography{reference}

@inproceedings{kaplan2020_scalinglaw,
  author    = {J. Kaplan and Sam McCandlish and T. Henighan and Tom B. Brown and Benjamin Chess and R. Child and Scott Gray and Alec Radford and Jeff Wu and Dario Amodei},
  year      = {2020},
  title     = {Scaling Laws for Neural Language Models},
  booktitle = {arXiv.org},
}

@inproceedings{you2018imagenet_dp,
  title={Imagenet Training in Minutes},
  author={You, Yang and Zhang, Zhao and Hsieh, Cho-Jui and Demmel, James and Keutzer, Kurt},
  booktitle={Proceedings of the International Conference on Parallel Processing},
  pages={1--10},
  year={2018}
}

@article{sergeev2018horovod_dp,
  title={Horovod: Fast and Easy Distributed Deep Learning in TensorFlow},
  author={Sergeev, Alexander and Del Balso, Mike},
  journal={arXiv preprint arXiv:1802.05799},
  year={2018}
}

@inproceedings{liang2024torchtitan,
  title={TorchTitan: One-Stop PyTorch Native Solution for Production Ready {LLM} Pretraining},
  author={Liang, Wanchao and Liu, Tianyu and Wright, Less and Constable, Will and Gu, Andrew and Huang, Chien-Chin and Zhang, Iris and Feng, Wei and Huang, Howard and Wang, Junjie and others},
  booktitle={Proceedings of the International Conference on Learning Representations},
  year={2025}
}

@article{shoeybi2019megatron,
  title={{Megatron-LM}: Training Multi-Billion Parameter Language Models Using Model Parallelism},
  author={Shoeybi, Mohammad and Patwary, Mostofa and Puri, Raul and LeGresley, Patrick and Casper, Jared and Catanzaro, Bryan},
  journal={arXiv preprint arXiv:1909.08053},
  year={2019}
}

@inproceedings{huang2019gpipe,
  title={Gpipe: Efficient Training of Giant Neural Networks Using Pipeline Parallelism},
  author={Huang, Yanping and Cheng, Youlong and Bapna, Ankur and Firat, Orhan and Chen, Dehao and Chen, Mia and Lee, HyoukJoong and Ngiam, Jiquan and Le, Quoc V and Wu, Yonghui and others},
  booktitle={Proceedings of the International Conference on Machine Learning},
  pages={103--112},
  year={2019}
}

@inproceedings{narayanan2019pipedream1f1b,
  title={PipeDream: Generalized Pipeline Parallelism for {DNN} Training},
  author={Narayanan, Deepak and Harlap, Aaron and Phanishayee, Amar and Seshadri, Vivek and Devanur, Nikhil R and Ganger, Gregory R and Gibbons, Phillip B and Zaharia, Matei},
  booktitle={Proceedings of the ACM Symposium on Operating Systems Principles},
  pages={1--15},
  year={2019}
}

@inproceedings{fan2021dapple1f1b,
  title={DAPPLE: A Pipelined Data Parallel Approach for Training Large Models},
  author={Fan, Shiqing and Rong, Yi and Meng, Chen and Cao, Zongyan and Wang, Siyu and Zheng, Zhen and Wu, Chuan and Long, Guoping and Yang, Jun and Xia, Lixue and others},
  booktitle={Proceedings of the ACM SIGPLAN Symposium on Principles and Practice of Parallel Programming},
  pages={431--445},
  year={2021}
}

@inproceedings{narayanan2021interleaved1f1b,
  title={Efficient Large-Scale Language Model Training on {GPU} Clusters Using Megatron-{LM}},
  author={Narayanan, Deepak and Shoeybi, Mohammad and Casper, Jared and LeGresley, Patrick and Patwary, Mostofa and Korthikanti, Vijay and Vainbrand, Dmitri and Kashinkunti, Prethvi and Bernauer, Julie and Catanzaro, Bryan and others},
  booktitle={Proceedings of the International Conference for High Performance Computing, Networking, Storage and Analysis},
  pages={1--15},
  year={2021}
}

@article{qi2023zerobubble,
  title={Zero Bubble Pipeline Parallelism},
  author={Qi, Penghui and Wan, Xinyi and Huang, Guangxing and Lin, Min},
  journal={arXiv preprint arXiv:2401.10241},
  year={2023}
}

@article{qi2024pipeline_controllable_memory,
  title={Pipeline Parallelism With Controllable Memory},
  author={Qi, Penghui and Wan, Xinyi and Amar, Nyamdavaa and Lin, Min},
  journal={arXiv preprint arXiv:2405.15362},
  year={2024}
}

@article{adaptivefreezeapf2024,
  title={Synchronize Only The Immature Parameters: Communication-Efficient Federated Learning by Freezing Parameters Adaptively},
  author={Chen, Chen and Xu, Hong and Wang, Wei and Li, Baochun and Li, Bo and Chen, Li and Zhang, Gong},
  journal={IEEE Transactions on Parallel and Distributed Systems},
  volume={35},
  number={7},
  pages={1155--1173},
  year={2023}
}

@article{autofreeze2021,
  title={Autofreeze: Automatically Freezing Model Blocks to Accelerate Fine-Tuning},
  author={Liu, Yuhan and Agarwal, Saurabh and Venkataraman, Shivaram},
  journal={arXiv preprint arXiv:2102.01386},
  year={2021}
}

@inproceedings{egeria2023_knowledge_guided_freezing,
  title={Egeria: Efficient {DNN} Training with Knowledge-Guided Layer Freezing},
  author={Wang, Yiding and Sun, Decang and Chen, Kai and Lai, Fan and Chowdhury, Mosharaf},
  booktitle={Proceedings of the European Conference on Computer Systems},
  pages={851--866},
  year={2023}
}

@article{freezeout2017,
  title={Freezeout: Accelerate Training by Progressively Freezing Layers},
  author={Brock, Andrew and Lim, Theodore and Ritchie, James M and Weston, Nick},
  journal={arXiv preprint arXiv:1706.04983},
  year={2017}
}

@article{lee2019would_elsa_freezing,
  title={What would elsa do? freezing layers during transformer fine-tuning},
  author={Lee, Jaejun and Tang, Raphael and Lin, Jimmy},
  journal={arXiv preprint arXiv:1911.03090},
  year={2019}
}

@inproceedings{he2021pipetransformer_freezing,
  title={Pipetransformer: Automated elastic pipelining for distributed training of large-scale models},
  author={He, Chaoyang and Li, Shen and Soltanolkotabi, Mahdi and Avestimehr, Salman},
  booktitle={Proceedings of the International Conference on Machine Learning},
  pages={4150--4159},
  year={2021},
  organization={PMLR}
}

@article{li2024smartfrz_freezing,
  title={Smartfrz: An efficient training framework using attention-based layer freezing},
  author={Li, Sheng and Yuan, Geng and Dai, Yue and Zhang, Youtao and Wang, Yanzhi and Tang, Xulong},
  journal={arXiv preprint arXiv:2401.16720},
  year={2024}
}

@inproceedings{woo2023convnext,
  title={Convnext v2: Co-Designing and Scaling Convnets with Masked Autoencoders},
  author={Woo, Sanghyun and Debnath, Shoubhik and Hu, Ronghang and Chen, Xinlei and Liu, Zhuang and Kweon, In So and Xie, Saining},
  booktitle={Proceedings of the IEEE/CVF Conference on Computer Vision and Pattern Recognition},
  pages={16133--16142},
  year={2023}
}

@article{touvron2023llama,
  title={The {L}lama 3 Herd of Models},
  author={Hugo Touvron and Louis Martin and Kevin Stone and Peter Albert and Amjad Almahairi and Yasmine Babaei and Nikolay Bashlykov and Soumya Batra and Prajjwal Bhargava and others},
  journal={arXiv preprint arXiv:2407.21783},
  year={2024}
}

@article{dubey2024llama,
  title={Llama 2: Open Foundation and Fine-Tuned Chat Models},
  author={Dubey, Abhimanyu and Jauhri, Abhinav and Pandey, Abhinav and Kadian, Abhishek and Al-Dahle, Ahmad and Letman, Aiesha and Mathur, Akhil and Schelten, Alan and Yang, Amy and Fan, Angela and others},
  journal={arXiv preprint 	arXiv:2307.09288},
  year={2023}
}

@misc{huggingface_vit,
  author      = {Huggingface},
  title       = {Vision Transformer (ViT)},
  url         = {https://huggingface.co/docs/transformers/model_doc/vit},
  lastaccessed = {June 30, 2025},
  year         = {2025}
}

@inproceedings{bossard14food101,
  title = {Food-101 -- Mining Discriminative Components with Random Forests},
  author = {Bossard, Lukas and Guillaumin, Matthieu and Van Gool, Luc},
  booktitle = {Proceedings of the European Conference on Computer Vision},
  pages = {446--461},
  year = {2014}
}

@inproceedings{deng2009imagenet,
  title={Image{N}et: A Large-Scale Hierarchical Image Database},
  author={Deng, Jia and Dong, Wei and Socher, Richard and Li, Li-Jia and Li, Kai and Fei-Fei, Li},
  booktitle={Proceedings of the IEEE/CVF Conference on Computer Vision and Pattern Recognition},
  pages={248--255},
  year={2009}
}

@article{peng2023instruction_alpaca_gpt4,
  title={Instruction Tuning with GPT-4},
  author={Peng, Baolin and Li, Chunyuan and He, Pengcheng and Galley, Michel and Gao, Jianfeng},
  journal={arXiv preprint arXiv:2304.03277},
  year={2023}
}

@misc{OpenHermes_2.5,
  title = {OpenHermes 2.5: An Open Dataset of Synthetic Data for Generalist LLM Assistants},
  author = {Teknium},
  year = {2023},
  publisher = {HuggingFace},
  url = {https://huggingface.co/datasets/teknium/OpenHermes-2.5}
}

@article{hendrycks2020measuring_mmlu,
  title={Measuring massive multitask language understanding},
  author={Hendrycks, Dan and Burns, Collin and Basart, Steven and Zou, Andy and Mazeika, Mantas and Song, Dawn and Steinhardt, Jacob},
  journal={arXiv preprint arXiv:2009.03300},
  year={2020}
}

@article{zellers2019hellaswag,
  title={Hellaswag: Can a machine really finish your sentence?},
  author={Zellers, Rowan and Holtzman, Ari and Bisk, Yonatan and Farhadi, Ali and Choi, Yejin},
  journal={arXiv preprint arXiv:1905.07830},
  year={2019}
}

@article{clark2018think_arc_challenge,
  title={Think you have solved question answering? try arc, the ai2 reasoning challenge},
  author={Clark, Peter and Cowhey, Isaac and Etzioni, Oren and Khot, Tushar and Sabharwal, Ashish and Schoenick, Carissa and Tafjord, Oyvind},
  journal={arXiv preprint arXiv:1803.05457},
  year={2018}
}

@inproceedings{lin2022truthfulqa,
  title={Truthfulqa: Measuring how models mimic human falsehoods},
  author={Lin, Stephanie and Hilton, Jacob and Evans, Owain},
  booktitle={Proceedings of the Annual Meeting of the Association for Computational Linguistics},
  pages={3214--3252},
  year={2022}
}

@article{karmarkar1984new,
  title        = {A New Polynomial-Time Algorithm for Linear Programming},
  author       = {Karmarkar, Narendra},
  journal      = {Combinatorica},
  volume       = {4},
  number       = {4},
  pages        = {373--395},
  year         = {1984},
  doi          = {10.1007/BF02579150}
}
\bibliographystyle{icml2026}

\newpage
\appendix
\onecolumn
\section{Notations}
\label{appendix:notations}
\begin{table}[H]
\centering
\caption{Notation used in \TimelyFreeze{}.}
\label{tab:notation}
\begin{tabular}{c l}
\toprule
\textbf{Symbol} & \textbf{Description} \\
\midrule
$\mathcal{G}=(\mathcal{V},\mathcal{E})$ 
& Directed acyclic graph (DAG) representing a pipeline schedule. \\

$\mathcal{V}$ 
& Set of action nodes in the pipeline DAG. \\

$\mathcal{E}$ 
& Set of directed edges encoding precedence constraints. \\

$v_{(a,m,s)}$ 
& Action node with action type $a$, microbatch $m$, at stage $s$. For simplicity, it is also denoted as $v_i$. \\

$a \in \{f,b\}$ 
& Action type: forward ($f$) or backward ($b$). \\

$m \in \{1,\dots,M\}$ 
& Microbatch index; $M$ is the total number of microbatches. \\

$s \in \{1,\dots,S\}$ 
& Pipeline stage index; $S$ is the number of stages. \\

$v_s,\; v_d$ 
& Source and destination nodes representing batch start and completion. \\

$P_i$ 
& Start time of action node $v_i$, determined by the longest path in $\mathcal{G}$. \\

$P_d$ 
& Start time of destination node $v_d$, corresponding to batch execution time. \\

$w_i$ 
& Execution time (duration) of action node $v_i$. \\

$w_i^{\min},\; w_i^{\max}$ 
& Minimum and maximum execution times of node $v_i$ under full and no freezing. \\

$r_i$ 
& Expected freeze ratio of action node $v_i$. \\

$r_{\max}$ 
& Maximum average freeze ratio allowed per pipeline stage specified by the user. \\

$t\in \{1,\cdots,T_{\mathrm{tot}}\}$
& Training step; $T_{\mathrm{tot}}$ is the total number of training steps. \\

$T_w,\; T_m,\;T_f$
& The last step of the warm-up, monitoring, and progressive freezing phase, respectively. \\

$N$
& Total number of model parameters. \\
$N_s$  
& Number of parameters in stage $s$ and $N=\sum_{s=1}^S N_s$ holds. \\

$\mathbb{I}_{t,i} $
& Binary freezing mask of the action node $v_i$ at step $t$; \\

& $\mathbb{I}^{(j)}_{t,i} = 1$ if parameter $j$ associated with $v_i$
is selected to be frozen at step $t$, and $0$ otherwise.\\

\bottomrule
\end{tabular}
\end{table}

\section{Edge Construction Rules}
\label{appendix:edge_construction_rules}
\begin{enumerate}[leftmargin=1.2em, label=\arabic*.]
    \item \emph{Source and destination connections:} 
    The source and destination nodes are connected to the actions nodes as
    \begin{displaymath}
    v_s \to v_{(f,1,1)}, \quad v_{(b,M,1)} \to v_d.
    \end{displaymath}
    Here, a batch starts with the forward action of the first microbatch at the first stage and ends with the backward action of the last microbatch at the first stage.
    
    \item \emph{Intra-stage dependencies:} 
    Within each stage, any action of microbatch $m+1$ starts only after the corresponding action of microbatch $m$ finishes. Moreover, each backward action must follow its corresponding forward action.
    \begin{displaymath}
        v_{(a,m,s)} \to v_{(a,m+1,s)}, \quad
        v_{(f,m,s)} \to v_{(b,m,s)}.
    \end{displaymath}
    \item \emph{Inter-stage dependencies:} A forward action at stage $s+1$ can begin only after the forward action at stage $s$ completes, and a backward action at stage $s-1$ can begin only after the backward action at stage $s$ completes.
    \begin{displaymath}
    v_{(f,m,s)} \to v_{(f,m,s+1)},\quad
    v_{(b,m,s)} \to v_{(b,m,s-1)}.
    \end{displaymath}
    \item \emph{Pipeline schedule dependencies:} 
    For action nodes assigned to the same GPU rank, the execution order defined by the pipeline schedule must be respected. This constraint depends on the pipeline schedule. For example, under GPipe, any backward action can begin only after the last forward microbatch completes (i.e., $v_{(f,M,s)} \to v_{(b,1,s)}$), whereas 1F1B does not impose this constraint. 
\end{enumerate}

\clearpage
\section{Algorithms}
\subsection{TimelyFreeze's Freezing Algorithm}
\label{appendix:step_level_algorithm}
\SetKwComment{Comment}{$\triangleright$\ }{}
\SetKwProg{Fn}{Function}{}{}
\SetKw{Continue}{continue}

\begin{algorithm}[H]
\setcounter{AlgoLine}{0}
\caption{Step-Level Freeze Ratio Scheduling (Stage $s$)}
\label{alg:timelyfreeze_steps}
\DontPrintSemicolon

\KwIn{
Phase boundaries $\{T_w, T_m, T_f, T_{total}\}$;
Parameter set $\gP_s$ and
action set $\mathcal{V}_s$ corresponding to stage $s$.
}

\For{$t \gets 1$ \KwTo $T_{total}$}{
  \If{$t \le T_w$}{
    \textbf{continue} \Comment*[r]{Warm-up (no freezing)}
  }

  \ForEach{$v_i \in \mathcal{V}_g$}{
    Record execution time $\tau_{i,t}$
  }

  \uIf{$T_w < t \le T_m/2$}{
    \ForEach{$v_i$}{
      $\textrm{AFR}_{i,t} \gets 0$ \Comment*[r]{Upperbound Monitoring}
    }
  }

  \uElseIf{$T_m/2 < t \le T_m$}{
    \ForEach{$v_i$}{
      $\textrm{AFR}_{i,t} \gets 1$ \Comment*[r]{Lowerbound Monitoring}
    }
  }

  \uElseIf{$t = T_m$}{
    $\{\tau\}\gets $ All-gather $\{\tau_{i,t'} : v_i \in \mathcal{V}_s,\; t'=1,\dots,T_m\}$ across stages $\{s=1,\cdots, S\}$\;

    $\{r_i\}, \{P_i\} \gets \textsc{LPSolver}(
    \{\tau\})$ \Comment*[r]{LP Formulation}
  }

  \Else{ 
    \ForEach{$v_i \in \mathcal{V}_s$}{
      $\textrm{AFR}_{i,t} \gets 
        \min\;\Big\{r_i, \;\;r_i \cdot \dfrac{t - T_m}{T_f - T_m}\Big\}$ \Comment*[r]{Freezing}

      Randomly select $I_i \subseteq \gP_s$
      such that $E[|I_i|] = \textrm{AFR}_{i,t} \cdot |\gP_s|$\;

      Freeze parameters in $I_i$\;
    }
  }
}
\end{algorithm}

\subsection{Hybrid Variants Algorithm}
\label{appendix:hybrid_variants_algorithm}

\SetKwComment{Comment}{$\triangleright$\ }{}
\SetKwProg{Fn}{Function}{}{}
\SetKw{Continue}{continue}

\begin{algorithm}[H]
\setcounter{AlgoLine}{0}
\caption{Metric-Aware Parameter Selection for \TimelyAPF{} / \TimelyAuto{}}
\label{alg:hybrid_variants}
\DontPrintSemicolon

\KwIn{
Expected freeze ratios $\{r_i\}$ from TimelyFreeze; Current step $t$ and stage $s$;
Number of parameters $N_s$ in stage $s$;
Baseline method's freezing mask $\mathbb{I}_{base,t}$ at step $t$.
}

\KwOut{
Binary freezing mask $\mathbb{I}_t \in \{0,1\}^{n_s}$
}

\ForEach{$v_i\in \gV_s$}{
  $N_{freeze} \gets \lfloor r_i \cdot N_s \rfloor$ \Comment*[r]{Number of freezing parameters}

  \uIf{$\|\mathbb{I}_{t,i}\|_1
 = N_{freeze}$}{
    $\mathbb{I}_{t,i} \gets \mathbb{I}_{base,t,i}$
  }
  \uElseIf{$\|\mathbb{I}_{t,i}\|_1
 < N_{freeze} $}{
    $\gP_{add} \gets
    \text{Uniformly sample } (N_{\text{freeze}} - \|\mathbb{I}_{t,i}\|_1)
    \text{ indices from } \{j \mid \mathbb{I}^{(j)}_{t,i} = 0\}$

    $\mathbb{I}^{(j)}_{t,i} \gets
    \mathbb{I}\!\left[
    \mathbb{I}^{(j)}_{\text{base},t,i} = 1 \;\lor\; j \in \gP_{add}
    \right],
    \quad \forall j \in \{1,\dots,n_s\}$
  }
  \Else{
    $\gP_{sub} \gets
    \text{Uniformly sample } (\|\mathbb{I}_{t,i}\|_1 - N_{\text{freeze}}) 
    \text{ indices from } \{j \mid \mathbb{I}^{(j)}_{t,i} = 1\}$

    $\mathbb{I}^{(j)}_{t,i} \gets
    \mathbb{I}\!\left[
    \mathbb{I}^{(j)}_{\text{base},t,i} = 1 \;\land\; j \notin \gP_{sub}
    \right],
    \quad \forall j \in \{1,\dots,n_s\}$
  }

    Freeze all parameters $j$ such that $\mathbb{I}^{(j)}_{t,i} = 1$\;

}
\end{algorithm}

\section{Time-to-Accuracy Analysis}
\label{appendix:tta_analysis}

In this section, we analyze the \emph{wall-clock time-to-accuracy} (TTA), defined as
\begin{equation}
    \mathrm{TTA}_\varepsilon^{\mathrm{alg}} := 
    \mathbb{E}\!\left[\sum_{t=1}^{T_\varepsilon^{\mathrm{alg}}}\tau_t^{\mathrm{alg}}\right].
    \label{eq:def_TTA}
\end{equation}
Here, $T_\varepsilon^{\mathrm{alg}}$ denotes the (sufficient) number of training steps required by freezing algorithm $\mathrm{alg}$ to reach a target error $\varepsilon$ under the specified convergence criterion, and $\tau_t^{\mathrm{alg}}$ is the execution time of step $t$ under $\mathrm{alg}$.
Accordingly, we decompose the TTA analysis into two components:
(1) the iteration complexity $T_\varepsilon^{\mathrm{alg}}$ (\autoref{appendix:convergence_analysis}), and
(2) the per-step execution time $\tau_t^{\mathrm{alg}}$ (\autoref{appendix:subsec:execution_time_tau_analysis}).
Finally, we combine these two results to compare the overall TTA of TimelyFreeze with that of the no-freezing baseline in \autoref{appendix:subsec:tta_comparison}.

\subsection{Convergence Analysis}
\label{appendix:convergence_analysis}


Under the assumptions in \autoref{app:assumptions} and the lemmas in \autoref{app:lemmas},
we analyze the convergence behavior of a SGD framework under TimelyFreeze and characterize the number of training steps required to reach a target stationary level.
Specifically, we derive the iteration complexity under TimelyFreeze (denoted by $T_\varepsilon^{\mathrm{ours}}$) in \autoref{appendix:subsec:iteration_complexity_timelyfreeze},
and compare it with that of the \emph{no-freezing baseline} ($T_\varepsilon^{\mathrm{base}}$) in \autoref{appendix:subsec:iteration_complexity_baseline}.

A training step is indexed by $t \in \{1,\dots,T\}$, and each step consists of $M$ microbatches. For simplicity, we assume that a gradient update is performed after each batch execution, so that the number of training steps coincides with the number
of executed batches. 
Let $\theta_t \in \mathbb{R}^N$ denote the model parameters at step $t$, where $N$ is the number of trainable parameters, and define the population objective as
\begin{equation}
F(\theta) := \mathbb{E}_{\xi \sim \mathcal{D}}\big[\ell(\theta;\xi)\big].
\end{equation}
For each microbatch $m \in \{1,\dots,M\}$, we define the stochastic gradient
\begin{equation}
g_{t,m} := \nabla_\theta \ell(\theta_t;\xi_{t,m}).
\end{equation}

\textbf{Iteration Complexity.}
For an algorithm $\mathrm{alg}\in\{\mathrm{ours},\mathrm{base}\}$, we define its \emph{iteration complexity} $T_{\varepsilon}^{\mathrm{alg}}$ as
the smallest number of training steps required to reach an
$\varepsilon$-stationary point, measured by the expected squared gradient norm, i.e.,

\begin{equation}
T_{\varepsilon}^{\mathrm{alg}}
\;:=\;
\inf\Big\{T\in\mathbb{N}:\;
\frac{1}{T}\sum_{t=1}^{T}
\mathbb{E}\big[\|\nabla F(\theta_t)\|^2\big]
\;\le\;
\varepsilon
\Big\}.
\label{def:Talg_def}
\end{equation}

\textbf{Freezing and Update Masks (Our Convention).}
Let $\mathbb{I}_{t,m} \in \{0,1\}^N$ denote the \emph{freezing mask} applied at step $t$ and microbatch $m$, where for each parameter $j\in\{1,\dots,N\}$,
\begin{equation}
\mathbb{I}_{t,m}^{(j)} = 1 \ \text{indicates that parameter $j$ is frozen (no update)},\quad
\mathbb{I}_{t,m}^{(j)} = 0 \ \text{indicates that parameter $j$ is updated}.
\end{equation}
For convenience, we define the corresponding \emph{update mask} as
\begin{equation}
\mathbb{U}_{t,m}^{(j)} := \mathbf{1} - \mathbb{I}_{t,m}^{(j)} \in \{0,1\}^N,
\end{equation}
where $\mathbb{U}_{t,m}^{(j)} = 1$ indicates that parameter $j$ is updated and $\mathbb{U}_{t,m}^{(j)} = 0$ otherwise.

\textbf{Update Rule.}
The model parameters are updated using masked microbatch gradients as
\begin{equation}
\theta_{t+1} = \theta_t - \eta\,\Delta_t,
\qquad
\Delta_t := \frac{1}{M}\sum_{m=1}^M \big(\mathbb{U}_{t,m}\odot g_{t,m}\big),
\label{eq:update_rule}
\end{equation}
where $\eta > 0$ is the stepsize.
For simplicity, we assume a constant stepsize $\eta$ across training steps.

\subsubsection{Assumptions}
\label{app:assumptions}
\begin{assumption}[$L$-Smoothness] 
\label{assump:l-smoothness}
$F$ is $L$-smooth:
\begin{equation}
F(y)\le F(x) + \langle \nabla F(x), y-x\rangle + \frac{L}{2}\|y-x\|^2,\quad \forall x,y.
\label{eq:smooth}
\end{equation}
\end{assumption}

\begin{assumption}[Unbiased Gradients and Bounded Variance]
\label{assump:unbiased_bounded_var}
For all step $t$ and microbatch $m$,
\begin{equation}
\mathbb{E}[g_{t,m}\mid \theta_t] = \nabla F(\theta_t),
\qquad
\mathbb{E}\big[\|g_{t,m}-\nabla F(\theta_t)\|^2\mid \theta_t\big]\le \sigma^2.
\label{eq:unbiased_var}
\end{equation}
\end{assumption}

\begin{assumption}[Conditional Independence of Freezing]
\label{assump:mask_independence}
Conditioned on $\theta_t$, the mask and the stochastic gradient are independent:
\begin{equation}
\mathbb{I}_{t,m}\ \perp\ g_{t,m}\ \mid \theta_t
\qquad\Longleftrightarrow\qquad
\mathbb{U}_{t,m}\ \perp\ g_{t,m}\ \mid \theta_t.
\label{eq:mask_indep}
\end{equation}
\end{assumption}

\begin{assumption}[Bounded Freeze Ratio]
Let the per-parameter \emph{update probability} be
\label{assump:bounded_freeze_ratio}
\begin{equation}
p_{t,m}^{(j)} := \mathbb{E}\big[\mathbb{U}_{t,m}^{(j)}\mid \theta_t\big]
= 1-\mathbb{E}\big[\mathbb{I}_{t,m}^{(j)}\mid \theta_t\big].
\end{equation}
Assume it is uniformly lower bounded by
\begin{equation}
 p_{\min} := 1-r_{\max} \le p_{t,m}^{(j)} ,\quad \forall t,m,j,
\label{eq:pmin}
\end{equation}
where $r_{max}\in[0,1]$ is a user-specified \emph{maximum freeze ratio} as first mentioned in \autoref{eq:max_freeze_ratio}. 
\end{assumption}

\begin{assumption}[Conditional Independence across Microbatches]    
\label{assump:microbatch_independence}
Conditioned on $\theta_t$, $\{(\mathbb{I}_{t,m},g_{t,m})\}_{m=1}^M$ are independent across $m$.
\end{assumption}

\subsubsection{Auxiliary Definitions}
\label{app:aux_definitions}

\begin{definition}[Average Update Probability]
\label{def:avg_update_p}
    \begin{equation}
\bar p_t := \frac{1}{M}\sum_{m=1}^M p_{t,m}.
\label{eq:def_avg_p}
\end{equation}
This quantity averages per-microbatch update probabilities over $M$ microbatches.
\end{definition}

\begin{definition}[Effective Update Probability]
\label{def:peff}
We define the \emph{effective update probability} at step $t$ as
\begin{equation}
p_{\mathrm{eff},t}
\;:=\;
\frac{\sum_{j=1}^N \bar p_t^{(j)} \big(\partial_j F(\theta_t)\big)^2}
{\|\nabla F(\theta_t)\|^2},
\label{eq:def_peff}
\end{equation}
whenever $\nabla F(\theta_t)\neq 0$ (and set $p_{\mathrm{eff},t}:=1$ otherwise).
This quantity captures the fraction of gradient energy that is effectively updated
under freezing.
\end{definition}

\begin{definition}[Average Effective Update Probability]
\label{def:peff_bar}
For a given training horizon $T$, we define
\begin{equation}
\bar p_{\mathrm{eff}}
\;:=\;
\frac{\sum_{t=1}^T \mathbb{E}\!\left[p_{\mathrm{eff},t}\,\|\nabla F(\theta_t)\|^2\right]}
{\sum_{t=1}^T \mathbb{E}\!\left[\|\nabla F(\theta_t)\|^2\right]}.
\label{eq:def_peff_bar}
\end{equation}
This quantity summarizes the effective update probabilities over the entire
training horizon.
\end{definition}

\subsubsection{Lemmas}
\label{app:lemmas}

\begin{lemma}[Mean of Masked Microbatch Gradient]
\label{lem:mean_masked}
Under Assumptions~\ref{assump:unbiased_bounded_var} and
\ref{assump:mask_independence}, for each $m$,
\begin{equation}
\mathbb{E}\big[\mathbb{U}_{t,m}\odot g_{t,m}\mid \theta_t\big]
=
\mathbb{E}\big[\mathbb{U}_{t,m}\mid \theta_t\big]\odot \mathbb{E}\big[g_{t,m}\mid \theta_t\big]
=
p_{t,m}\odot \nabla F(\theta_t),
\end{equation}
where $p_{t,m}\in[0,1]^N$ stacks $\{p_{t,m}^{(j)}\}_j$.
Hence, using \autoref{eq:update_rule} and Definition~\ref{def:avg_update_p},
\begin{equation}
\mathbb{E}[\Delta_t\mid \theta_t]
=
\bar p_t \odot \nabla F(\theta_t).
\end{equation}
\end{lemma}


\begin{lemma}[Descent Inner Product (Worst-case Bound)]
\label{lem:inner}
Under Assumptions~\ref{assump:mask_independence} and
\ref{assump:bounded_freeze_ratio},
\begin{equation}
\Big\langle \nabla F(\theta_t),\, \mathbb{E}[\Delta_t\mid \theta_t]\Big\rangle
=
\sum_{j=1}^N \bar p_t^{(j)} \big(\partial_j F(\theta_t)\big)^2
\ \ge\
p_{\min}\,\|\nabla F(\theta_t)\|^2.
\label{eq:inner_lower}
\end{equation}
\end{lemma}


\begin{lemma}[Descent Inner Product (Exact Characterization)]
\label{lem:inner_peff}
Using Definition~\ref{def:peff}, we have
\begin{equation}
\Big\langle \nabla F(\theta_t),\, \mathbb{E}[\Delta_t\mid \theta_t]\Big\rangle
=
p_{\mathrm{eff},t}\,\|\nabla F(\theta_t)\|^2,
\end{equation}
and hence $p_{\mathrm{eff},t}\in[p_{\min},1]$.
\end{lemma}

\begin{lemma}[Second Moment of Masked Microbatch Average]
\label{lem:second}
Under Assumptions~\ref{assump:unbiased_bounded_var} and \ref{assump:microbatch_independence},
\begin{equation}
\mathbb{E}\big[\|\Delta_t\|^2\mid \theta_t\big]
\le
\Big(1+\frac{1}{M}\Big)\|\nabla F(\theta_t)\|^2 + \frac{\sigma^2}{M}.
\label{eq:second_moment}
\end{equation}
\end{lemma}

\noindent\emph{Proof Sketch.}
Using the second-moment decomposition
$\mathbb{E}\|X\|^2 = \|\mathbb{E}X\|^2 + \mathrm{Var}(X)$ and Lemma~\ref{lem:mean_masked},
\[
\|\mathbb{E}[\Delta_t\mid\theta_t]\|^2 = \|\bar p_t\odot \nabla F(\theta_t)\|^2 \le \|\nabla F(\theta_t)\|^2.
\]
By conditional independence across $m$,
\[
\mathrm{Var}(\Delta_t\mid\theta_t)
=
\frac{1}{M^2}\sum_{m=1}^M \mathrm{Var}(\mathbb{U}_{t,m}\odot g_{t,m}\mid\theta_t)
\le
\frac{1}{M^2}\sum_{m=1}^M \mathbb{E}\big[\|\mathbb{U}_{t,m}\odot g_{t,m}\|^2\mid\theta_t\big].
\]
Since $\mathbb{U}_{t,m}^{(j)}\in\{0,1\}$, $\|\mathbb{U}_{t,m}\odot g_{t,m}\|^2\le \|g_{t,m}\|^2$,
and by Assumption~\ref{assump:unbiased_bounded_var},
$\mathbb{E}\|g_{t,m}\|^2\le \|\nabla F(\theta_t)\|^2+\sigma^2$.
Combining yields \autoref{eq:second_moment}.


\subsubsection{Iteration Complexity under TimelyFreeze}
\label{appendix:subsec:iteration_complexity_timelyfreeze}

\begin{theorem}[Iteration Complexity under TimelyFreeze]
\label{thm:main_peff}
Based on Assumptions~\ref{assump:l-smoothness}--\ref{assump:microbatch_independence},
the update rule~\autoref{eq:update_rule}, and
Lemmas~\ref{lem:inner_peff} and \ref{lem:second},
fix a training horizon $T\ge 1$ and let $\bar p_{\mathrm{eff}}$ be defined in
Definition~\ref{def:peff_bar}.
If the stepsize satisfies
\begin{equation}
\eta \ \le\ \frac{p_{\min}}{L(1+1/M)}
\;=\;
\frac{1-r_{\max}}{L(1+1/M)},
\label{eq:stepsize_cond_peff}
\end{equation}
then
\begin{equation}
\frac{1}{T}\sum_{t=1}^T \mathbb{E}\|\nabla F(\theta_t)\|^2
\le
\frac{2\big(F(\theta_1)-F^\star\big)}
{\bar p_{\mathrm{eff}}\,\eta\,T}
+
\frac{L\eta}{\bar p_{\mathrm{eff}}}\cdot \frac{\sigma^2}{M},
\label{eq:avg_grad_bound_peff}
\end{equation}
where $F^\star := \inf_\theta F(\theta)$.
\end{theorem}

\begin{proof}
By $L$-smoothness~\autoref{eq:smooth}, for $x=\theta_t$ and
$y=\theta_{t+1}=\theta_t-\eta\Delta_t$, we have
\[
F(\theta_{t+1})
\le
F(\theta_t)
-\eta\langle \nabla F(\theta_t), \Delta_t\rangle
+\frac{L\eta^2}{2}\|\Delta_t\|^2.
\]
Taking conditional expectation given $\theta_t$ and applying
Lemmas~\ref{lem:inner_peff} and \ref{lem:second} yield
\[
\mathbb{E}[F(\theta_{t+1})\mid\theta_t]
\le
F(\theta_t)
-\eta\,p_{\mathrm{eff},t}\|\nabla F(\theta_t)\|^2
+\frac{L\eta^2}{2}
\Big(\big(1+\tfrac{1}{M}\big)\|\nabla F(\theta_t)\|^2
+ \tfrac{\sigma^2}{M}\Big).
\]
Under the stepsize condition~\autoref{eq:stepsize_cond_peff} and since
$p_{\mathrm{eff},t}\ge p_{\min}$ for all $t$, we have
\[
\eta\,p_{\mathrm{eff},t}
-\frac{L\eta^2}{2}\Big(1+\tfrac{1}{M}\Big)
\;\ge\;
\frac{\eta\,p_{\mathrm{eff},t}}{2}.
\]
Therefore,
\[
\mathbb{E}[F(\theta_{t+1})]
\le
\mathbb{E}[F(\theta_t)]
-\frac{\eta}{2}\,
\mathbb{E}\!\left[p_{\mathrm{eff},t}\|\nabla F(\theta_t)\|^2\right]
+\frac{L\eta^2}{2}\cdot\frac{\sigma^2}{M}.
\]
Summing over $t=1,\dots,T$ and using $F(\theta_{T+1})\ge F^\star$ give
\[
\frac{\eta}{2}
\sum_{t=1}^T
\mathbb{E}\!\left[p_{\mathrm{eff},t}\|\nabla F(\theta_t)\|^2\right]
\le
F(\theta_1)-F^\star
+\frac{L\eta^2}{2}\cdot\frac{\sigma^2}{M}\,T.
\]
By Definition~\ref{def:peff_bar},
\[
\sum_{t=1}^T
\mathbb{E}\!\left[p_{\mathrm{eff},t}\|\nabla F(\theta_t)\|^2\right]
=
\bar p_{\mathrm{eff}}
\sum_{t=1}^T
\mathbb{E}\!\left[\|\nabla F(\theta_t)\|^2\right].
\]
Dividing both sides by
$\frac{\eta}{2}\,\bar p_{\mathrm{eff}}\,T$
yields~\autoref{eq:avg_grad_bound_peff}.
\end{proof}

\textbf{Iteration Complexity Extraction.}
Let $B_{\mathrm{ours}}(T)$ denote the right-hand side of
\autoref{eq:avg_grad_bound_peff}.
By Definition~\ref{def:Talg_def}, the iteration complexity
$T_{\varepsilon}^{\mathrm{ours}}$ is the smallest $T$ such that
$B_{\mathrm{ours}}(T)\le \varepsilon$.

Equivalently, if
\begin{equation}
\varepsilon
\;>\;
\frac{L\eta}{\bar p_{\mathrm{eff}}}\cdot \frac{\sigma^2}{M},
\label{eq:TF_noise_floor_cond_peff}
\end{equation}
which ensures that the target accuracy $\varepsilon$ lies above the
\emph{noise floor} induced by stochastic gradient variance, then $T_{\varepsilon}^{\mathrm{ours}}$ admits the closed-form bound
\begin{equation}
T_{\varepsilon}^{\mathrm{ours}}
\;=\;
\left\lceil
\frac{2\big(F(\theta_1)-F^\star\big)}
{\bar p_{\mathrm{eff}}\,\eta
\left(
\varepsilon-\frac{L\eta}{\bar p_{\mathrm{eff}}}\cdot\frac{\sigma^2}{M}
\right)}
\right\rceil.
\label{eq:Tours_closed_peff}
\end{equation}

In the noise-free case $\sigma=0$, the noise-floor term vanishes and
\autoref{eq:Tours_closed_peff} simplifies to
\begin{equation}
T_{\varepsilon}^{\mathrm{ours}}
\;=\;
\left\lceil
\frac{C}{\bar p_{\mathrm{eff}} \;\varepsilon}
\right\rceil,
\label{eq:Tours_clean_peff}
\end{equation}
where
\(
C := \frac{2(F(\theta_1)-F^\star)}{\eta}
\)
is independent of freezing.



\subsubsection{Comparison with No-Freezing Baseline}
\label{appendix:subsec:iteration_complexity_baseline}

For the no-freezing baseline, define $\bar g_t := \frac{1}{M}\sum_{m=1}^M g_{t,m}$ and
\begin{equation}
\theta_{t+1} = \theta_t - \eta\,\bar g_t.
\end{equation}
Repeating the same argument (with $\mathbb{U}_{t,m}\equiv \mathbf{1}$ and $r_{\max}=0$) yields, under
$\eta \le \frac{1}{L(1+1/M)}$,
\begin{equation}
\frac{1}{T}\sum_{t=1}^T \mathbb{E}\|\nabla F(\theta_t)\|^2
\le
\frac{2\big(F(\theta_1)-F^\star\big)}{\eta\,T}
+
L\eta\cdot \frac{\sigma^2}{M}.
\label{eq:avg_grad_bound_base}
\end{equation}
Accordingly, for $\varepsilon$ satisfying 
\begin{equation}
\varepsilon \;>\; {L\eta}\cdot \frac{\sigma^2}{M},
\label{eq:base_noise_floor_cond}
\end{equation}
$T_{\varepsilon}^{\mathrm{base}}$ can be expressed as
\begin{equation}
T_{\varepsilon}^{\mathrm{base}}
\;=\;
\left\lceil
\frac{2\big(F(\theta_1)-F^\star\big)}
{\eta\left(\varepsilon-{L\eta}\cdot\frac{\sigma^2}{M}\right)}
\right\rceil,
\label{eq:Tbase_closed}
\end{equation}
which can be simplified to 
\begin{equation}
T_{\varepsilon}^{\mathrm{ours}}
\;=\;
\left\lceil{\frac{C}{\varepsilon}}
\right\rceil,
\label{eq:Tbase_clean}
\end{equation}
in the noise-free case $\sigma=0$. (Note that $C = \frac{2(F(\theta_1)-F^\star)}{\eta}$.)

\begin{corollary}[Freeze-ratio Scaling]
\label{cor:ratio}
In the noise-free case $\sigma=0$, combining
\autoref{eq:Tours_clean_peff} and \autoref{eq:Tbase_clean} yields
\begin{equation}
T_\varepsilon^{\mathrm{ours}}
\;\approx\;
\frac{1}{\bar p_{\mathrm{eff}}}\,
T_\varepsilon^{\mathrm{base}},
\label{eq:ratio_clean}
\end{equation}
where $\bar p_{\mathrm{eff}}\in[1-r_{\max},\,1]$ denotes the average effective
update probability defined in Definition~\ref{def:peff_bar}.

That is, compared to the no-freezing baseline,
TimelyFreeze increases the iteration complexity by a factor inversely
proportional to the fraction of gradient energy that is effectively updated.
In the worst case, this recovers the scaling
$T_\varepsilon^{\mathrm{ours}} = \Theta\!\big(\frac{1}{1-r_{\max}}
T_\varepsilon^{\mathrm{base}}\big)$,
while in typical regimes where $\bar p_{\mathrm{eff}} \gg 1-r_{\max}$,
the resulting iteration complexity remains significantly tighter than
the worst-case bound implied by $r_{\max}$.

More generally, when $\sigma>0$, the above scaling continues to hold for the
optimization term, up to the standard noise-floor effects in
\autoref{eq:avg_grad_bound_peff} and \autoref{eq:avg_grad_bound_base}.

\end{corollary}

\subsection{Execution Time Analysis}
\label{appendix:subsec:execution_time_tau_analysis}

We now analyze how parameter freezing affects the per-step execution time $\tau_t^{\mathrm{ours}}$ under
TimelyFreeze, and relate the resulting speedup explicitly to the maximum
freeze ratio $r_{\max}$.

\textbf{Remark on Progressive Freezing.}
TimelyFreeze employs a \emph{progressive freezing} strategy, under which the
per-step execution time $\tau_t^{\mathrm{ours}}$ may vary during the transition
phase.
For analytical clarity, we focus on the \emph{stable freezing regime}, where the
freeze ratios have converged and the pipeline schedule becomes stationary, so
that the per-step execution time can be well approximated by a
constant $\tau^{\mathrm{ours}}$ (i.e.,
$\tau_t^{\mathrm{ours}} \approx \tau^{\mathrm{ours}}$ for all $t$).
This assumption is also practically justified, as the stable freezing regime dominates the
overall training time in long-running jobs (see \autoref{fig:thp_fr}).

\textbf{DAG-based Per-step Time Representation.}
Recall that the per-step execution time is characterized by the makespan of the
pipeline DAG (see \autoref{subsec:pipeline_dag_construction}).
We therefore identify the per-step execution time with the critical-path length
of the end-to-end graph,
\begin{equation}
\tau(\mathbf{r}) \;\equiv\; P_d(\mathbf{r}),
\end{equation}
where $P_d(\mathbf{r})$ denotes the longest-path value (i.e., the start time of
the destination node) under a freezing decision $\mathbf{r}$.

\textbf{Upper and Lower Makespan Envelopes.}
Given the measured execution-time bounds $w_i^{\min}$ and $w_i^{\max}$ for each
action block $i$, we define the corresponding makespan envelopes as
\begin{equation}
P_d^{\min} := P_d(\mathbf{1}),
\qquad
P_d^{\max} := P_d(\mathbf{0}),
\end{equation}
where $P_d(\cdot)$ denotes the longest-path operator on the DAG.
Here, $P_d(\mathbf{1})$ corresponds to the case where all parameters are frozen
($r_i=1$ and $w_i=w_i^{\min}$ for all action nodes $i$), while
$P_d(\mathbf{0})$ corresponds to the no-freezing case
($r_i=0$ and $w_i=w_i^{\max}$ for all action nodes $i$).
Accordingly, the no-freezing baseline satisfies
\begin{equation}
\tau^{\mathrm{base}} = P_d^{\max}.
\label{eq:tau_base_is_Pdmax}
\end{equation}

Since freezing primarily reduces the backward-related portion of execution time,
we approximate the achievable makespan reduction as scaling linearly with the
maximum freeze ratio $r_{\max}$.
Specifically, we interpolate between the two makespan envelopes as
\begin{equation}
\tau^{\mathrm{ours}}
\;\approx\;
P_d^{\max} - r_{\max}\big(P_d^{\max}-P_d^{\min}\big)
\;=\;
P_d^{\min} + (1-r_{\max})\big(P_d^{\max}-P_d^{\min}\big).
\label{eq:tau_ours_interp}
\end{equation}
This approximation is most accurate when (i) the critical path is dominated by
backward computation and (ii) the identity of the bottleneck path does not change
significantly under freezing.

\subsubsection{Comparison with No-Freezing Baseline}
Combining the baseline \autoref{eq:tau_base_is_Pdmax} and TimelyFreeze \autoref{eq:tau_ours_interp}, we obtain
\begin{equation}
\frac{\tau^{\mathrm{ours}}}{\tau^{\mathrm{base}}}
\;\approx\;
\frac{P_d^{\min} + (1-r_{\max})(P_d^{\max}-P_d^{\min})}{P_d^{\max}}.
\label{eq:tau_ratio_raw}
\end{equation}
For simplicity, we define the time-reduction factor
\begin{equation}
\kappa
\;:=\;
\frac{\tau^{\mathrm{ours}}}{\tau^{\mathrm{base}}}
\;\approx\;
(1-r_{\max})
+
r_{\max}\,\frac{P_d^{\min}}{P_d^{\max}},
\label{eq:kappa_def_pd}
\end{equation}
where $\kappa\in(0,1]$.
This expression makes explicit how the achievable per-step speedup depends on
both the maximum freeze ratio $r_{\max}$ and the gap between the upper and lower
makespan envelopes. 
Intuitively, when backward computation dominates the critical path,
the per-step execution time decreases nearly proportionally to $1-r_{\max}$,
whereas the irreducible forward and pipeline overheads limit the speedup
as $P_d^{\min}/P_d^{\max}$ increases.


\subsection{Time-to-Accuracy Comparison}
\label{appendix:subsec:tta_comparison}

We combine the iteration complexity analysis
(\autoref{appendix:convergence_analysis}) and the per-step execution time
analysis (\autoref{appendix:subsec:execution_time_tau_analysis}) to characterize
the end-to-end \emph{time-to-accuracy} (TTA).

In the stable freezing regime, the per-step execution time can be approximated
by a constant $\tau^{\mathrm{alg}}$, with
$\tau^{\mathrm{base}} = P_d^{\max}$ for the no-freezing baseline.
Accordingly, the TTA definition in \autoref{eq:def_TTA} admits the approximation
\begin{equation}
\mathrm{TTA}_\varepsilon^{\mathrm{alg}}
\;\approx\;
T_\varepsilon^{\mathrm{alg}} \cdot \bar\tau^{\mathrm{alg}},
\label{eq:tta_approx}
\end{equation}
where $T_\varepsilon^{\mathrm{alg}}$ denotes a sufficient number of training
steps given by the convergence bound.

\begin{theorem}[Noise-free TTA Comparison]
\label{thm:tta_noise_free}
In the noise-free case $\sigma=0$, under the stable-regime approximation
\autoref{eq:tta_approx}, suppose the iteration complexity and per-step execution
time satisfy
\begin{equation}
T_\varepsilon^{\mathrm{ours}}
\;\approx\;
\frac{1}{\bar p_{\mathrm{eff}}}\,T_\varepsilon^{\mathrm{base}}
\qquad
\text{(cf.~Corollary~\ref{cor:ratio})},
\label{eq:iter_ratio_peff_ref}
\end{equation}
and
\begin{equation}
\kappa
:= \frac{\tau^{\mathrm{ours}}}{\tau^{\mathrm{base}}}
\;\approx\;
(1-r_{\max}) + r_{\max}\frac{P_d^{\min}}{P_d^{\max}}
\qquad
\text{(cf.~\autoref{eq:kappa_def_pd})}.
\label{eq:kappa_ref}
\end{equation}
Then the time-to-accuracy ratio admits the approximation
\begin{equation}
\frac{\mathrm{TTA}^{\mathrm{ours}}_\varepsilon}
{\mathrm{TTA}^{\mathrm{base}}_\varepsilon}
\;\approx\;
\frac{\kappa}{\bar p_{\mathrm{eff}}}.
\label{eq:tta_ratio_clean}
\end{equation}
In particular, if
\begin{equation}
\kappa < \bar p_{\mathrm{eff}},
\label{eq:tta_improve_cond}
\end{equation}
then TimelyFreeze strictly improves wall-clock time-to-accuracy, i.e.,
\[
\mathrm{TTA}^{\mathrm{ours}}_\varepsilon
<
\mathrm{TTA}^{\mathrm{base}}_\varepsilon.
\]
\end{theorem}

\begin{proof}
By the approximation $\mathrm{TTA}_\varepsilon^{\mathrm{alg}}
\approx T_\varepsilon^{\mathrm{alg}}\bar\tau^{\mathrm{alg}}$,
\[
\frac{\mathrm{TTA}^{\mathrm{ours}}_\varepsilon}
{\mathrm{TTA}^{\mathrm{base}}_\varepsilon}
\;\approx\;
\frac{T_\varepsilon^{\mathrm{ours}}}{T_\varepsilon^{\mathrm{base}}}
\cdot
\frac{\tau^{\mathrm{ours}}}{\tau^{\mathrm{base}}}
\;\approx\;
\frac{1}{\bar p_{\mathrm{eff}}}\cdot \kappa.
\]
The improvement condition \autoref{eq:tta_improve_cond} immediately follows. 
\end{proof}

\clearpage
\section{Experiment Results}
\subsection{Experiment Setting Details}
\label{appendix:exp_setting}
The implementation code is available at \url{https://anonymous.4open.science/r/TimelyFreeze/}. 
\begin{table*}[th]
\renewcommand{\arraystretch}{1.1}
\centering
\caption{Experimental setup for language and vision finetuning tasks.}
\resizebox{\textwidth}{!}{
\rowcolors{3}{white}{lightgray!10}
\begin{tabular}{l|c|c|c|c|c}
\toprule
 &
\multicolumn{3}{c|}{\textbf{Language Tasks}} &
\multicolumn{2}{c}{\textbf{Vision Tasks}} \\ 
\midrule
Model & LLaMA-3.2-1B & LLaMA-3-8B & LLaMA-2-13B & ConvNeXt-V2-L & ViT-L/32 \\ 
\midrule
Dataset & Alpaca-GPT4 & OpenHermes-2.5 & OpenHermes-2.5 & Food-101 & ImageNet-1K \\
GPU & 4 $\times$ A6000 & 4 $\times$ H200 & 4 $\times$ H200 & 4 $\times$ RTX 3090 & 8 $\times$ RTX 3090 \\
Seed & 42, 22 , 11 & 42, 22 , 11 & 42 & 42 & 42 \\
Training Steps & 800 & 2,000 & 2,000 & 20,000 & 17,500 \\
Warmup Steps & 100 & 200 & 200 & 2,500 & 1,500 \\
Optimizer & AdamW & AdamW & AdamW & AdamW & SGD \\
Learning Rate & 5.0e-6 & 3.0e-6 & 3.0e-6 & 1e-4 & 3e-3 \\
LR Scheduler & \multicolumn{3}{c|}{Cosine Annealing} & \multicolumn{2}{c}{Cosine Annealing} \\
Global Batch Size & 128 & 64 & 64 & 64 & 512 \\
Local Batch Size & 16 & 16 & 16 & 64 & 512 \\
Seq. Length / Image Size & 1024 & 1024 & 1024 & 224 $\times$ 224 & 224 $\times$ 224 \\
Pipeline Parallel Degree & 4 & 4 & 4 & 4 & 8 \\
Num. Microbatches & 8 & 8 & 8 & 8 & 8 \\
Phase Boundaries $T_w, T_m, T_f$ & 60,\;100,\;200 & 160,\;200,\;250 & 150,\;200,\;250 & 2350,\;2850,\;5600 & 1400,\;1600,\;2400 \\
Max Freeze Ratio $r_{\max}$ & 0.8 & 0.8 & 0.8 & 0.5 & 0.8 \\
APF Threshold $T_{APF}$ & 0.01 & 0.0001 & 0.0001 & 0.005 & 0.005 \\
AutoFreeze Percentage $P_{Auto}$ & 80\% & 80\% & 80\% & 80\% & 80\% \\
\bottomrule
\end{tabular}}
\label{tab:exp_setup}
\end{table*}

\subsection{LLaMA 1B Result Table}
\label{appendix:llama1b}
\begin{table*}[th]
\renewcommand{\arraystretch}{1.14}
\centering
\caption{
Comparison of freezing methods across different pipeline schedules (Llama-3.2-1B).
The average accuracy of the pretrained (no fine-tuning) baseline is 34.93.
The best and second-best values among the five freezing methods (APF, AutoFreeze, \TimelyFreeze{}, and two hybrid variants) are highlighted in bold and underline, respectively.
}
\setlength{\tabcolsep}{3pt}
\vspace*{-0.2cm}

\resizebox{\textwidth}{!}{%
\begin{tabular}{l|cc|cc c l|cc|cc}
\multicolumn{5}{c}{\textbf{GPipe}} & \cellcolor{white}\quad\quad &\multicolumn{5}{c}{\textbf{1F1B}} \\
\cmidrule[\heavyrulewidth]{1-5}\cmidrule[\heavyrulewidth]{7-11}

\multicolumn{1}{c|}{\multirow{2}{*}{\makecell[c]{Freeze \\ Method}}} &
\multicolumn{2}{c|}{Accuracy Preservation} &
\multicolumn{2}{c}{Time Efficiency}
& \cellcolor{white}\quad\quad &
\multicolumn{1}{c|}{\multirow{2}{*}{\makecell[c]{Freeze \\ Method}}} &
\multicolumn{2}{c|}{Accuracy Preservation} &
\multicolumn{2}{c}{Time Efficiency} \\
\cmidrule{2-5}\cmidrule{8-11}
&
\makecell[c]{Avg.\,Acc.\,($\Delta$)$\uparrow$} &
Frz.\,Ratio &
\makecell[c]{Throughput\,($\Delta$)$\uparrow$} &
MFU$\uparrow$
& \cellcolor{white}\quad\quad &
&
\makecell[c]{Avg.\,Acc.\,($\Delta$)$\uparrow$} &
Frz.\,Ratio &
\makecell[c]{Throughput\,($\Delta$)$\uparrow$} &
MFU$\uparrow$ \\
\cmidrule{1-5}\cmidrule{7-11}

No Freezing
& \aaccsign{36.70}{+0.00} & 0.00 & \ttimesign{6965}{0.00} & 16.05
& \cellcolor{white}\quad\quad &
No Freezing
& \aaccsign{35.91}{+0.00} & 0.00 & \ttimesign{7152}{0.00} & 16.48 \\
\graycline{1-5}\graycline{7-11}

APF
& \aaccsign{35.86}{-0.84} & 35.08 & \ttimesign{8183}{17.49} & 18.81
& \cellcolor{white}\quad\quad &
APF
& \aaccsign{\underline{35.73}}{-0.18} & 35.06 & \ttimesign{8223}{14.98} & 18.91 \\
AutoFreeze
& \aaccsign{35.78}{-0.92} & 24.91 & \ttimesign{7536}{8.20} & 17.34
& \cellcolor{white}\quad\quad &
AutoFreeze
& \aaccsign{34.51}{-1.41} & 26.47 & \ttimesign{7624}{6.60} & 17.54 \\

\rowcolor{gray!15}
\TimelyFreeze{}
& \aaccsign{\textbf{36.64}}{-0.06} & 19.97 & \ttimesign{\underline{8821}}{26.64} & \underline{20.27}
& \cellcolor{white}\quad\quad &
\TimelyFreeze{}
& \aaccsign{\textbf{36.36}}{+0.45} & 15.53 & \ttimesign{\textbf{9257}}{29.44} & \textbf{21.31} \\

\rowcolor{lightgray!10}
\scriptsize\quad +APF
& \aaccsign{36.18}{-0.53} & 20.19 & \ttimesign{\textbf{8860}}{27.21} & \textbf{20.41}
& \cellcolor{white}\quad\quad &
\scriptsize\quad +APF
& \aaccsign{35.69}{-0.23} & 15.36 & \ttimesign{9161}{28.09} & 21.10 \\

\rowcolor{lightgray!10}
\scriptsize\quad +AutoFreeze
& \aaccsign{\underline{36.39}}{-0.32} & 19.49 & \ttimesign{8809}{26.47} & 20.28
& \cellcolor{white}\quad\quad &
\scriptsize\quad +AutoFreeze
& \aaccsign{35.09}{-0.83} & 15.53 & \ttimesign{\underline{9184}}{28.42} & \underline{21.16} \\
\cmidrule[\heavyrulewidth]{1-5}\cmidrule[\heavyrulewidth]{7-11}

\multicolumn{5}{c}{\textbf{Interleaved 1F1B}} & \cellcolor{white}\quad\quad & \multicolumn{5}{c}{\textbf{ZBV}} \\
\cmidrule[\heavyrulewidth]{1-5}\cmidrule[\heavyrulewidth]{7-11}

\multicolumn{1}{c|}{\multirow{2}{*}{\makecell[c]{Freeze \\ Method}}} &
\multicolumn{2}{c|}{Accuracy Preservation} &
\multicolumn{2}{c}{Time Efficiency}
& \cellcolor{white}\quad\quad &
\multicolumn{1}{c|}{\multirow{2}{*}{\makecell[c]{Freeze \\ Method}}} &
\multicolumn{2}{c|}{Accuracy Preservation} &
\multicolumn{2}{c}{Time Efficiency} \\
\cmidrule{2-5}\cmidrule{8-11}
&
\makecell[c]{Avg.\,Acc.\,($\Delta$)$\uparrow$} &
Frz.\,Ratio &
\makecell[c]{Throughput\,($\Delta$)$\uparrow$} &
MFU$\uparrow$
& \cellcolor{white}\quad\quad &
&
\makecell[c]{Avg.\,Acc.\,($\Delta$)$\uparrow$} &
Frz.\,Ratio &
\makecell[c]{Throughput\,($\Delta$)$\uparrow$} &
MFU$\uparrow$ \\
\cmidrule{1-5}\cmidrule{7-11}

No Freezing
& \aaccsign{36.42}{+0.00} & 0.00 & \ttimesign{6800}{0.00} & 15.68
& \cellcolor{white}\quad\quad &
No Freezing
& \aaccsign{36.57}{+0.00} & 0.00 & \ttimesign{7541}{0.00} & 17.31 \\
\graycline{1-5}\graycline{7-11}

APF
& \aaccsign{\underline{35.73}}{-0.69} & 69.45 & \ttimesign{7564}{11.24} & 17.40
& \cellcolor{white}\quad\quad &
APF
& \aaccsign{\underline{35.59}}{-0.98} & 65.64 & \ttimesign{9372}{24.29} & 21.50 \\
AutoFreeze
& \aaccsign{34.25}{-2.16} & 45.96 & \ttimesign{7112}{4.60} & 16.38
& \cellcolor{white}\quad\quad &
AutoFreeze
& \aaccsign{35.32}{-1.25} & 26.93 & \ttimesign{7906}{4.85} & 18.13 \\

\rowcolor{gray!15}
\TimelyFreeze{}
& \aaccsign{35.52}{-0.90} & 37.89 & \ttimesign{\underline{8355}}{22.87} & \underline{19.23}
& \cellcolor{white}\quad\quad &
\TimelyFreeze{}
& \aaccsign{\textbf{35.95}}{-0.62} & 39.65 & \ttimesign{\textbf{9426}}{25.01} & \textbf{21.63} \\

\rowcolor{lightgray!10}
\scriptsize\quad +APF
& \aaccsign{35.36}{-1.06} & 38.50 & \ttimesign{8277}{21.73} & 19.06
& \cellcolor{white}\quad\quad &
\scriptsize\quad +APF
& \aaccsign{35.16}{-1.42} & 38.41 & \ttimesign{9284}{23.12} & 21.32 \\

\rowcolor{lightgray!10}
\scriptsize\quad +AutoFreeze
& \aaccsign{\textbf{36.13}}{-0.28} & 38.50 & \ttimesign{\textbf{8366}}{23.04} & \textbf{19.25}
& \cellcolor{white}\quad\quad &
\scriptsize\quad +AutoFreeze
& \aaccsign{\underline{35.59}}{-0.98} & 39.28 & \ttimesign{\underline{9402}}{24.68} & \underline{21.58} \\

\cmidrule[\heavyrulewidth]{1-5}\cmidrule[\heavyrulewidth]{7-11}
\end{tabular}%
}

\label{tab:llama1b}
\end{table*}

\begin{table*}[th]
\renewcommand{\arraystretch}{1.14}
\centering
\caption{
Comparison of freezing methods across different pipeline schedules (LLaMA-2-13B).
The average accuracy of the pretrained (no fine-tuning) baseline is 46.73.
The best and second-best values among the five freezing methods (APF, AutoFreeze, \TimelyFreeze{}, and two hybrid variants) are highlighted in bold and underline, respectively.
}
\setlength{\tabcolsep}{3pt}
\vspace*{-0.2cm}

\resizebox{\textwidth}{!}{%
\begin{tabular}{l|cc|cc c l|cc|cc}
\multicolumn{5}{c}{\textbf{GPipe}} & \cellcolor{white}\quad\quad &\multicolumn{5}{c}{\textbf{1F1B}} \\
\cmidrule[\heavyrulewidth]{1-5}\cmidrule[\heavyrulewidth]{7-11}

\multicolumn{1}{c|}{\multirow{2}{*}{\makecell[c]{Freeze \\ Method}}} &
\multicolumn{2}{c|}{Accuracy Preservation} &
\multicolumn{2}{c}{Time Efficiency}
& \cellcolor{white}\quad\quad &
\multicolumn{1}{c|}{\multirow{2}{*}{\makecell[c]{Freeze \\ Method}}} &
\multicolumn{2}{c|}{Accuracy Preservation} &
\multicolumn{2}{c}{Time Efficiency} \\
\cmidrule{2-5}\cmidrule{8-11}
&
\makecell[c]{Avg.\,Acc.\,($\Delta$)$\uparrow$} &
Frz.\,Ratio &
\makecell[c]{Throughput\,($\Delta$)$\uparrow$} &
MFU$\uparrow$
& \cellcolor{white}\quad\quad &
&
\makecell[c]{Avg.\,Acc.\,($\Delta$)$\uparrow$} &
Frz.\,Ratio &
\makecell[c]{Throughput\,($\Delta$)$\uparrow$} &
MFU$\uparrow$ \\
\cmidrule{1-5}\cmidrule{7-11}

No Freezing
& \aaccsign{51.09}{+0.00} & 0.00 & \ttimesign{3550}{0.00} & 27.40
& \cellcolor{white}\quad\quad &
No Freezing
& \aaccsign{51.15}{0.00} & 0.00 & \ttimesign{3917}{0.00} & 30.23 \\
\graycline{1-5}\graycline{7-11}

APF
& \aaccsign{50.14}{-0.94} & 60.75 & \ttimesign{4819}{35.73} & 37.65
& \cellcolor{white}\quad\quad &
APF
& \aaccsign{50.31}{-0.83} & 60.80 & \ttimesign{5378}{37.31} & 42.08 \\
AutoFreeze
& \aaccsign{50.23}{-0.86} & 84.45 & \ttimesign{4989}{40.52} & 39.11
& \cellcolor{white}\quad\quad &
AutoFreeze
& \aaccsign{50.21}{-0.94} & 87.64 & \ttimesign{\underline{5660}}{44.51} & \underline{44.26} \\

\rowcolor{gray!15}
\TimelyFreeze{}
& \aaccsign{\underline{50.81}}{-0.28} & 72.18 & \ttimesign{5132}{44.54} & 40.13
& \cellcolor{white}\quad\quad &
\TimelyFreeze{}
& \aaccsign{\textbf{51.06}}{-0.09} & 69.33 & \ttimesign{5590}{42.72} & 43.68 \\

\rowcolor{lightgray!10}
\scriptsize\quad +APF
& \aaccsign{\textbf{51.00}}{-0.09} & 72.24 & \ttimesign{\textbf{5199}}{46.44} & \textbf{40.70}
& \cellcolor{white}\quad\quad &
\scriptsize\quad +APF
& \aaccsign{\underline{50.59}}{-0.55} & 69.68 & \ttimesign{\textbf{5743}}{46.63} & \textbf{44.97} \\

\rowcolor{lightgray!10}
\scriptsize\quad +AutoFreeze
& \aaccsign{\underline{50.81}}{-0.28} & 72.21 & \ttimesign{\underline{5138}}{44.73} & \underline{40.19}
& \cellcolor{white}\quad\quad &
\scriptsize\quad +AutoFreeze
& \aaccsign{50.52}{-0.63} & 69.65 & \ttimesign{5584}{42.57} & 43.62 \\
\cmidrule[\heavyrulewidth]{1-5}\cmidrule[\heavyrulewidth]{7-11}

\multicolumn{5}{c}{\textbf{Interleaved 1F1B}} & \cellcolor{white}\quad\quad & \multicolumn{5}{c}{\textbf{ZBV}} \\
\cmidrule[\heavyrulewidth]{1-5}\cmidrule[\heavyrulewidth]{7-11}

\multicolumn{1}{c|}{\multirow{2}{*}{\makecell[c]{Freeze \\ Method}}} &
\multicolumn{2}{c|}{Accuracy Preservation} &
\multicolumn{2}{c}{Time Efficiency}
& \cellcolor{white}\quad\quad &
\multicolumn{1}{c|}{\multirow{2}{*}{\makecell[c]{Freeze \\ Method}}} &
\multicolumn{2}{c|}{Accuracy Preservation} &
\multicolumn{2}{c}{Time Efficiency} \\
\cmidrule{2-5}\cmidrule{8-11}
&
\makecell[c]{Avg.\,Acc.\,($\Delta$)$\uparrow$} &
Frz.\,Ratio &
\makecell[c]{Throughput\,($\Delta$)$\uparrow$} &
MFU$\uparrow$
& \cellcolor{white}\quad\quad &
&
\makecell[c]{Avg.\,Acc.\,($\Delta$)$\uparrow$} &
Frz.\,Ratio &
\makecell[c]{Throughput\,($\Delta$)$\uparrow$} &
MFU$\uparrow$ \\
\cmidrule{1-5}\cmidrule{7-11}

No Freezing
& \aaccsign{50.90}{0.00} & 0.00 & \ttimesign{4271}{0.00} & 32.97
& \cellcolor{white}\quad\quad &
No Freezing
& \aaccsign{50.22}{0.00} & 0.00 & \ttimesign{4660}{0.00} & 35.97 \\
\graycline{1-5}\graycline{7-11}

APF
& \aaccsign{49.62}{-1.28} & 62.47 & \ttimesign{\underline{5851}}{36.99} & \underline{45.72}
& \cellcolor{white}\quad\quad &
APF
& \aaccsign{49.77}{-0.45} & 45.23 & \ttimesign{\textbf{5945}}{27.58} & \textbf{46.16} \\
AutoFreeze
& \aaccsign{48.58}{-2.32} & 84.46 & \ttimesign{\textbf{5992}}{40.29} & \textbf{46.96}
& \cellcolor{white}\quad\quad &
AutoFreeze
& \aaccsign{\underline{49.91}}{-0.31} & 27.76 & \ttimesign{4888}{4.90} & 37.76 \\

\rowcolor{gray!15}
\TimelyFreeze{}
& \aaccsign{49.71}{-1.19} & 72.22 & \ttimesign{5775}{35.22} & 44.95
& \cellcolor{white}\quad\quad &
\TimelyFreeze{}
& \aaccsign{\textbf{50.36}}{+0.14} & 63.42 & \ttimesign{\underline{5809}}{24.67} & \underline{45.03} \\

\rowcolor{lightgray!10}
\scriptsize\quad +APF
& \aaccsign{\textbf{50.37}}{-0.53} & 72.35 & \ttimesign{5616}{31.50} & 43.65
& \cellcolor{white}\quad\quad &
\scriptsize\quad +APF
& \aaccsign{48.69}{-1.53} & 63.48 & \ttimesign{5796}{24.37} & 44.92 \\

\rowcolor{lightgray!10}
\scriptsize\quad +AutoFreeze
& \aaccsign{\underline{49.87}}{-1.04} & 71.71 & \ttimesign{5600}{31.12} & 43.52
& \cellcolor{white}\quad\quad &
\scriptsize\quad +AutoFreeze
& \aaccsign{49.62}{-0.61} & 63.51 & \ttimesign{5791}{24.28} & 44.89 \\

\cmidrule[\heavyrulewidth]{1-5}\cmidrule[\heavyrulewidth]{7-11}
\end{tabular}%
}

\label{tab:llama13b}
\end{table*}

\pagebreak
\subsection{LLaMA Series Benchmark Results}
\label{appendix:llama1b8b_benchmarks}
\begin{table*}[ht]
\renewcommand{\arraystretch}{1.12}
\centering
\caption{Detailed Benchmark Scores under different pipeline schedules (Llama-3.2-1B).}
\setlength{\tabcolsep}{3pt}

\begin{tabular}{c}
\resizebox{0.49\linewidth}{!}{%
\begin{tabular}{cccc|c}
\multicolumn{5}{c}{\textbf{Base Pretrained Model}} \\
\toprule
\makecell[c]{MMLU\\[-5pt]\scriptsize{(5 shots)}} &
\makecell[c]{HellaSwag\\[-5pt]\scriptsize{(0 shots)}} &
\makecell[c]{ARC-C\\[-5pt]\scriptsize{(10 shots)}} &
\makecell[c]{TruthfulQA\\[-5pt]\scriptsize{(0 shots)}} &
Avg. Acc. $\uparrow$ \\
\midrule
31.99 & 47.75 & 36.95 & 23.01 & 34.93 \\
\bottomrule
\end{tabular}}%
\\[2.5em]
\end{tabular}

\vspace*{0.4em}
\resizebox{\textwidth}{!}{%
\begin{tabular}{l|cccc|c c l|cccc|c}

\multicolumn{6}{c}{\textbf{GPipe}} &
\cellcolor{white}\quad\quad &
\multicolumn{6}{c}{\textbf{1F1B}} \\
\cmidrule[\heavyrulewidth]{1-6}
\cmidrule[\heavyrulewidth]{8-13}

Freeze Method & MMLU & \footnotesize HellaSwag & ARC-C & TruthQA & Avg. &
\cellcolor{white}\quad\quad &
Freeze Method & MMLU & \footnotesize HellaSwag & ARC-C & TruthQA & Avg. \\
\cmidrule{1-6}
\cmidrule{8-13}

No Freezing
& 32.04 & 49.15 & 37.63 & 27.99 & 36.70
& \cellcolor{white}\quad\quad &
No Freezing
& 30.42 & 49.34 & 37.57 & 26.32 & 35.91 \\
\graycline{1-6}\graycline{8-13}

APF
& 32.74 & 48.55 & 37.09 & 25.05 & 35.86
& \cellcolor{white}\quad\quad &
APF
& 32.52 & 48.39 & 36.77 & 25.25 & \underline{35.73} \\
AutoFreeze
& 30.48 & 48.61 & 36.66 & 27.38 & 35.78
& \cellcolor{white}\quad\quad &
AutoFreeze
& 29.15 & 44.61 & 37.63 & 26.65 & 34.51 \\

\rowcolor{gray!15}
\TimelyFreeze{}
& 31.79 & 49.32 & 37.46 & 27.99 & \textbf{36.64}
& \cellcolor{white}\quad\quad &
\TimelyFreeze{}
& 31.30 & 49.49 & 37.80 & 26.85 & \textbf{36.36} \\

\rowcolor{lightgray!10}
\scriptsize\quad +APF
& 31.32 & 49.04 & 37.37 & 26.97 & 36.18
& \cellcolor{white}\quad\quad &
\scriptsize\quad +APF
& 29.66 & 49.26 & 37.51 & 26.32 & 35.69 \\

\rowcolor{lightgray!10}
\scriptsize\quad +AutoFreeze
& 31.56 & 49.25 & 37.23 & 27.50 & \underline{36.39}
& \cellcolor{white}\quad\quad &
\scriptsize\quad +AutoFreeze
& 28.15 & 48.94 & 37.17 & 26.07 & 35.09 \\
\cmidrule[\heavyrulewidth]{1-6}
\cmidrule[\heavyrulewidth]{8-13}

\multicolumn{6}{c}{\textbf{Interleaved 1F1B}} &
\cellcolor{white}\quad\quad &
\multicolumn{6}{c}{\textbf{ZBV}} \\
\cmidrule[\heavyrulewidth]{1-6}
\cmidrule[\heavyrulewidth]{8-13}

Freeze Method & MMLU & \footnotesize HellaSwag & ARC-C & TruthQA & Avg. &
\cellcolor{white}\quad\quad &
Freeze Method & MMLU & \footnotesize HellaSwag & ARC-C & TruthQA & Avg. \\
\cmidrule{1-6}
\cmidrule{8-13}

No Freezing
& 30.90 & 49.66 & 37.97 & 27.13 & 36.42
& \cellcolor{white}\quad\quad &
No Freezing
& 31.57 & 49.72 & 38.03 & 26.97 & 36.57 \\
\graycline{1-6}\graycline{8-13}

APF
& 32.55 & 48.35 & 36.89 & 25.13 & \underline{35.73}
& \cellcolor{white}\quad\quad &
APF
& 32.33 & 48.24 & 37.23 & 24.56 & \underline{35.59} \\
AutoFreeze
& 28.39 & 44.53 & 37.65 & 26.44 & 34.25
& \cellcolor{white}\quad\quad &
AutoFreeze
& 29.10 & 48.08 & 37.63 & 26.48 & 35.32 \\

\rowcolor{gray!15}
\TimelyFreeze{}
& 28.79 & 49.08 & 38.08 & 26.11 & 35.52
& \cellcolor{white}\quad\quad &
\TimelyFreeze{}
& 29.47 & 49.29 & 38.20 & 26.84 & \textbf{35.95} \\

\rowcolor{lightgray!10}
\scriptsize\quad +APF
& 28.61 & 49.06 & 37.85 & 25.91 & 35.36
& \cellcolor{white}\quad\quad &
\scriptsize\quad +APF
& 27.99 & 48.94 & 38.11 & 25.58 & 35.16 \\

\rowcolor{lightgray!10}
\scriptsize\quad +AutoFreeze
& 30.42 & 49.36 & 38.31 & 26.44 & \textbf{36.13}
& \cellcolor{white}\quad\quad &
\scriptsize\quad +AutoFreeze
& 29.09 & 49.14 & 38.00 & 26.15 & \underline{35.59} \\
\cmidrule[\heavyrulewidth]{1-6}
\cmidrule[\heavyrulewidth]{8-13}

\end{tabular}%
}

\label{tab:llama1b_benchmark}
\end{table*}

\begin{table*}[ht]
\renewcommand{\arraystretch}{1.12}
\centering
\caption{Detailed Benchmark Scores under different pipeline schedules (Llama-3-8B).}
\setlength{\tabcolsep}{3pt}

\begin{tabular}{c}
\resizebox{0.49\linewidth}{!}{%
\begin{tabular}{cccc|c}
\multicolumn{5}{c}{\textbf{Base Pretrained Model}} \\
\toprule
\makecell[c]{MMLU\\[-5pt]\scriptsize{(5 shots)}} &
\makecell[c]{HellaSwag\\[-5pt]\scriptsize{(0 shots)}} &
\makecell[c]{ARC-C\\[-5pt]\scriptsize{(10 shots)}} &
\makecell[c]{TruthfulQA\\[-5pt]\scriptsize{(0 shots)}} &
Avg. Acc. $\uparrow$ \\
\midrule
63.55 & 60.00 & 51.28 & 28.40 & 50.81 \\
\bottomrule
\end{tabular}}%
\\[2.5em]
\end{tabular}

\vspace*{0.4em}
\resizebox{\textwidth}{!}{%
\begin{tabular}{l|cccc|c c l|cccc|c}

\multicolumn{6}{c}{\textbf{GPipe}} &
\cellcolor{white}\quad\quad &
\multicolumn{6}{c}{\textbf{1F1B}} \\
\cmidrule[\heavyrulewidth]{1-6}
\cmidrule[\heavyrulewidth]{8-13}

Freeze Method & MMLU & \footnotesize HellaSwag & ARC-C & TruthQA & Avg. &
\cellcolor{white}\quad\quad &
Freeze Method & MMLU & \footnotesize HellaSwag & ARC-C & TruthQA & Avg. \\
\cmidrule{1-6}
\cmidrule{8-13}

No Freezing
& 64.09 & 61.45 & 55.63 & 37.33 & 54.63
& \cellcolor{white}\quad\quad &
No Freezing
& 64.09 & 61.45 & 55.63 & 37.33 & 54.63 \\
\graycline{1-6}\graycline{8-13}

APF
& 64.09 & 61.44 & 55.12 & 37.94 & 54.65
& \cellcolor{white}\quad\quad &
APF
& 64.09 & 61.44 & 55.12 & 37.94 & 54.65 \\
AutoFreeze
& 64.39 & 61.74 & 54.44 & 35.37 & 53.99
& \cellcolor{white}\quad\quad &
AutoFreeze
& 64.39 & 61.74 & 54.44 & 35.37 & 53.99 \\

\rowcolor{gray!15}
\TimelyFreeze{}
& 63.44 & 61.73 & 55.86 & 38.15 & \underline{54.79}
& \cellcolor{white}\quad\quad &
\TimelyFreeze{}
& 64.02 & 61.39 & 55.94 & 37.41 & 54.69 \\

\rowcolor{lightgray!10}
\scriptsize\quad +APF
& 63.27 & 61.60 & 56.45 & 37.94 & \textbf{54.82}
& \cellcolor{white}\quad\quad &
\scriptsize\quad +APF
& 64.36 & 61.51 & 55.58 & 37.90 & \textbf{54.84} \\

\rowcolor{lightgray!10}
\scriptsize\quad +AutoFreeze
& 63.60 & 61.75 & 55.35 & 37.90 & 54.65
& \cellcolor{white}\quad\quad &
\scriptsize\quad +AutoFreeze
& 64.00 & 61.48 & 55.60 & 38.19 & \underline{54.82} \\
\cmidrule[\heavyrulewidth]{1-6}
\cmidrule[\heavyrulewidth]{8-13}

\multicolumn{6}{c}{\textbf{Interleaved 1F1B}} &
\cellcolor{white}\quad\quad &
\multicolumn{6}{c}{\textbf{ZBV}} \\
\cmidrule[\heavyrulewidth]{1-6}
\cmidrule[\heavyrulewidth]{8-13}

Freeze Method & MMLU & \footnotesize HellaSwag & ARC-C & TruthQA & Avg. &
\cellcolor{white}\quad\quad &
Freeze Method & MMLU & \footnotesize HellaSwag & ARC-C & TruthQA & Avg. \\
\cmidrule{1-6}
\cmidrule{8-13}

No Freezing
& 63.83 & 61.08 & 55.89 & 38.31 & 54.78
& \cellcolor{white}\quad\quad &
No Freezing
& 63.91 & 61.36 & 55.55 & 37.82 & 54.66 \\
\graycline{1-6}\graycline{8-13}

APF
& 64.72 & 61.53 & 55.89 & 37.82 & \textbf{54.99}
& \cellcolor{white}\quad\quad &
APF
& 64.56 & 61.65 & 55.55 & 36.11 & 54.47 \\
AutoFreeze
& 64.34 & 61.29 & 55.20 & 37.33 & 54.54
& \cellcolor{white}\quad\quad &
AutoFreeze
& 63.64 & 61.48 & 56.23 & 38.07 & 54.86 \\

\rowcolor{gray!15}
\TimelyFreeze{}
& 64.37 & 61.41 & 55.35 & 37.41 & 54.64
& \cellcolor{white}\quad\quad &
\TimelyFreeze{}
& 63.89 & 61.59 & 55.26 & 38.15 & 54.72 \\

\rowcolor{lightgray!10}
\scriptsize\quad +APF
& 64.29 & 61.56 & 55.63 & 37.37 & 54.71
& \cellcolor{white}\quad\quad &
\scriptsize\quad +APF
& 64.15 & 61.76 & 55.72 & 38.47 & \underline{55.03} \\

\rowcolor{lightgray!10}
\scriptsize\quad +AutoFreeze
& 64.40 & 61.54 & 55.29 & 37.66 & \underline{54.72}
& \cellcolor{white}\quad\quad &
\scriptsize\quad +AutoFreeze
& 63.97 & 61.64 & 56.40 & 38.19 & \textbf{55.05} \\
\cmidrule[\heavyrulewidth]{1-6}
\cmidrule[\heavyrulewidth]{8-13}

\end{tabular}%
}

\label{tab:llama8b_benchmark}
\end{table*}

\begin{table*}[ht]
\renewcommand{\arraystretch}{1.12}
\centering
\caption{Detailed Benchmark Scores under different pipeline schedules (Llama-2-13B).}
\setlength{\tabcolsep}{3pt}

\begin{tabular}{c}
\resizebox{0.49\linewidth}{!}{%
\begin{tabular}{cccc|c}
\multicolumn{5}{c}{\textbf{Base Pretrained Model}} \\
\toprule
\makecell[c]{MMLU\\[-5pt]\scriptsize{(5 shots)}} &
\makecell[c]{HellaSwag\\[-5pt]\scriptsize{(0 shots)}} &
\makecell[c]{ARC-C\\[-5pt]\scriptsize{(10 shots)}} &
\makecell[c]{TruthfulQA\\[-5pt]\scriptsize{(0 shots)}} &
Avg. Acc. $\uparrow$ \\
\midrule
52.38 & 60.18 & 47.53 & 26.81 & 46.73 \\
\bottomrule
\end{tabular}}%
\\[2.5em]
\end{tabular}

\vspace*{0.4em}
\resizebox{\textwidth}{!}{%
\begin{tabular}{l|cccc|c c l|cccc|c}
\multicolumn{6}{c}{\textbf{GPipe}} & \cellcolor{white}\quad\quad & \multicolumn{6}{c}{\textbf{1F1B}} \\
\cmidrule[\heavyrulewidth]{1-6}\cmidrule[\heavyrulewidth]{8-13}

Freeze Method & MMLU & \footnotesize HellaSwag & ARC-C & TruthQA & Avg. &
\cellcolor{white}\quad\quad &
Freeze Method & MMLU & \footnotesize HellaSwag & ARC-C & TruthQA & Avg. \\
\cmidrule{1-6}\cmidrule{8-13}

No Freezing
& 56.33 & 60.33 & 54.27 & 33.41 & 51.09
& \cellcolor{white}\quad\quad &
No Freezing
& 55.57 & 60.71 & 54.27 & 34.03 & 51.15 \\
\graycline{1-6}\graycline{8-13}

APF
& 55.03 & 60.35 & 53.24 & 31.95 & 50.14
& \cellcolor{white}\quad\quad &
APF
& 55.24 & 60.34 & 52.99 & 32.68 & 50.31 \\
AutoFreeze
& 54.89 & 60.41 & 53.41 & 32.19 & 50.23
& \cellcolor{white}\quad\quad &
AutoFreeze
& 53.07 & 59.29 & 53.58 & 34.88 & 50.21 \\

\rowcolor{gray!15}
\TimelyFreeze{}
& 55.21 & 60.76 & 54.69 & 32.56 & \underline{50.81}
& \cellcolor{white}\quad\quad &
\TimelyFreeze{}
& 55.71 & 60.22 & 55.12 & 33.17 & \textbf{51.06} \\

\rowcolor{lightgray!10}
\scriptsize\quad +APF
& 54.81 & 60.95 & 54.69 & 33.54 & \textbf{51.00}
& \cellcolor{white}\quad\quad &
\scriptsize\quad +APF
& 55.78 & 60.82 & 52.47 & 33.29 & \underline{50.59} \\

\rowcolor{lightgray!10}
\scriptsize\quad +AutoFreeze
& 55.76 & 61.33 & 53.58 & 32.56 & 50.81
& \cellcolor{white}\quad\quad &
\scriptsize\quad +AutoFreeze
& 55.68 & 60.72 & 52.73 & 32.93 & 50.52 \\
\cmidrule[\heavyrulewidth]{1-6}\cmidrule[\heavyrulewidth]{8-13}

\multicolumn{6}{c}{\textbf{Interleaved 1F1B}} & \cellcolor{white}\quad\quad & \multicolumn{6}{c}{\textbf{ZBV}} \\
\cmidrule[\heavyrulewidth]{1-6}\cmidrule[\heavyrulewidth]{8-13}

Freeze Method & MMLU & \footnotesize HellaSwag & ARC-C & TruthQA & Avg. &
\cellcolor{white}\quad\quad &
Freeze Method & MMLU & \footnotesize HellaSwag & ARC-C & TruthQA & Avg. \\
\cmidrule{1-6}\cmidrule{8-13}

No Freezing
& 55.68 & 59.98 & 53.67 & 34.27 & 50.90
& \cellcolor{white}\quad\quad &
No Freezing
& 55.59 & 60.17 & 52.56 & 32.56 & 50.22 \\
\graycline{1-6}\graycline{8-13}

APF
& 55.36 & 60.44 & 50.85 & 31.82 & 49.62
& \cellcolor{white}\quad\quad &
APF
& 55.22 & 59.98 & 52.05 & 31.82 & 49.77 \\
AutoFreeze
& 54.76 & 57.89 & 54.61 & 27.05 & 48.58
& \cellcolor{white}\quad\quad &
AutoFreeze
& 55.60 & 60.60 & 53.07 & 30.35 & \underline{49.91} \\

\rowcolor{gray!15}
\TimelyFreeze{}
& 54.91 & 60.43 & 51.79 & 31.70 & 49.71
& \cellcolor{white}\quad\quad &
\TimelyFreeze{}
& 55.43 & 60.53 & 52.90 & 32.56 & \textbf{50.36} \\

\rowcolor{lightgray!10}
\scriptsize\quad +APF
& 55.24 & 60.62 & 52.82 & 32.80 & \textbf{50.37}
& \cellcolor{white}\quad\quad &
\scriptsize\quad +APF
& 54.21 & 59.37 & 51.19 & 29.99 & 48.69 \\

\rowcolor{lightgray!10}
\scriptsize\quad +AutoFreeze
& 54.45 & 60.48 & 53.07 & 31.46 & \underline{49.87}
& \cellcolor{white}\quad\quad &
\scriptsize\quad +AutoFreeze
& 55.59 & 60.23 & 50.94 & 31.70 & 49.62 \\
\cmidrule[\heavyrulewidth]{1-6}\cmidrule[\heavyrulewidth]{8-13}

\end{tabular}%
}

\label{tab:llama13b_benchmark}
\end{table*}

\clearpage
\section{Pipeline Schedule Visualization}
\subsection{Four-GPU Pipeline Schedules}
\label{appendix:schedules_comparison_4gpus}
\autoref{fig:gpipe_schedules_comparison} presents example pipeline schedules for training the Llama-3-8B model using the GPipe schedule with eight microbatches on four H200 GPUs.
This figure visualizes the pipeline schedules corresponding to the main results reported in \autoref{tab:llama8b}.
We compare three parameter-freezing methods—AutoFreeze, APF, and TimelyFreeze—against the baseline training without freezing.
All freezing methods reduce the batch execution time to some extent; however, TimelyFreeze achieves the largest improvement, reducing the batch time by 31.66\% relative to the baseline of 698\,ms.
\\

\begin{figure}[H]
\centering
\begin{subfigure}[t]{\linewidth}
    \centering
    \includegraphics[height=3.5cm]{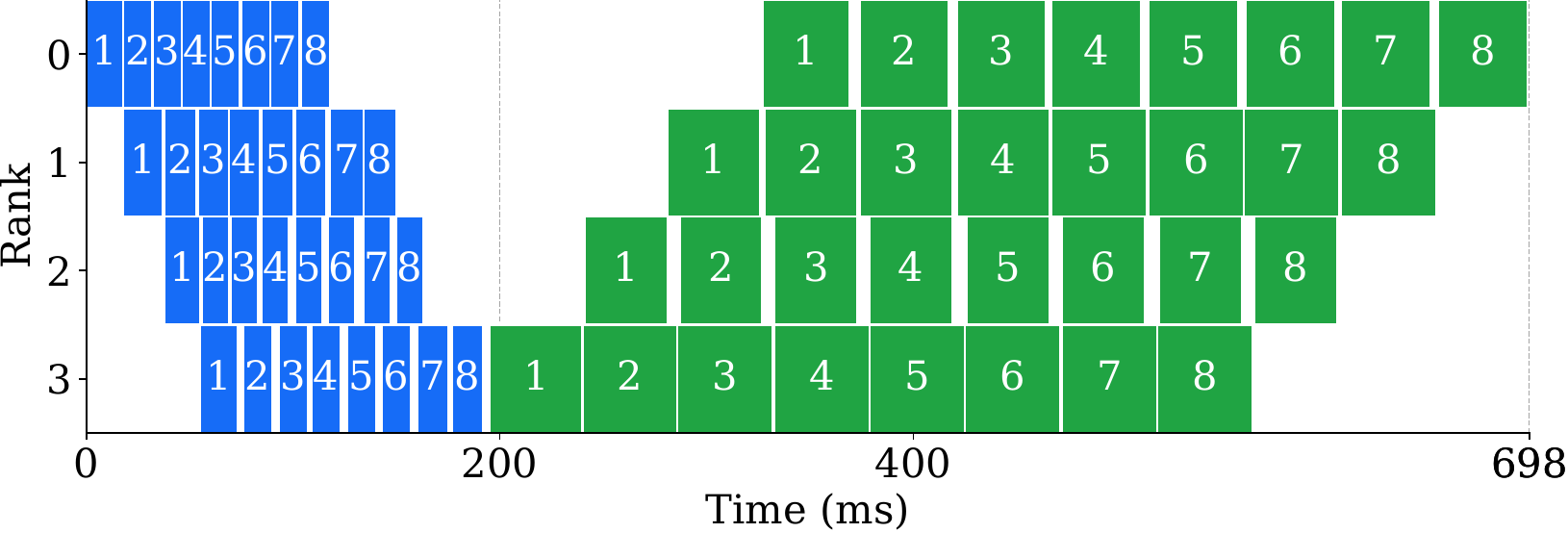}
    \caption{GPipe: Execution pipeline without freezing.}
    \label{fig:gpipe_nofreeze}
    \vspace{1em}
\end{subfigure}
\begin{subfigure}[t]{\linewidth}
    \centering
    \includegraphics[height=3.5cm]{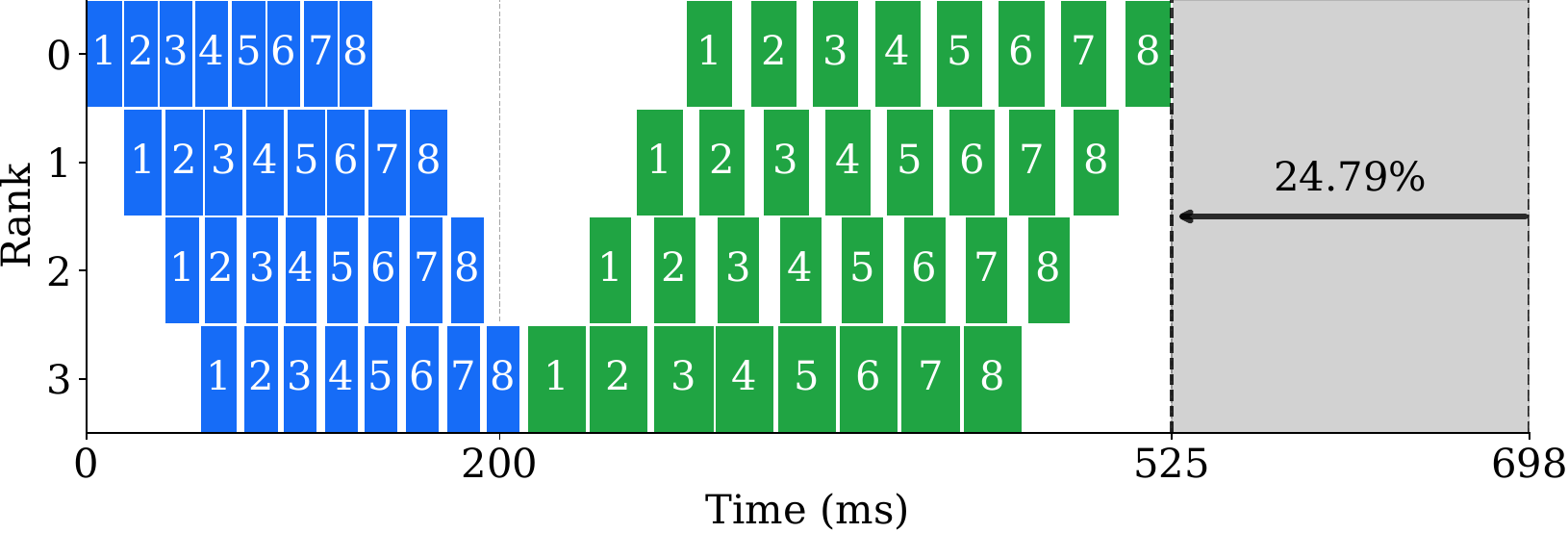}
    \caption{GPipe: Execution pipeline with AutoFreeze.}
    \label{fig:gpipe_auto}
    \vspace{1em}
\end{subfigure}
\begin{subfigure}[t]{\linewidth}
    \centering
    \includegraphics[height=3.5cm]{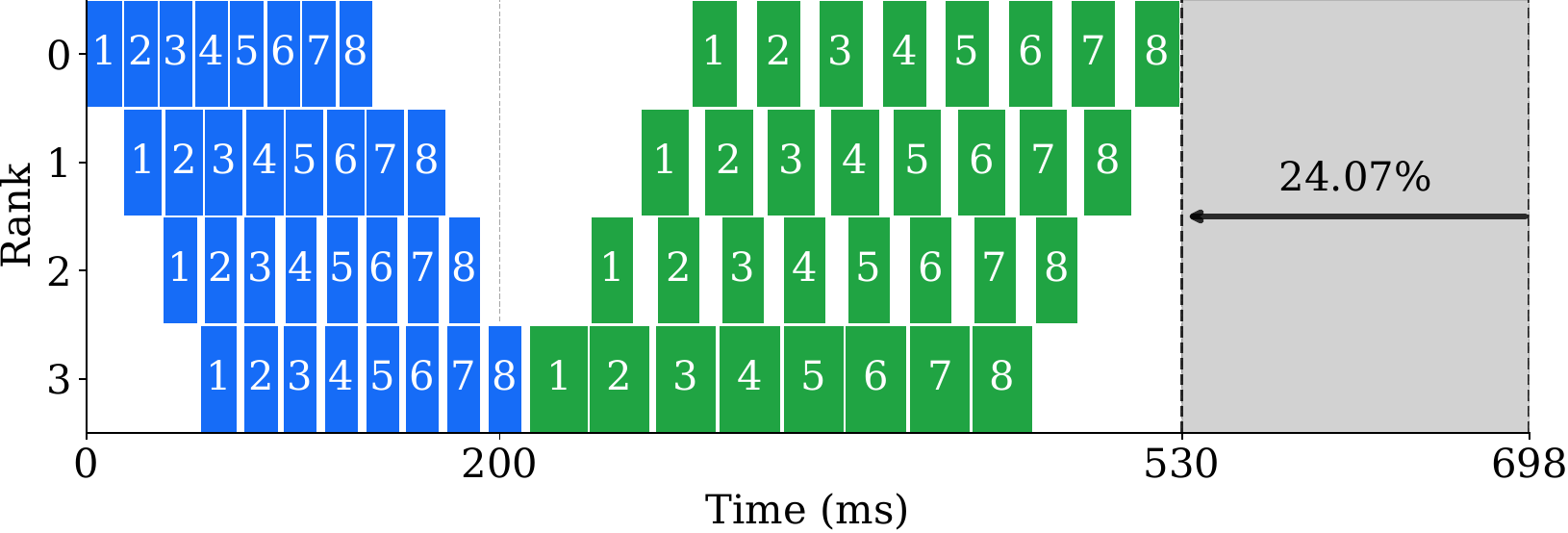}
    \caption{GPipe: Execution pipeline with APF.}
    \label{fig:gpipe_apf}
    \vspace{1em}
\end{subfigure}
\begin{subfigure}[t]{\linewidth}
    \centering
    \includegraphics[height=3.5cm]{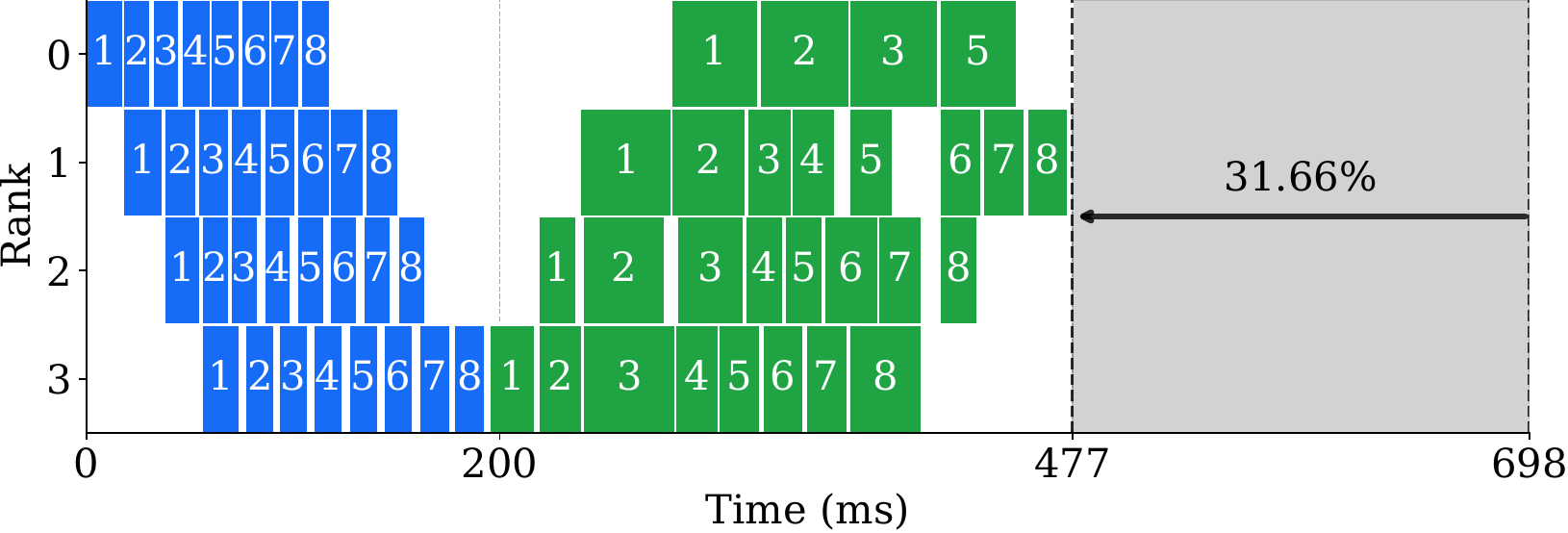}
    \caption{GPipe: Execution pipeline with \TimelyFreeze{}.}
    \label{fig:gpipe_ours}
    \vspace{1em}
\end{subfigure}
    \caption{GPipe schedules across four freezing methods.}
    \label{fig:gpipe_schedules_comparison}
\end{figure}

\clearpage
Figures~\ref{fig:1f1b_schedules_comparison} through~\ref{fig:zbv_schedules_comparison} demonstrate consistent results across other pipeline schedules, including 1F1B, Interleaved 1F1B, and Zero-Bubble (ZBV).
\\

\begin{figure}[H]
\centering
\begin{subfigure}[t]{\linewidth}
    \centering
    \includegraphics[height=4cm]{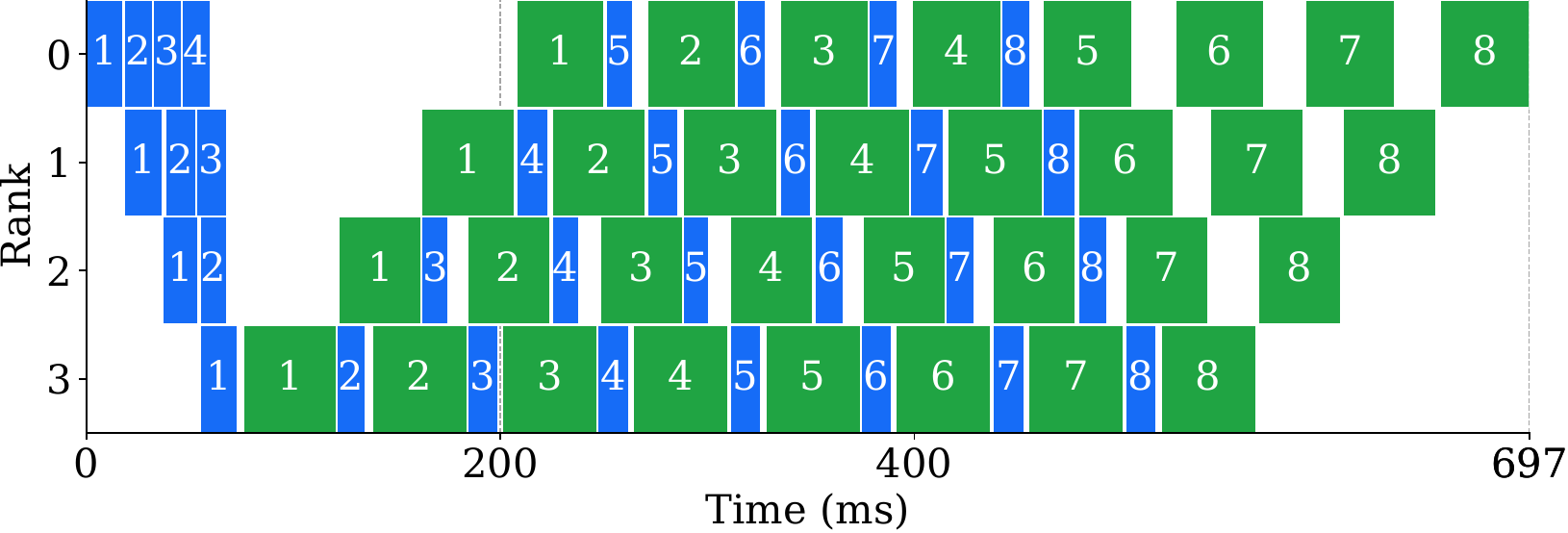}
    \caption{1F1B: Execution pipeline without freezing.}
    \label{fig:1f1b_nofreeze}
    \vspace{1.6em}
\end{subfigure}
\begin{subfigure}[t]{\linewidth}
    \centering
    \includegraphics[height=4cm]{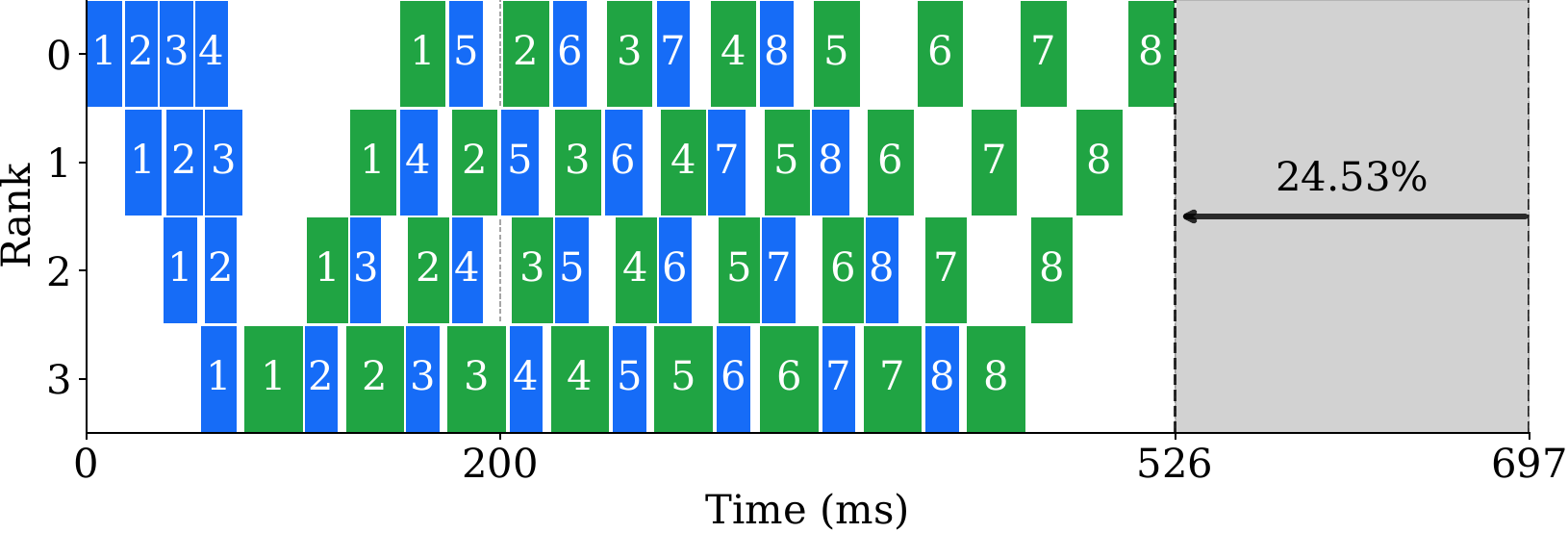}
    \caption{1F1B: Execution pipeline with AutoFreeze.}
    \label{fig:1f1b_auto}
    \vspace{1.6em}
\end{subfigure}
\begin{subfigure}[t]{\linewidth}
    \centering
    \includegraphics[height=4cm]{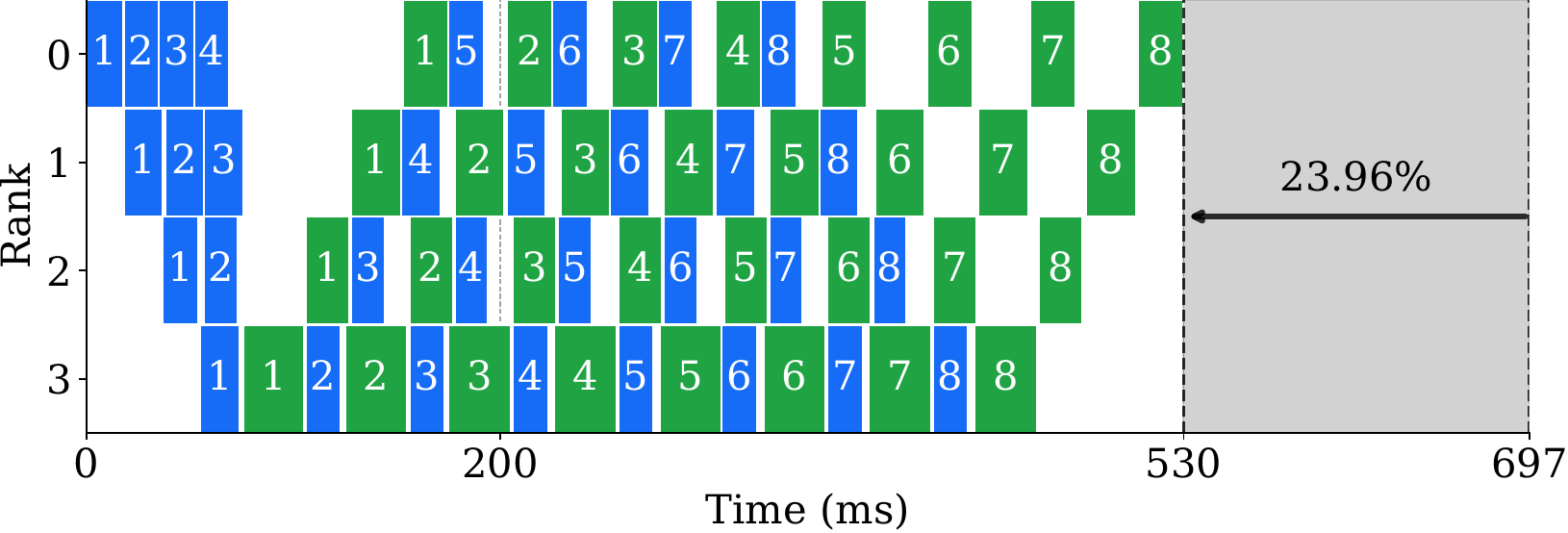}
    \caption{1F1B: Execution pipeline with APF.}
    \label{fig:1f1b_apf}
    \vspace{1.6em}
\end{subfigure}
\begin{subfigure}[t]{\linewidth}
    \centering
    \includegraphics[height=4cm]{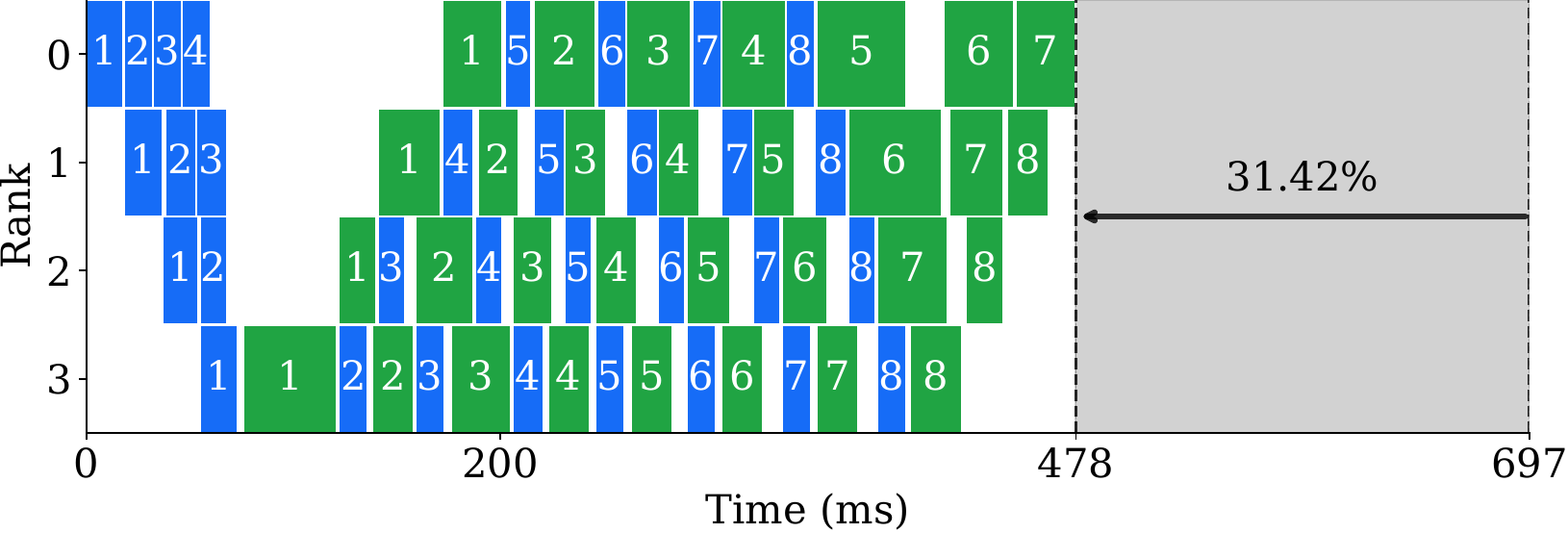}
    \caption{1F1B: Execution pipeline with \TimelyFreeze{}.}
    \label{fig:1f1b_ours}
    \vspace{1.6em}
\end{subfigure}
    \caption{1F1B schedules across four freezing methods.}
    \label{fig:1f1b_schedules_comparison}
\end{figure}

\begin{figure}[H]
\centering
\begin{subfigure}[t]{\linewidth}
    \centering
    \includegraphics[height=4cm]{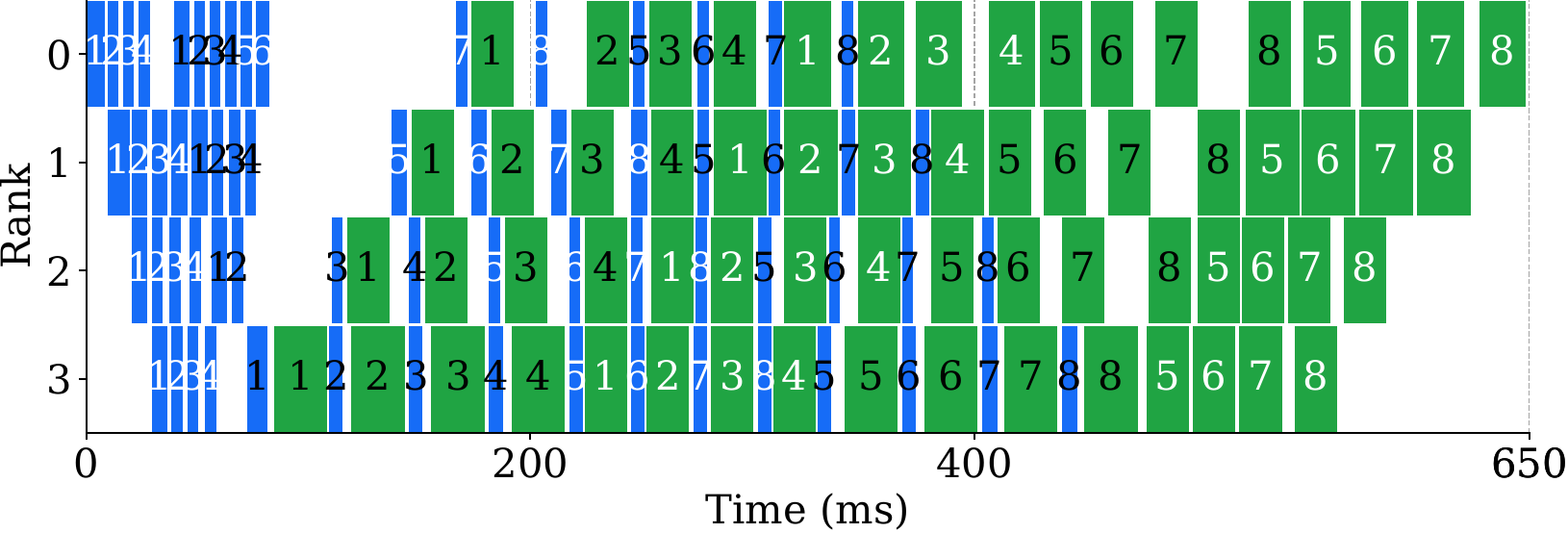}
    \caption{Interleaved 1F1B: Execution pipeline without freezing.}
    \label{fig:interleaved1f1b_nofreeze}
    \vspace{1.6em}
\end{subfigure}
\begin{subfigure}[t]{\linewidth}
    \centering
    \includegraphics[height=4cm]{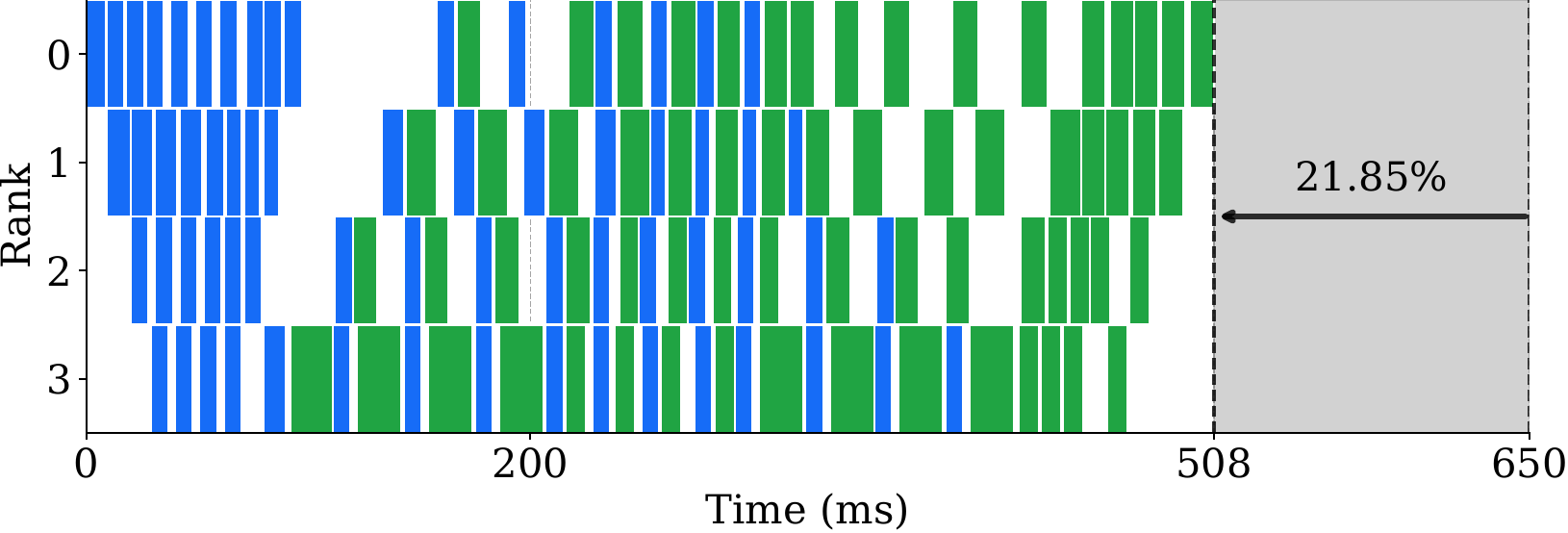}
    \caption{Interleaved 1F1B: Execution pipeline with AutoFreeze.}
    \label{fig:interleaved1f1b_auto}
    \vspace{1.6em}
\end{subfigure}
\begin{subfigure}[t]{\linewidth}
    \centering
    \includegraphics[height=4cm]{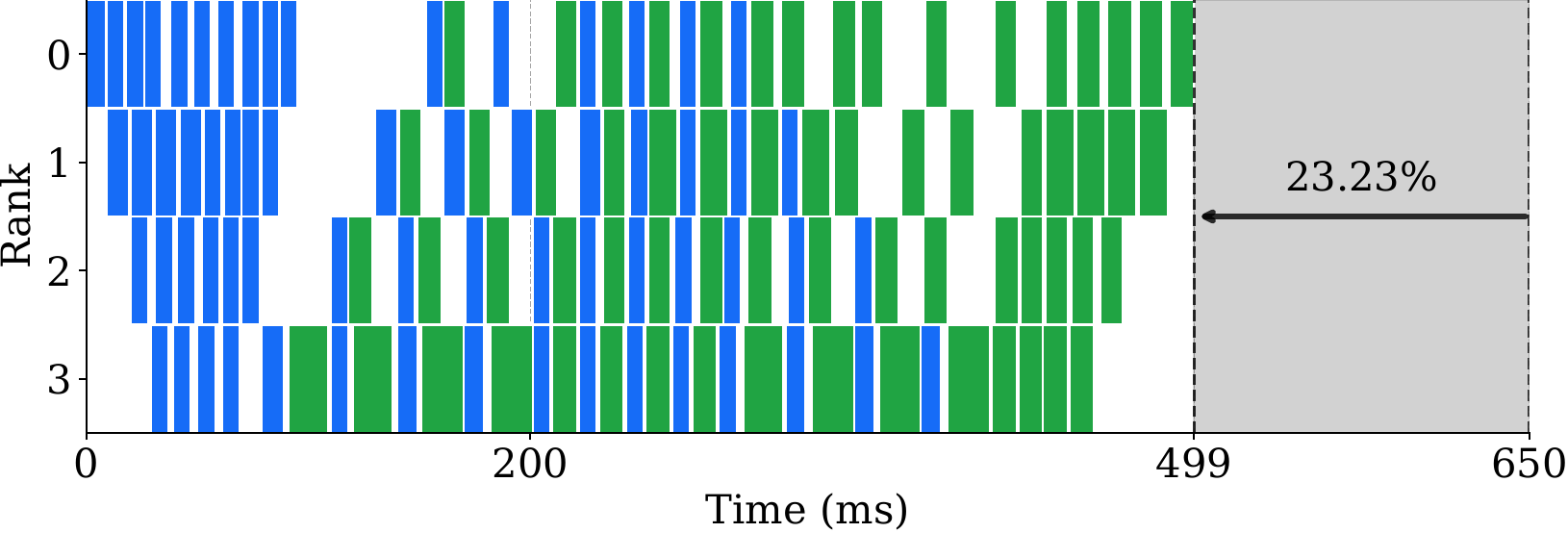}
    \caption{Interleaved 1F1B: Execution pipeline with APF.}
    \label{fig:interleaved1f1b_apf}
    \vspace{1.6em}
\end{subfigure}
\begin{subfigure}[t]{\linewidth}
    \centering
    \includegraphics[height=4cm]{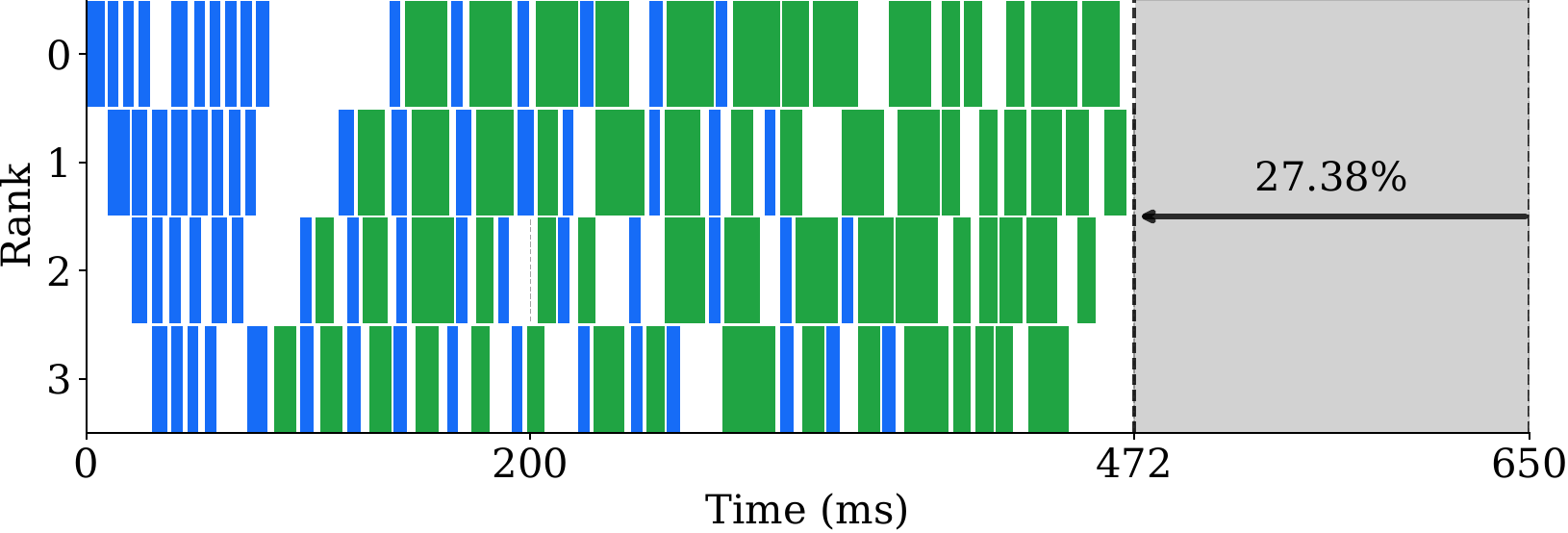}
    \caption{Interleaved 1F1B: Execution pipeline with \TimelyFreeze{}.}
    \label{fig:interleaved1f1b_ours}
    \vspace{1.6em}
\end{subfigure}
    \caption{Interleaved 1F1B schedules across four freezing methods.}
    \label{fig:interleaved1f1b_schedules_comparison}
\end{figure}

\begin{figure}[H]
\centering
\begin{subfigure}[t]{\linewidth}
    \centering
    \includegraphics[height=4cm]{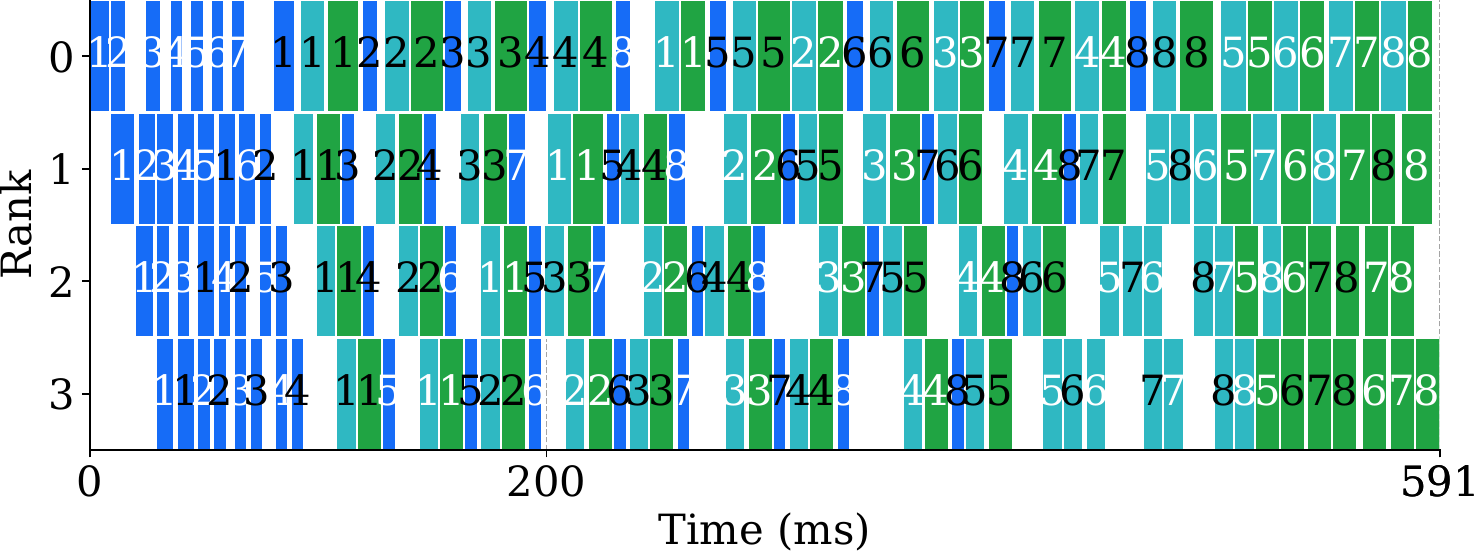}
    \caption{ZBV: Execution pipeline without freezing.}
    \label{fig:zbv_nofreeze}
    \vspace{1.6em}
\end{subfigure}
\begin{subfigure}[t]{\linewidth}
    \centering
    \includegraphics[height=4cm]{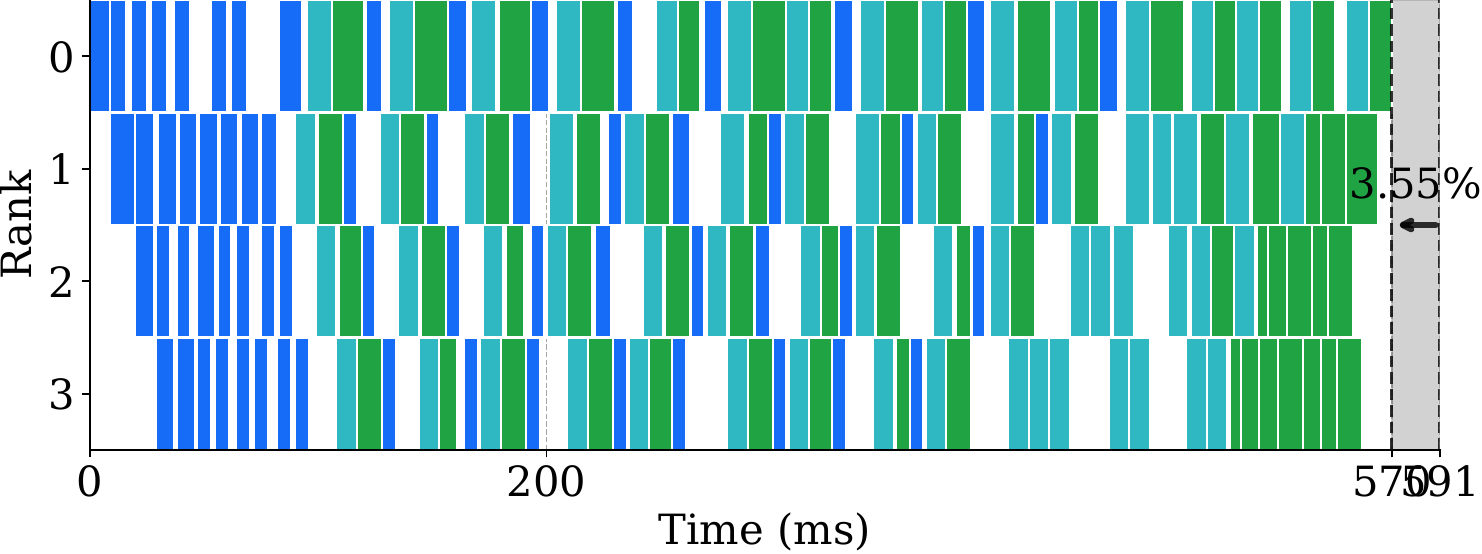}
    \caption{ZBV: Execution pipeline with AutoFreeze.}
    \label{fig:zbv_auto}
    \vspace{1.6em}
\end{subfigure}
\begin{subfigure}[t]{\linewidth}
    \centering
    \includegraphics[height=4cm]{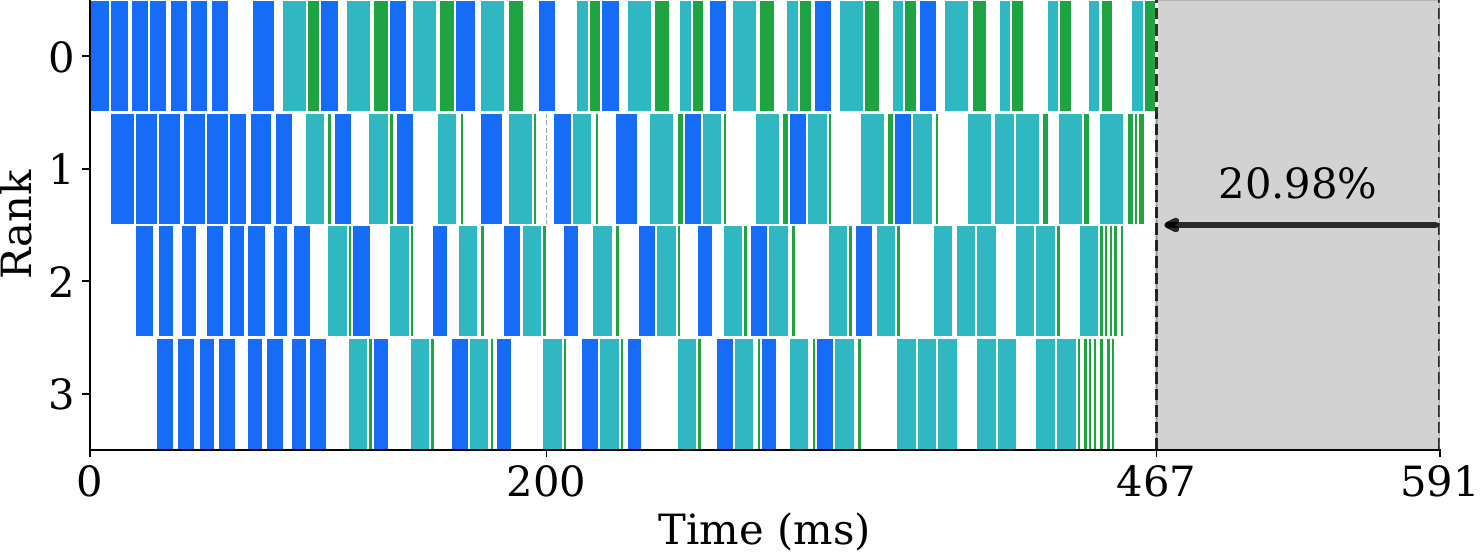}
    \caption{ZBV: Execution pipeline with APF.}
    \label{fig:zbv_apf}
    \vspace{1.6em}
\end{subfigure}
\begin{subfigure}[t]{\linewidth}
    \centering
    \includegraphics[height=4cm]{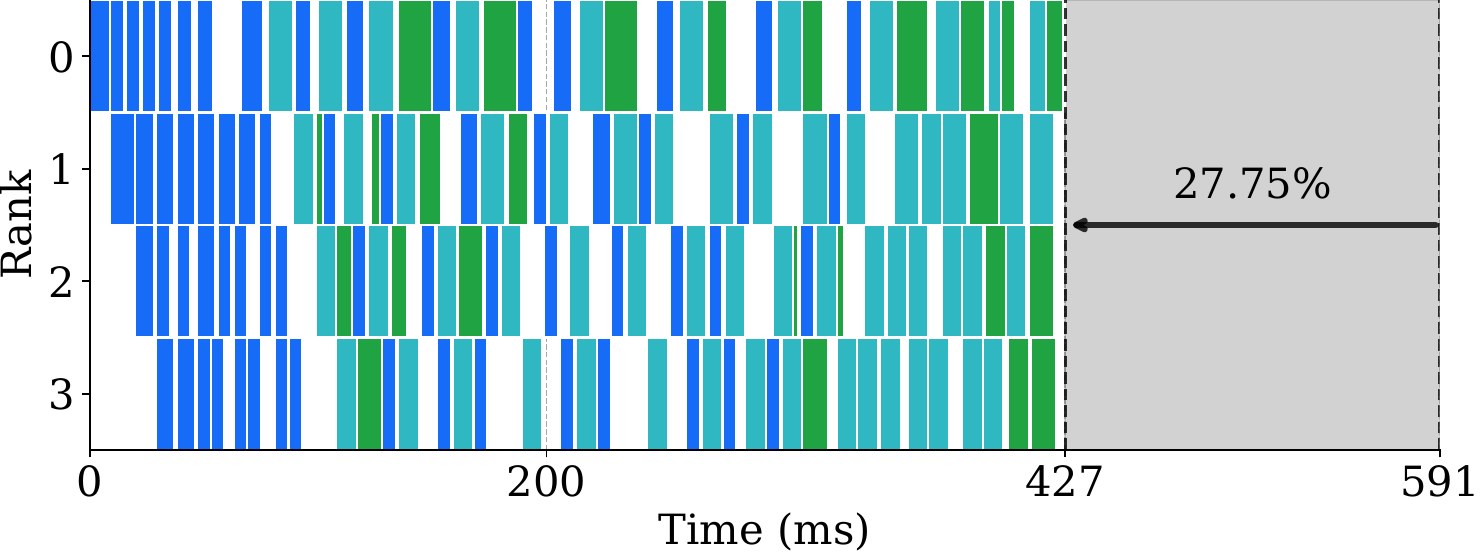}
    \caption{ZBV: Execution pipeline with \TimelyFreeze{}.}
    \label{fig:zbv_ours}
    \vspace{1.6em}
\end{subfigure}
    \caption{ZBV schedules across four freezing methods.}
    \label{fig:zbv_schedules_comparison}
\end{figure}

\clearpage
\subsection{6-GPU Pipeline Schedules}

Figures~\ref{fig:gpipe_schedules_comparison_6gpus} and \ref{fig:1f1b_schedules_comparison_6gpus} present example pipeline schedules for training the Llama-3.2-1B model using the GPipe and 1F1B schedules with six microbatches on six A6000 GPUs.
All other settings are identical to those used to obtain the results reported in \autoref{tab:llama1b}.

Compared to the baseline without freezing, both APF and TimelyFreeze reduce the batch execution time; however, TimelyFreeze consistently achieves larger gains.
In this setting, TimelyFreeze outperforms APF by up to 10 percentage points, which is notably higher than the 6--7 percentage point improvements observed in Figures~\ref{fig:gpipe_schedules_comparison} and \ref{fig:1f1b_schedules_comparison}.
This result indicates that TimelyFreeze becomes increasingly effective as the degree of pipeline parallelism increases, where pipeline-level optimization plays a more critical role.
\\

\begin{figure}[H]
\centering
\begin{subfigure}[t]{\linewidth}
    \centering
    \includegraphics[height=4cm]{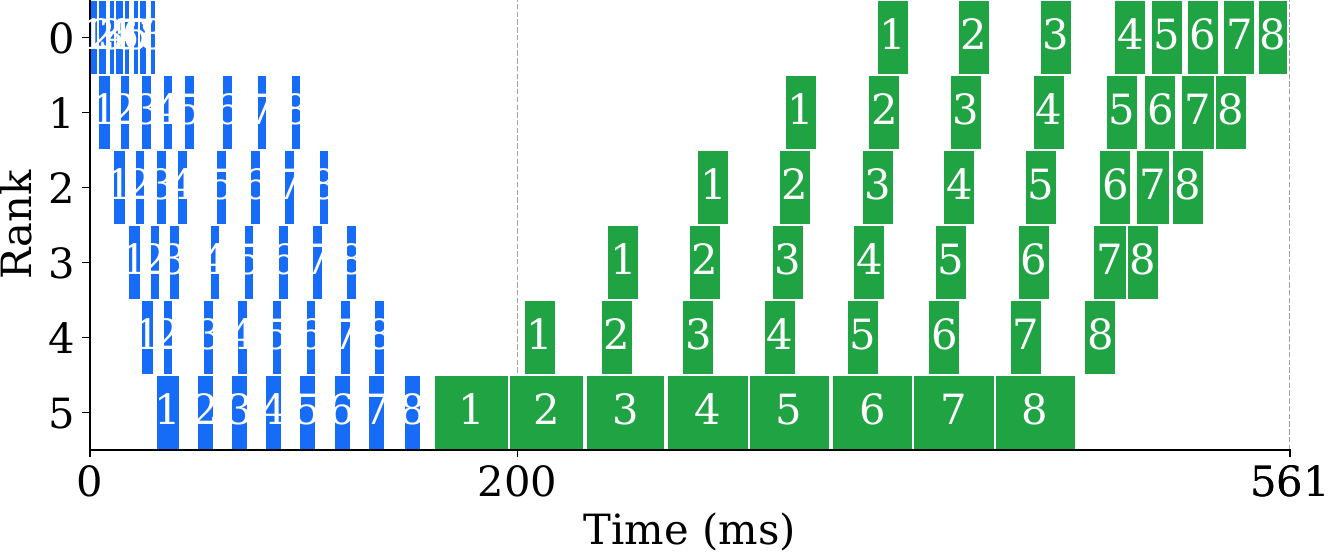}
    \caption{GPipe: Execution pipeline without freezing.}
    \label{fig:gpipe_nofreeze_6gpus}
    \vspace{1.6em}
\end{subfigure}
\begin{subfigure}[t]{\linewidth}
    \centering
    \includegraphics[height=4cm]{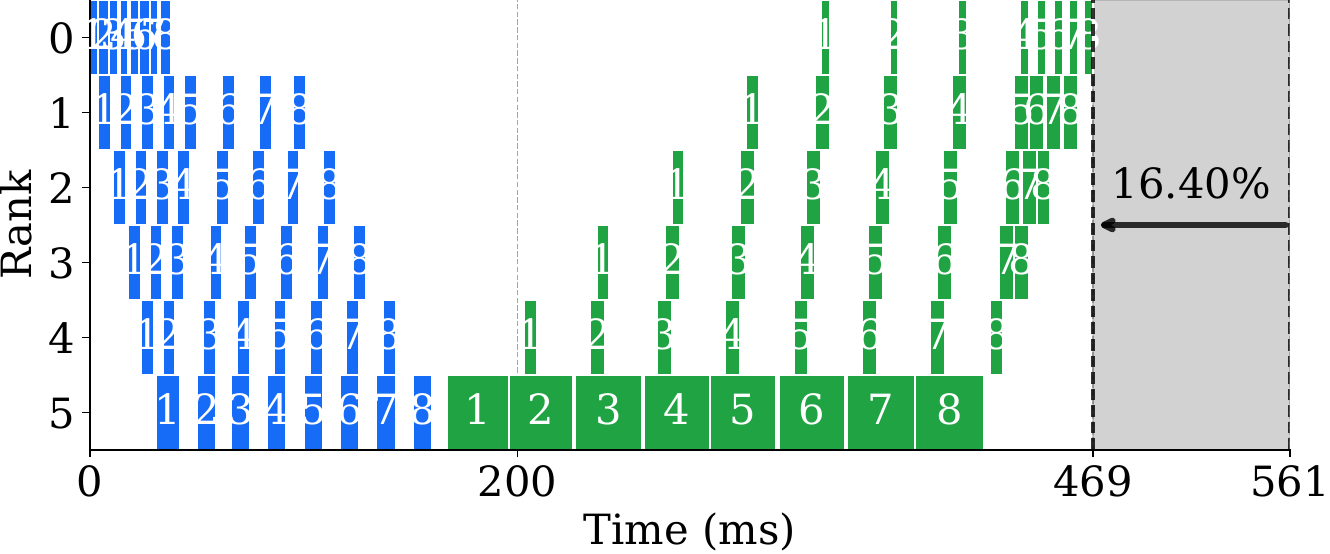}
    \caption{GPipe: Execution pipeline with APF.}
    \label{fig:gpipe_apf_6gpus}
    \vspace{1.6em}
\end{subfigure}
\begin{subfigure}[t]{\linewidth}
    \centering
    \includegraphics[height=4cm]{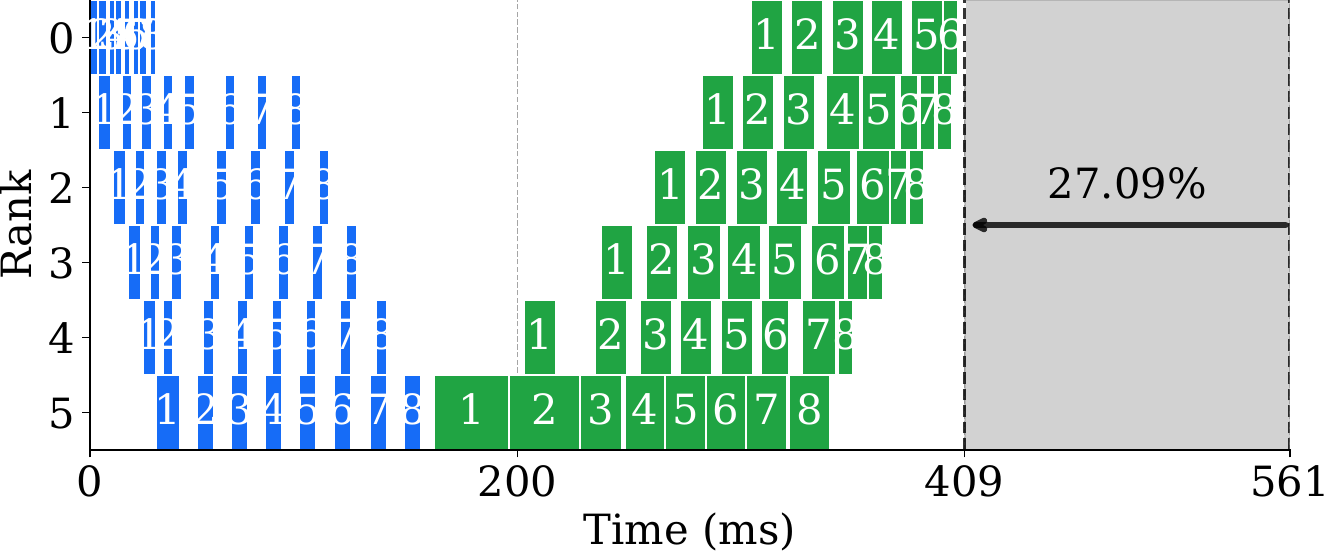}
    \caption{GPipe: Execution pipeline with \TimelyFreeze{}.}
    \label{fig:gpipe_ours_6gpus}
    \vspace{1.6em}
\end{subfigure}
    \caption{GPipe schedules across four freezing methods.}
    \label{fig:gpipe_schedules_comparison_6gpus}
\end{figure}

\begin{figure}[H]
\centering
\begin{subfigure}[t]{\linewidth}
    \centering
    \includegraphics[height=3.5cm]{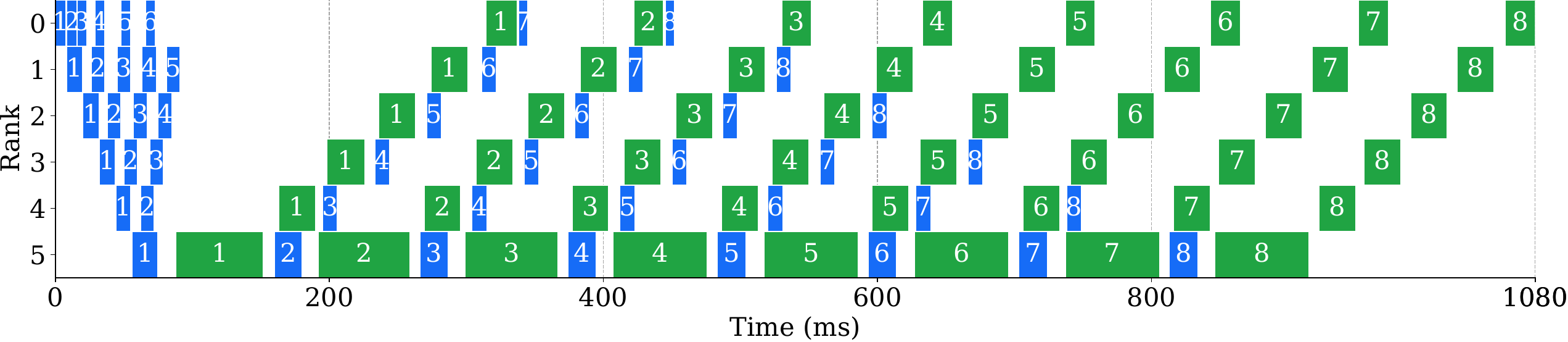}
    \caption{1F1B: Execution pipeline without freezing.}
    \label{fig:1f1b_nofreeze_6gpus}
    \vspace{1.6em}
\end{subfigure}
\begin{subfigure}[t]{\linewidth}
    \centering
    \includegraphics[height=3.5cm]{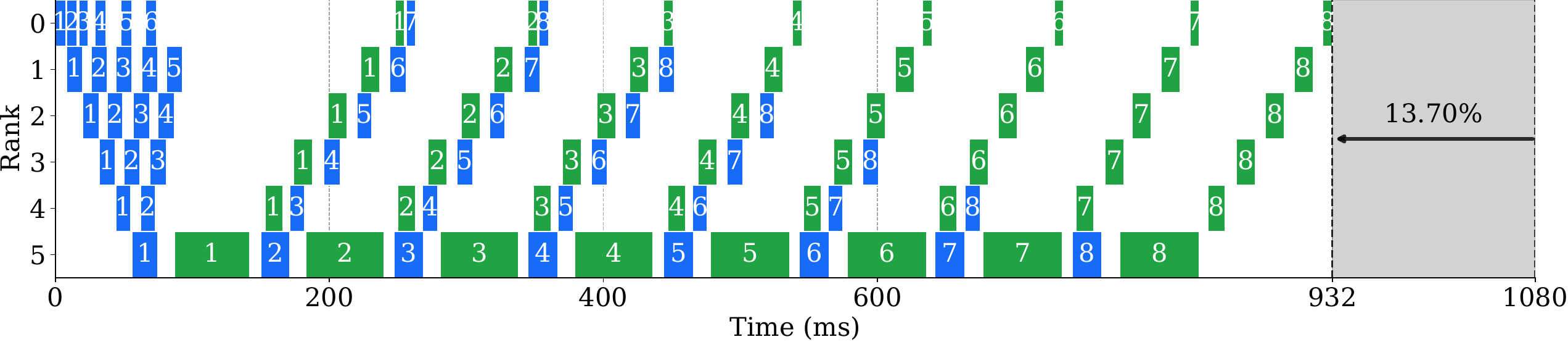}
    \caption{1F1B: Execution pipeline with APF.}
    \label{fig:1f1b_apf_6gpus}
    \vspace{1.6em}
\end{subfigure}
\begin{subfigure}[t]{\linewidth}
    \centering
    \includegraphics[height=3.5cm]{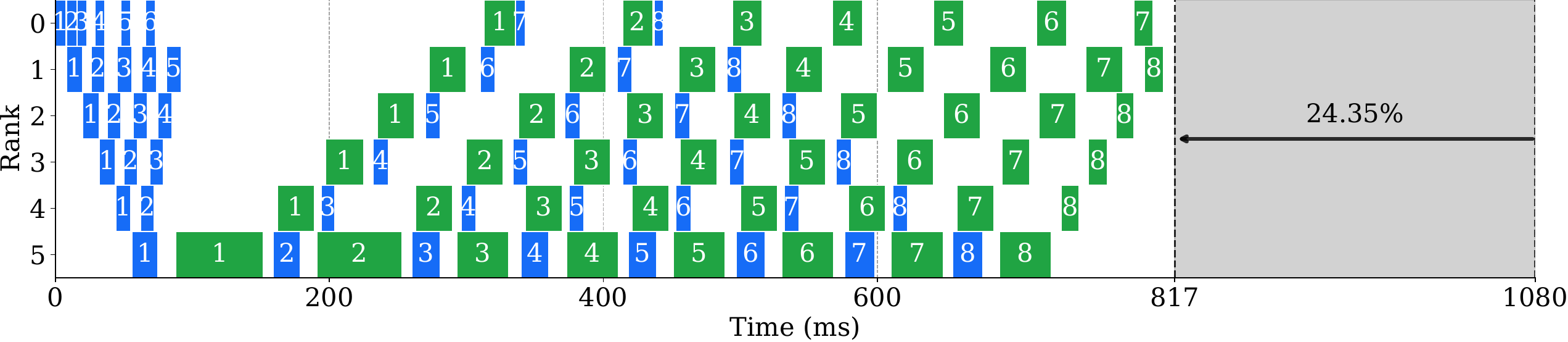}
    \caption{1F1B: Execution pipeline with \TimelyFreeze{}.}
    \label{fig:1f1b_ours_6gpus}
    \vspace{1.6em}
\end{subfigure}
    \caption{1F1B schedules across freezing methods.}
    \label{fig:1f1b_schedules_comparison_6gpus}
\end{figure}

\clearpage
\subsection{8-GPU Pipeline Schedules}
\begin{figure}[H]
\centering
\begin{subfigure}[t]{\linewidth}
    \centering
    \includegraphics[height=4cm]{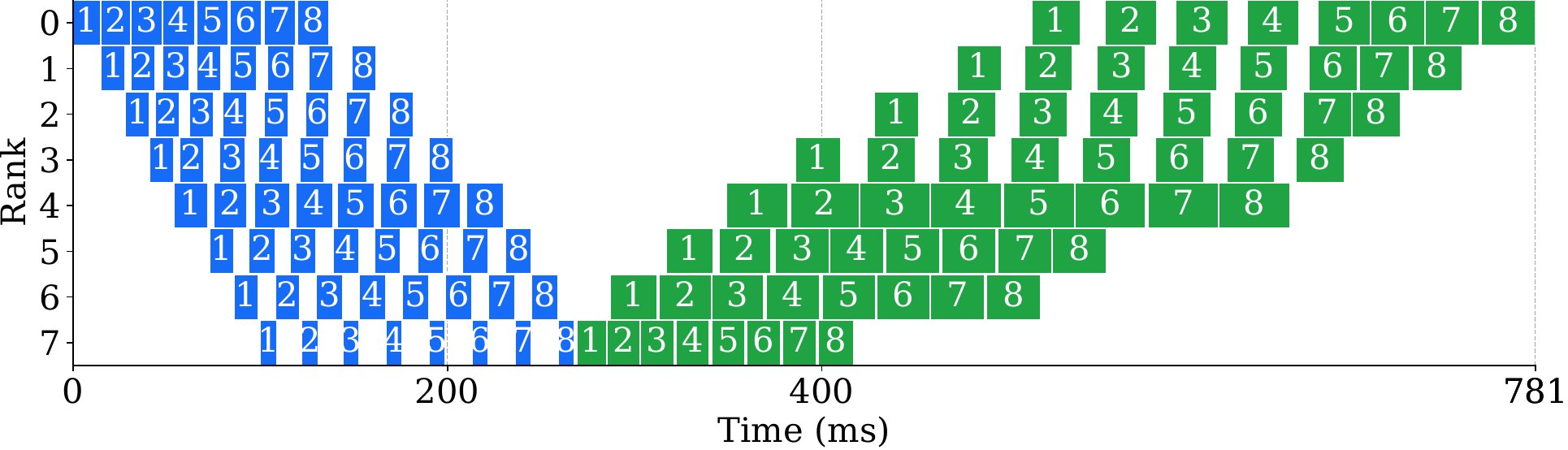}
    \caption{GPipe: Execution pipeline without freezing.}
    \label{fig:gpipe_nofreeze_8gpus}
    \vspace{1.6em}
\end{subfigure}
\begin{subfigure}[t]{\linewidth}
    \centering
    \includegraphics[height=4cm]{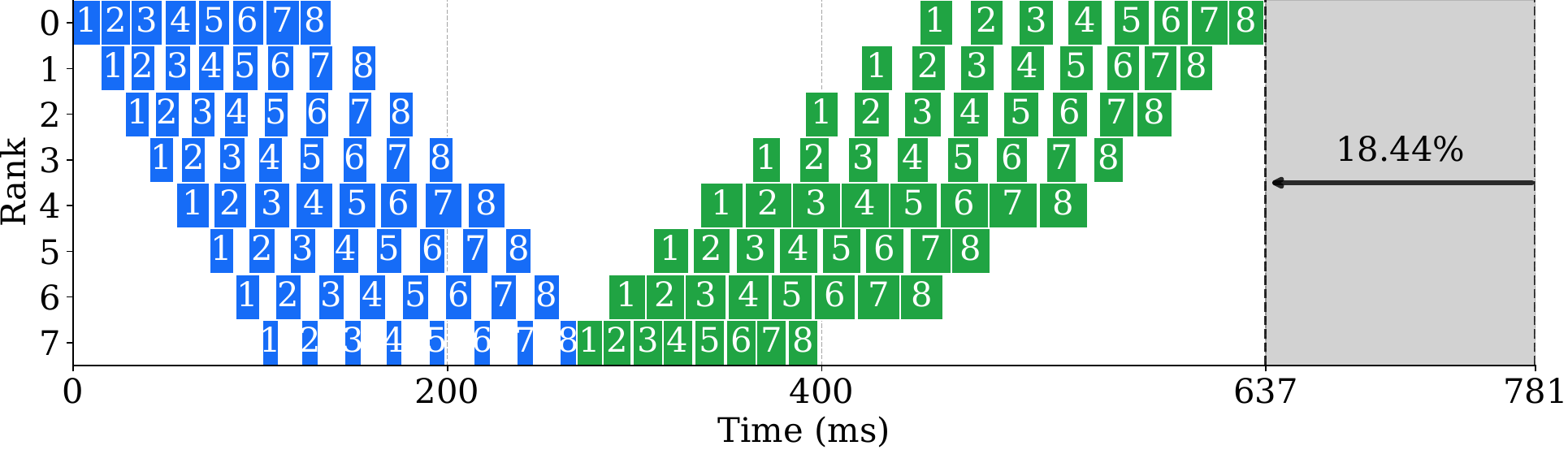}
    \caption{GPipe: Execution pipeline with APF.}
    \label{fig:gpipe_apf_8gpus}
    \vspace{1.6em}
\end{subfigure}
\begin{subfigure}[t]{\linewidth}
    \centering
    \includegraphics[height=4cm]{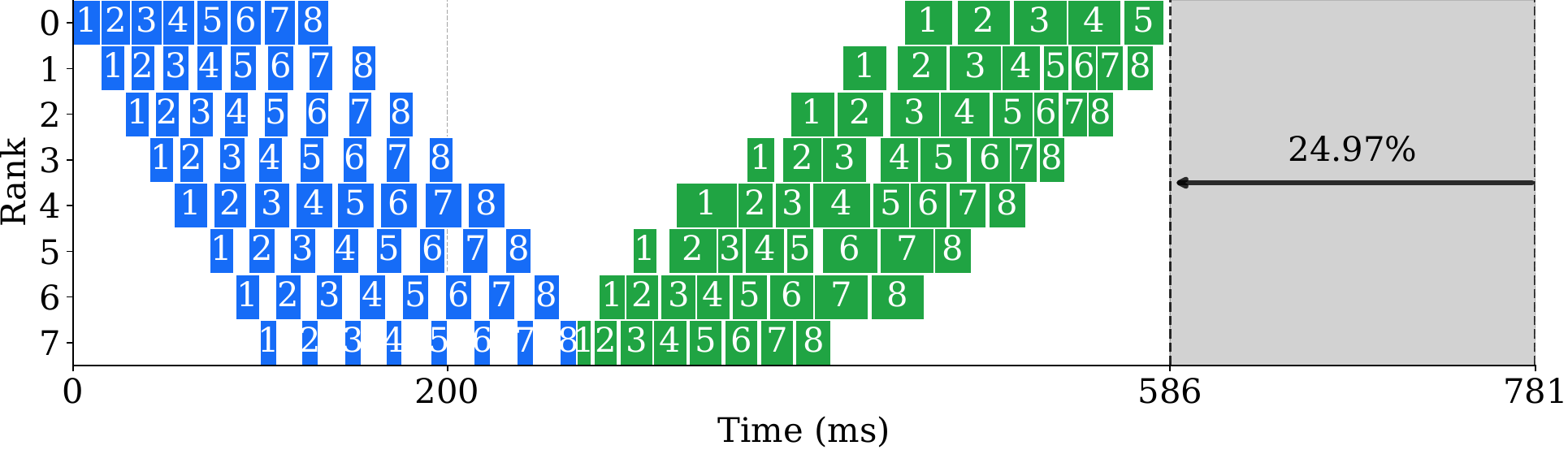}
    \caption{GPipe: Execution pipeline with \TimelyFreeze{}.}
    \label{fig:gpipe_ours_8gpus}
    \vspace{1.6em}
\end{subfigure}
    \caption{GPipe schedules across freezing methods.}
    \label{fig:gpipe_schedules_comparison_8gpus}
\end{figure}


\clearpage
\section{Vision Model Results}
TimelyFreeze also demonstrates strong generalization to image classification models,
including ViT-L/32 and ConvNeXt-V2-L, which differ substantially from transformer LLMs in
layer structure, operator heterogeneity, and training dynamics.
Tables \ref{tab:convnext_partitioning} and \ref{tab:vit_columnwise} summarize the results.

\label{appendix:vision_models}
\subsection{ConvNeXt-V2-L Result Table}
\label{appendix:convnext}

\vspace{3mm}

\begin{table*}[th]
\renewcommand{\arraystretch}{1.15}
\centering
\caption{
Effect of freezing methods across partitioning heuristics and pipeline schedules.
The best and second-best values among the three freezing methods (APF, AutoFreeze, and TimelyFreeze) are highlighted in bold and underline, respectively.
}
\label{tab:convnext_partitioning}

\resizebox{0.95\textwidth}{!}{
\begin{tabular}{cc|ccc|ccc}
\toprule
\multirow{2}{*}{\makecell{Partitioning \\ Method}} &
\multirow{2}{*}{\makecell{Freeze \\ Method}} &
\multicolumn{3}{c|}{\textbf{GPipe}} &
\multicolumn{3}{c}{\textbf{1F1B}} \\
\cline{3-8}
& &
\makecell{Top1 Acc. ($\Delta$)$\uparrow$} &
\makecell{Train Time ($\Delta$)$\downarrow$} &
Freeze Ratio &
\makecell{Top1 Acc. ($\Delta$)$\uparrow$} &
\makecell{Train Time ($\Delta$)$\downarrow$} &
Freeze Ratio \\
\midrule

\multirow{4}{*}{Memory}
& No Freezing
& \aaccsign{87.00}{+0.00} & \ttimesign{10507}{0.00} & 0.00
& \aaccsign{91.57}{+0.00} & \ttimesign{25212}{0.00} & 0.00 \\
\graycline{2-8}
& APF
& \aaccsign{84.50}{-2.50} & \ttimesign{\textbf{8812}}{16.13} & 36.31
& \aaccsign{90.38}{-1.19} & \ttimesign{\textbf{20784}}{17.56} & 36.92 \\

& AutoFreeze
& \aaccsign{\textbf{86.17}}{-0.83} & \ttimesign{10256}{2.39} & 21.57
& \aaccsign{\textbf{91.59}}{+0.02} & \ttimesign{24383}{3.29} & 23.03 \\

& \cellcolor{gray!15}\TimelyFreeze{}
& \cellcolor{gray!15}\aaccsign{\underline{85.53}}{-1.47}
& \cellcolor{gray!15}\ttimesign{\underline{8985}}{14.49}
& \cellcolor{gray!15}44.27
& \cellcolor{gray!15}\aaccsign{\underline{90.76}}{-0.81}
& \cellcolor{gray!15}\ttimesign{\underline{21687}}{13.98}
& \cellcolor{gray!15}29.12 \\
\midrule

\multirow{4}{*}{Parameter}
& No Freezing
& \aaccsign{86.28}{+0.00} & \ttimesign{15144}{0.00} & 0.00
& \aaccsign{91.59}{+0.00} & \ttimesign{31935}{0.00} & 0.00 \\
\graycline{2-8}
& APF
& \aaccsign{85.13}{-1.16} & \ttimesign{\textbf{12224}}{19.28} & 38.07
& \aaccsign{90.60}{-0.99} & \ttimesign{\underline{25397}}{20.47} & 36.66 \\

& AutoFreeze
& \aaccsign{\textbf{87.03}}{+0.75} & \ttimesign{13202}{12.82} & 10.38
& \aaccsign{\textbf{91.78}}{+0.19} & \ttimesign{27280}{14.58} & 10.40 \\

& \cellcolor{gray!15}\TimelyFreeze{}
& \cellcolor{gray!15}\aaccsign{\underline{86.47}}{+0.19} 
& \cellcolor{gray!15}\ttimesign{\underline{11378}}{24.87}
& \cellcolor{gray!15}21.17
& \cellcolor{gray!15}\aaccsign{\underline{91.04}}{-0.55} 
& \cellcolor{gray!15}\ttimesign{\textbf{24741}}{22.53}
& \cellcolor{gray!15}31.26 \\
\midrule

\multirow{4}{*}{Time}
& No Freezing
& \aaccsign{86.92}{+0.00} & \ttimesign{10476}{0.00} & 0.00
& \aaccsign{91.67}{+0.00} & \ttimesign{22820}{0.00} & 0.00 \\
\graycline{2-8}
& APF
& \aaccsign{84.80}{-2.13} & \ttimesign{\underline{8624}}{17.68} & 38.34
& \aaccsign{90.38}{-1.30} & \ttimesign{\underline{19220}}{15.78} & 37.10 \\

& AutoFreeze
& \aaccsign{\textbf{86.61}}{-0.31} & \ttimesign{9751}{6.92} & 23.11
& \aaccsign{\textbf{91.58}}{-0.09} & \ttimesign{21558}{5.53} & 24.10 \\

& \cellcolor{gray!15}\TimelyFreeze{}
& \cellcolor{gray!15}\aaccsign{\underline{85.80}}{-1.13} 
& \cellcolor{gray!15}\ttimesign{\textbf{8430}}{19.53}
& \cellcolor{gray!15}44.00
& \cellcolor{gray!15}\aaccsign{\underline{91.28}}{-0.40}
& \cellcolor{gray!15}\ttimesign{\textbf{18671}}{18.18}
& \cellcolor{gray!15}46.01 \\
\bottomrule
\end{tabular}
}
\end{table*}

We evaluate TimelyFreeze under three widely used partitioning heuristics:
\emph{memory-based}, \emph{parameter-based}, and \emph{time-based}.  
Memory-based partitioning balances peak activation and parameter memory across devices and is commonly used when training near memory limits, where feasibility and OOM avoidance dominate. Parameter-based partitioning is the simplest and most widely adopted baseline: it requires no profiling and splits layers by parameter counts as a coarse proxy for compute and memory load. Time-based partitioning equalizes per-stage latency using backward (or forward–backward) timing profiles, and is typically preferred when wall-clock throughput is the primary concern.

\subsection{ViT-L/32 Result Table}
\label{appendix:vit}
\vspace{3mm}

\begin{table*}[th]
\renewcommand{\arraystretch}{1.15}
\centering
\caption{
Finetuning results on ViT-L/32 with ImageNet-1K under different pipeline schedules.
}
\label{tab:vit_columnwise}

\resizebox{0.9\textwidth}{!}{
\begin{tabular}{c|ccc|ccc}
\toprule
\multirow{2}{*}{\makecell{Freeze \\ Method}} &
\multicolumn{3}{c|}{\textbf{GPipe}} &
\multicolumn{3}{c}{\textbf{1F1B}} \\
\cline{2-7}
&
\makecell{Top1 Acc. ($\Delta$)$\uparrow$} &
\makecell{Train Time ($\Delta$)$\downarrow$} &
Freeze Ratio &
\makecell{Top1 Acc. ($\Delta$)$\uparrow$} &
\makecell{Train Time ($\Delta$)$\downarrow$} &
Freeze Ratio \\
\midrule

No Freezing
& \aaccsign{75.46}{+0.00} 
& \ttimesign{14339}{0.00}
& 0.00
& \aaccsign{75.46}{+0.00} 
& \ttimesign{13374}{0.00}
& 0.00 \\
\graymidrule
APF
& \aaccsign{42.07}{-33.40}
& \ttimesign{\underline{12082}}{15.74}
& 39.15
& \aaccsign{42.07}{-33.40}
& \ttimesign{\underline{11305}}{15.47}
& 39.71 \\

AutoFreeze
& \aaccsign{\underline{74.98}}{-0.48}
& \ttimesign{13050}{8.99}
& 44.04
& \aaccsign{\underline{74.98}}{-0.48}
& \ttimesign{11991}{10.34}
& 44.04 \\

\cellcolor{gray!15}\TimelyFreeze{}
& \cellcolor{gray!15}\aaccsign{\textbf{75.03}}{-0.44}
& \cellcolor{gray!15}\ttimesign{\textbf{11195}}{21.93}
& \cellcolor{gray!15}50.00
& \cellcolor{gray!15}\aaccsign{\textbf{75.04}}{-0.42}
& \cellcolor{gray!15}\ttimesign{\textbf{10266}}{23.24}
& \cellcolor{gray!15}43.39 \\

\bottomrule
\end{tabular}
}
\end{table*}

\clearpage
\section{Variation of Per-Parameter Freeze Ratios across Methods}
\label{appendix:frozen_params_histogram}
Figure~\ref{fig:frozen_params_histogram_timelyfreeze_all} compares the
distributions of per-parameter freeze ratios across methods.
\TimelyFreeze{} exhibits a nearly uniform distribution, reflecting its
schedule-driven and parameter-agnostic design.
In contrast, APF produces highly skewed freeze ratios with large
parameter-wise variance, while AutoFreeze shows pronounced layer-wise
imbalances between early and late layers.
The hybrid variants (\TimelyAPF{} and \TimelyAuto{}) follow the trends of
their metric-based counterparts but apply freezing more moderately,
regularized by the stage-wise budgets computed by \TimelyFreeze{}.

\begin{figure}[!th]
    \centering
    \begin{subfigure}[t]{\linewidth}
        \centering
        \includegraphics[width=\linewidth, trim={2.9cm 0.3cm 2.9cm 3.9cm}, clip]{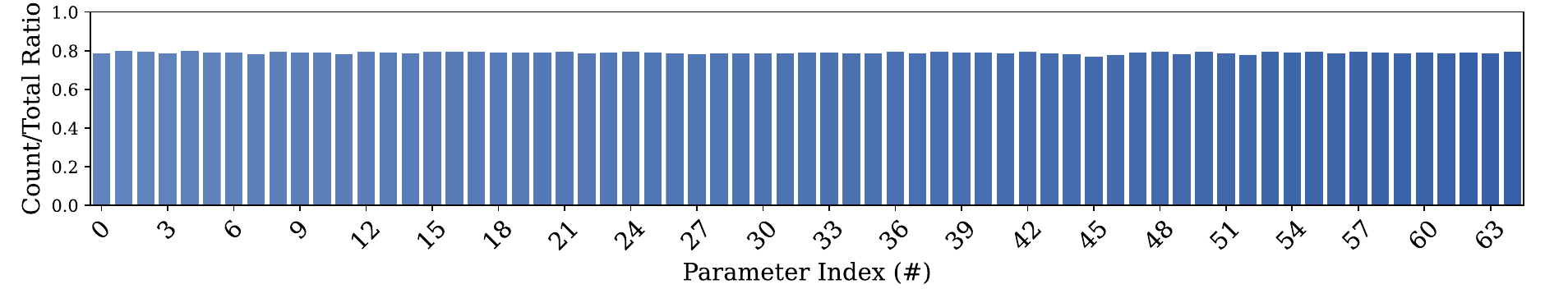}
        \caption{\TimelyFreeze{}.}
    \label{fig:frozen_params_histogram_timelyfreeze}
    \end{subfigure}
    \vspace{0.em}
    \begin{subfigure}[t]{\linewidth}
        \centering
        \includegraphics[width=\linewidth, trim={2.9cm 0.3cm 2.9cm 3.9cm}, clip]{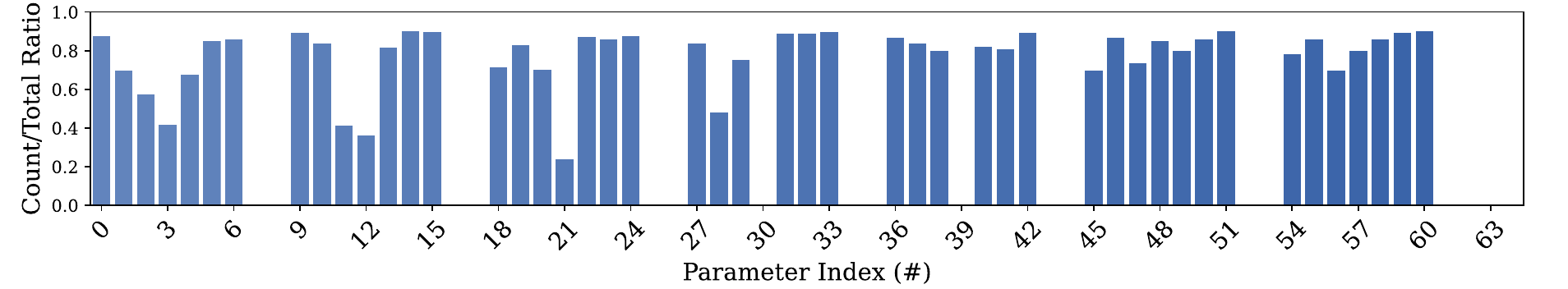}
        \caption{APF.}
        \label{fig:frozen_params_histogram_apf}
    \end{subfigure}
    \vspace{0.em}
    \begin{subfigure}[t]{\linewidth}
        \centering
        \includegraphics[width=\linewidth, trim={2.9cm 0.3cm 2.9cm 3.9cm}, clip]{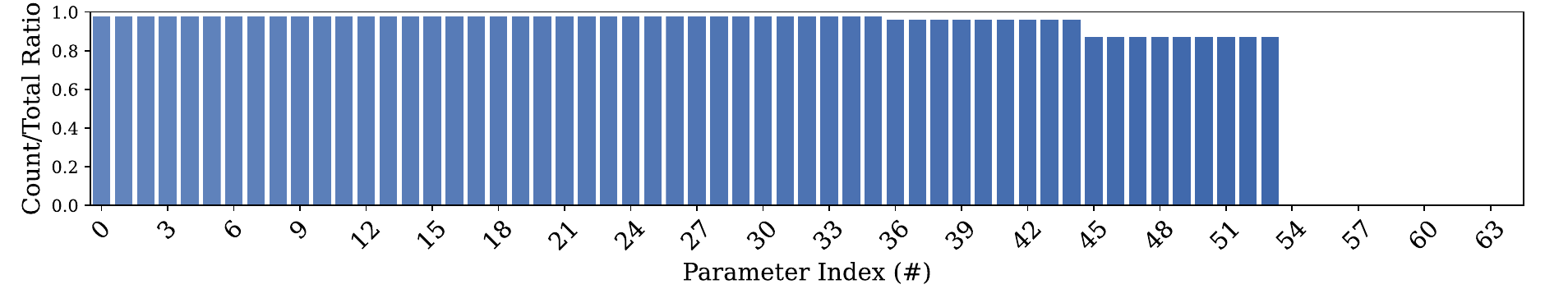}
        \caption{AutoFreeze.}
        \label{fig:frozen_params_histogram_autofreeze}
    \end{subfigure}
    \vspace{0.em}
    \begin{subfigure}[t]{\linewidth}
        \centering
        \includegraphics[width=\linewidth, trim={2.9cm 0.3cm 2.9cm 3.9cm}, clip]{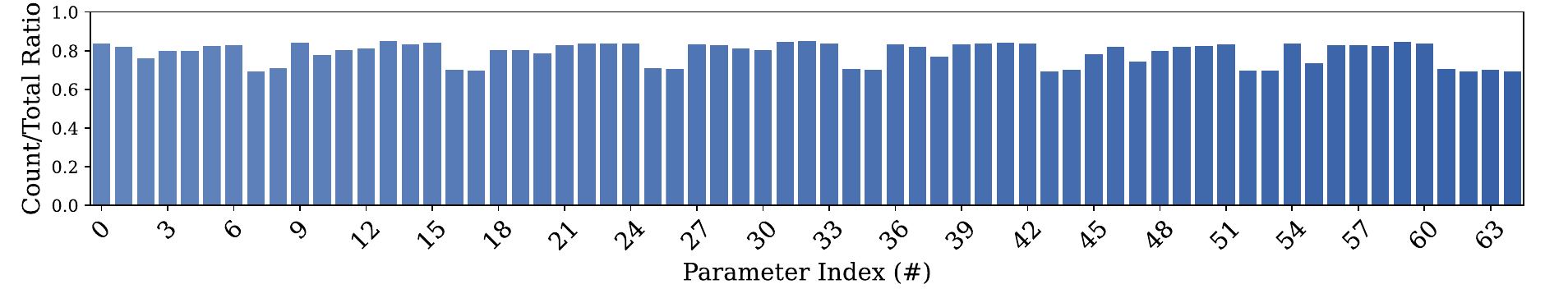}
        \caption{\TimelyAPF{}.}
        \label{fig:frozen_params_histogram_timelyapf}
    \end{subfigure}
    \vspace{0.em}
    \begin{subfigure}[t]{\linewidth}
        \centering
        \includegraphics[width=\linewidth, trim={2.9cm 0.3cm 2.9cm 3.9cm}, clip]{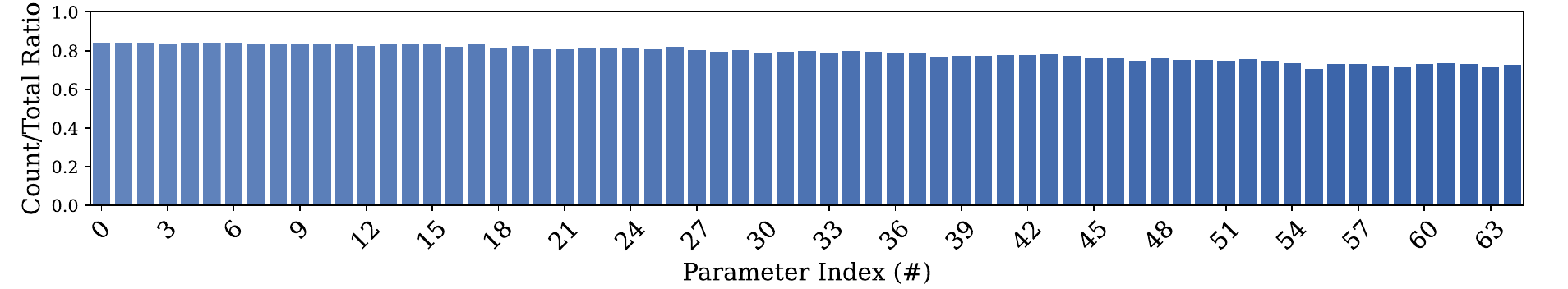}
        \caption{\TimelyAuto{}.}
        \label{fig:frozen_params_histogram_timelyauto}
    \end{subfigure}
    \caption{Frozen ratio histogram per parameter in Rank 3 with different methods.}
    \label{fig:frozen_params_histogram_timelyfreeze_all}
\end{figure}

\clearpage
\section{Effect of Freeze Ratio on Backward Time}
\label{appendix:afr_plots}

\begin{figure*}[th]
\centering

\begin{tabular}{@{}cc@{}}
\begin{subfigure}[t]{0.45\textwidth}
    \centering
    \includegraphics[
        width=\linewidth,
        trim=0 50pt 0 1pt,
        clip
    ]{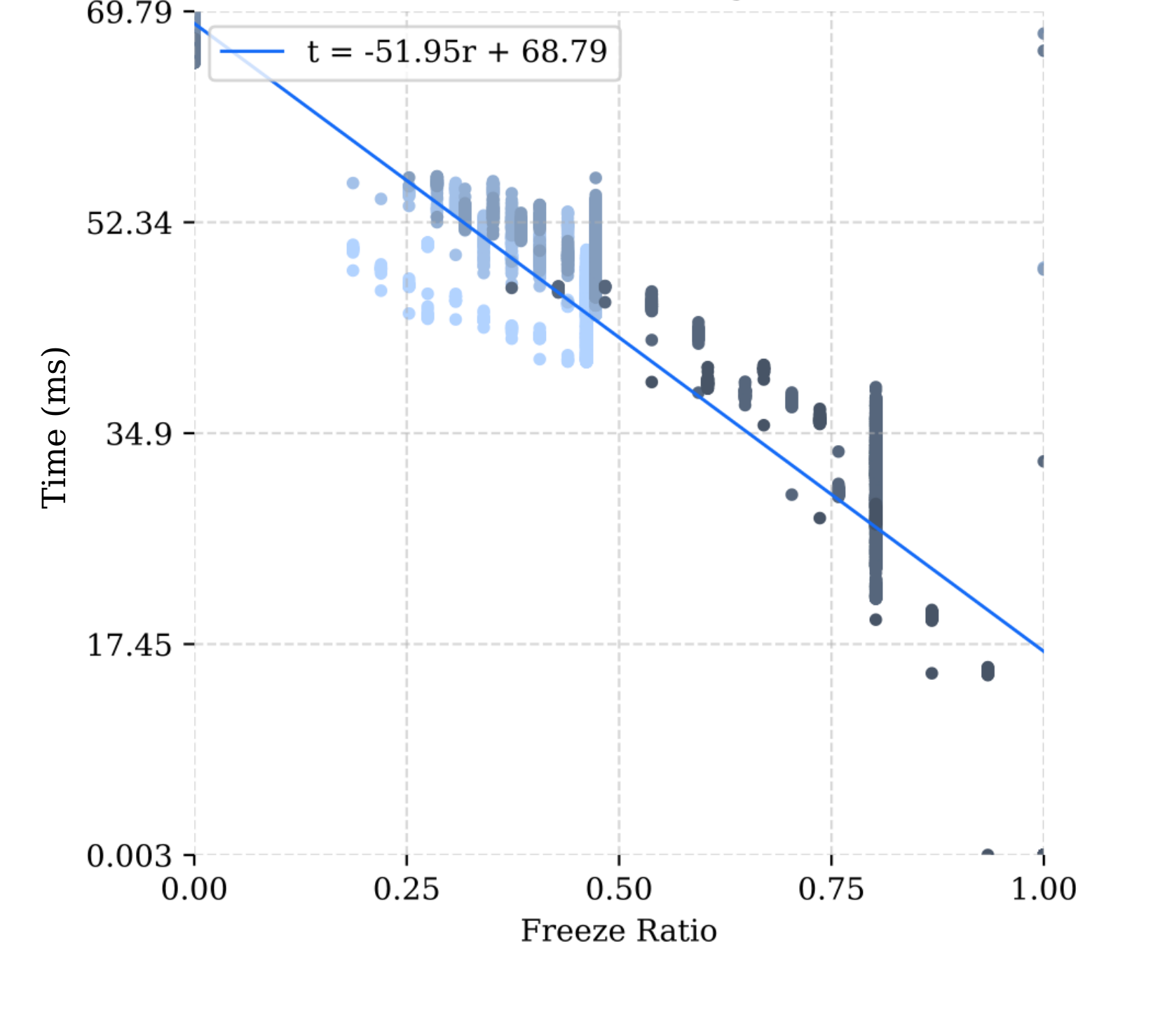}
    \subcaption{Stage 0 (GPU Rank 0).}
    \label{fig:bwd_stage0}
\end{subfigure}
&
\begin{subfigure}[t]{0.45\textwidth}
    \centering
    \includegraphics[
        width=\linewidth,
        trim=0 50pt 0 1pt,
        clip
    ]{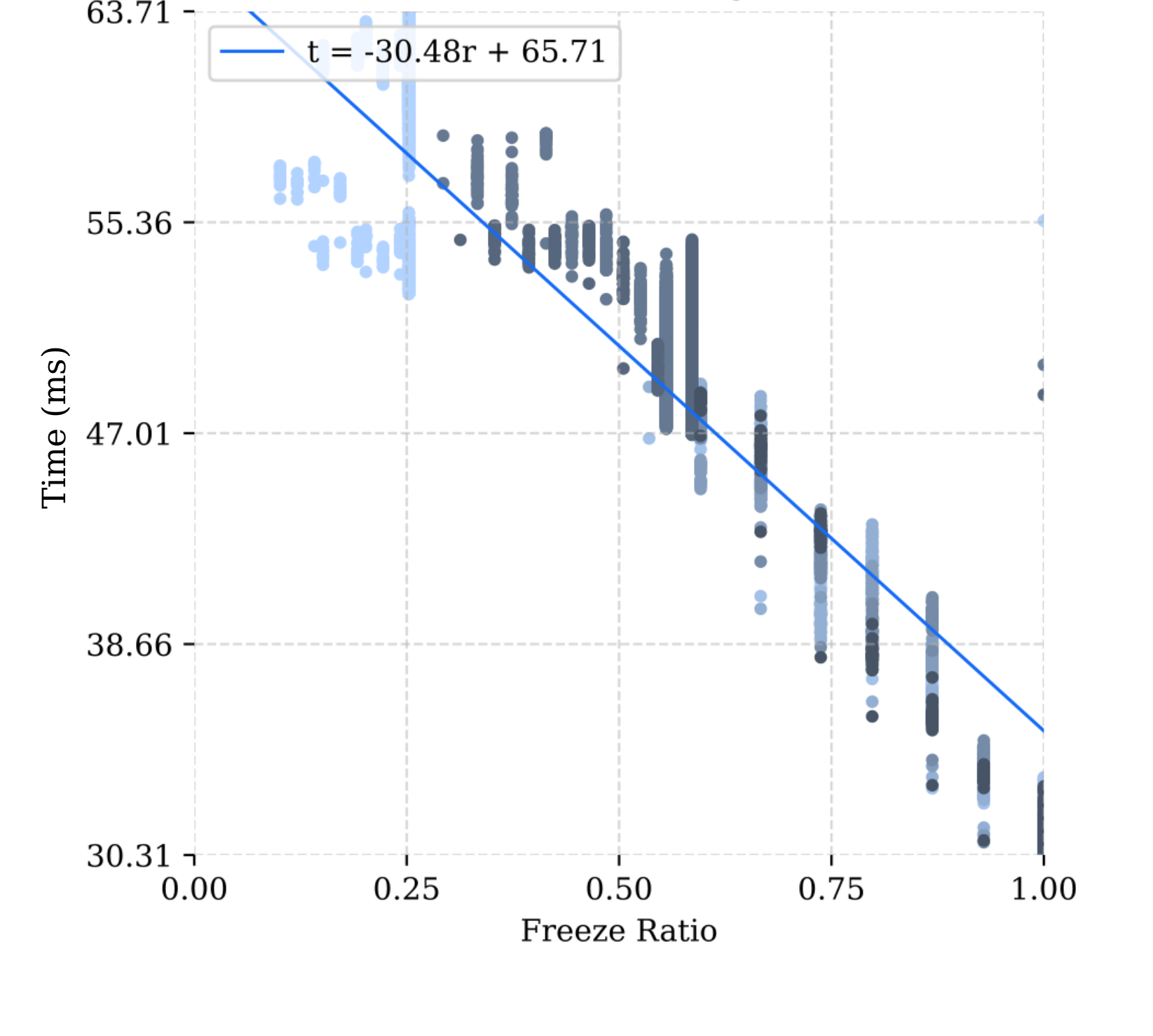}
    \subcaption{Stage 1 (GPU Rank 1).}
    \label{fig:bwd_stage1}
\end{subfigure}
\\[1.3cm]

\begin{subfigure}[t]{0.45\textwidth}
    \centering
    \includegraphics[
        width=\linewidth,
        trim=0 50pt 0 1pt,
        clip
    ]{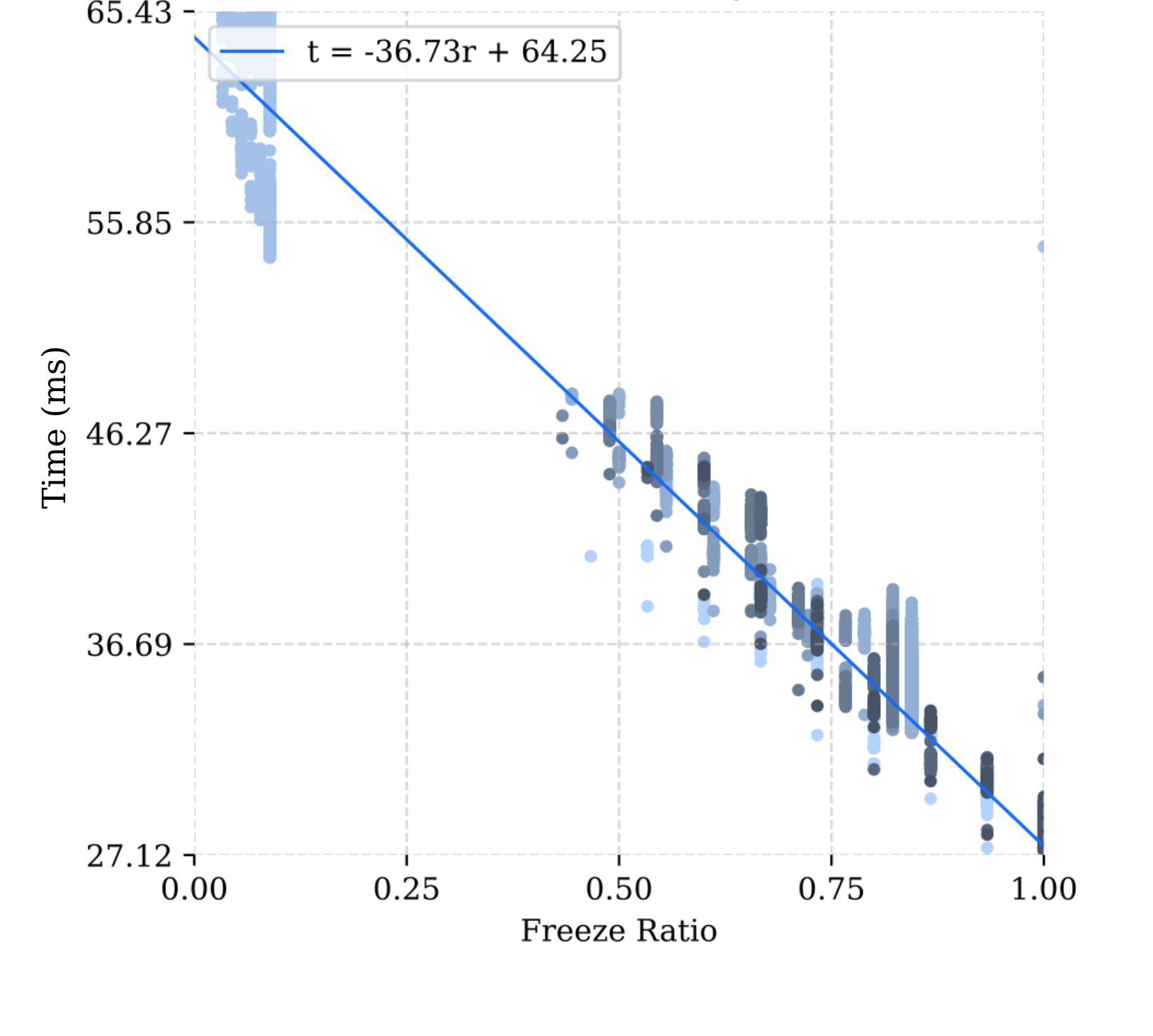}
    \subcaption{Stage 2 (GPU Rank 2).}
    \label{fig:bwd_stage2}
\end{subfigure}
&
\begin{subfigure}[t]{0.45\textwidth}
    \centering
    \includegraphics[
        width=\linewidth,
        trim=0 50pt 0 1pt,
        clip
    ]{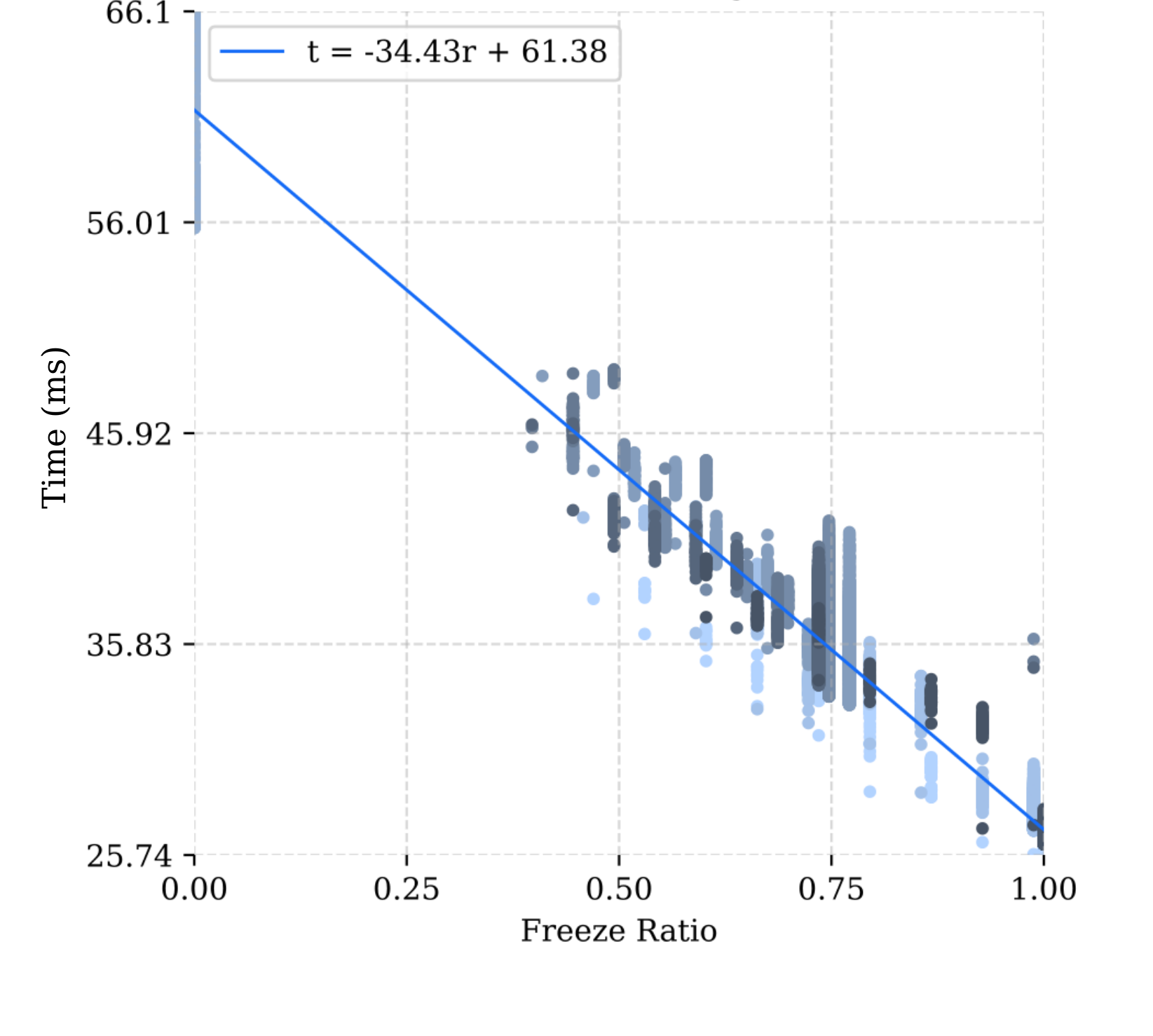}
    \subcaption{Stage 3 (GPU Rank 3).}
    \label{fig:bwd_stage3}
\end{subfigure}
\end{tabular}

\caption{
Backward computation time per effective freeze ratio across pipeline stages.
Each subplot corresponds to a model stage (GPU rank 0--3).
As the freeze ratio increases, the backward time decreases proportionally,
verifying that TimelyFreeze dynamically modulates per-stage workload as designed.
Blue markers indicate monitored microbatches, where color intensity reflects
the microbatch index (lighter for earlier microbatches and darker for later ones).
}
\label{fig:bwd_time_per_freeze_ratio}

\end{figure*}


\end{document}